\documentclass[leqno,11pt]{article}

\usepackage[english]{babel}
\usepackage[utf8]{inputenc} 

\DeclareUnicodeCharacter{03C6}{$\phi$}
\DeclareUnicodeCharacter{2082}{\textsubscript{2}}

\usepackage[T1]{fontenc} 
\usepackage{csquotes} 

\usepackage[
 backend=biber,
 maxnames=99,
 minnames=99
]{biblatex}
\renewbibmacro{in:}{}
\usepackage{anyfontsize} 
\usepackage[top=2.3cm,bottom=2.3cm,right=2.3cm,left=2.3cm]{geometry}
\usepackage{titling}
\usepackage{setspace} 
\usepackage{float} 

\usepackage{physics} 

\usepackage{mathtools} 
\usepackage{amsthm} 
\usepackage{amssymb} 
\usepackage{mathrsfs} 
\usepackage{bbm} 
\usepackage{array} 
\usepackage{upgreek} 

\usepackage{appendix}
\usepackage{chngcntr}

\usepackage[dvipsnames]{xcolor} 

\usepackage[pdfpagelabels,hypertexnames=false]{hyperref}

\definecolor{mylinkcolor}{named}{magenta}
\definecolor{myurlcolor}{named}{Blue}
\definecolor{mycitecolor}{named}{RoyalBlue}
\hypersetup{
colorlinks=true,citecolor=mycitecolor,linkcolor=mylinkcolor,urlcolor=Blue
}

\renewcommand{\eqref}[1]{%
 \textup{\textcolor{mylinkcolor}{\hyperref[#1]{(\ref*{#1})}}}%
}

\numberwithin{equation}{section}

\newcommand{\mytheoremname}{Theorem}
\newcommand{\mylemmaname}{Lemma}
\newcommand{\mypropname}{Proposition}
\newcommand{\mycorname}{Corollary}
\newcommand{\myexname}{Example}
\newcommand{\myremname}{Remark}
\newcommand{\mydefname}{Definition}

\newtheorem{theorem}{\protect\mytheoremname}[section]
\newtheorem{lemma}[theorem]{\protect\mylemmaname}
\newtheorem{prop}[theorem]{\protect\mypropname}
\newtheorem{cor}[theorem]{\protect\mycorname}
\theoremstyle{definition}
\newtheorem{example}[theorem]{\protect\myexname}
\newtheorem{remark}[theorem]{\protect\myremname}

\newcounter{hypo}
\newenvironment{hyp}{
 \begin{enumerate}
\setcounter{enumi}{\value{hypo}} \item}{\stepcounter{hypo} \end{enumerate}}

\newcommand{\egaldef}{=:}
\newcommand{\defegal}{:=}

\renewcommand{\tilde}{\widetilde}
\renewcommand{\hat}{\widehat}
\renewcommand{\bar}{\overline}
\renewcommand{\leq}{\leqslant}
\renewcommand{\geq}{\geqslant}

\newcommand\N{\mathbb{N}}
\newcommand\R{\mathbb{R}}
\newcommand\C{\mathbb{C}}
\newcommand\eps{\varepsilon}
\newcommand\rmi{\mathrm{i}}
\newcommand\e{\mathrm{e}}

\newcommand{\1}{\mathbbm{1}}
\newcommand{\ent}[1]{\left\lfloor #1 \right\rfloor}

\renewcommand\Tr{\mathrm{Tr}} 
\newcommand{\Vect}{\mathrm{Span}}
\newcommand{\Four}{\mathcal{F}} 

\newcommand{\ot}{\otimes}
\newcommand{\ots}{\otimes_{\mathrm{s}}}
\newcommand{\tensn}[2]{\ot^{#2} #1}
\newcommand{\tensymn}[2]{\ots^{#2} #1}

\newcommand{\symformsymb}{\mathsf{w}}
\newcommand{\symform}[2]{\symformsymb\left(#1,#2\right)}
\newcommand{\ad}[1][]{\mathrm{ad}_{#1}}

\newcommand{\Hi}{\mathcal{H}}
\newcommand{\hi}{\mathcal{Z}}
\newcommand{\normhi}[1]{\abs{#1}}
\newcommand{\conjug}{\mathfrak{c}}

\newcommand{\symop}[1][]{\mathcal{S}_{#1}}

\newcommand{\focks}{\Gamma_{\mathrm{s}}(\hi)}
\newcommand{\focksalg}{\Gamma_{\mathrm{s}}^{\mathrm{fin}}(\hi)}

\newcommand{\wickbigcore}{\mathscr{H}_0}
\newcommand{\wickcore}{\mathcal{D}_c}
\newcommand{\wickpol}[1][]{\mathcal{P}_{#1}(\hi)}
\newcommand{\wickhom}[1]{\mathcal{P}_{#1}(\hi)}
\newcommand{\wickleq}[1]{\mathcal{P}_{\leq #1}(\hi)}
\newcommand{\Wick}[2][]{\mathrm{Wick}_{#1}\left(#2\right)}
\newcommand{\wickdots}[1]{:#1:}
\newcommand{\compwick}[1][]{\ \sharp_{#1}\ } 

\newcommand{\ppde}{P(\phi)_2}
\newcommand{\FVclass}{\mathcal{K}(\hi)}

\newcommand{\creasymb}[1][]{\mathrm{a}^{*}_{#1}}
\newcommand{\crea}[2][]{\creasymb[#1]\left(#2\right)}
\newcommand{\annisymb}[1][]{\mathrm{a}_{#1}}
\newcommand{\anni}[2][]{\annisymb[#1]\left(#2\right)}
\newcommand{\fieldop}[2][]{\phi_{#1}\left(#2\right)}
\newcommand{\weylop}[2][]{\mathrm{W}_{#1}\left(#2\right)}
\newcommand{\nbop}[1][]{\mathrm{N}_{#1}}
\newcommand{\liftop}[1]{\Gamma\left(#1\right)}
\newcommand{\dG}[2][]{\dd \Gamma_{#1}\left(#2\right)}
\newcommand{\gat}[2]{\partial_z^{#1}#2}
\newcommand{\gatbar}[2]{\partial_{\bar{z}}^{#1}#2}
\newcommand{\gatgatbar}[3]{\partial_z^{#1}\partial_{\bar{z}}^{#2}#3}
\newcommand{\gatdotgatbar}[3]{\partial_z^{#1} #2 \cdot \partial_{\bar{z}}^{#1}#3}

\begin{document}

\title{\Large\textbf{Higher-Order Approximation of Coherent State Dynamics in Self-Interacting Quantum Field Theories\\ \vskip0.5cm}}

\makeatletter
\renewcommand{\@fnsymbol}[1]{%
 \ifcase#1\or
 $\,\sharp$\or
 $\,\natural$\or
 $\,\flat$\else
 \@arabic{#1}%
 \fi
}
\makeatother

\author{Zied~\textsc{Ammari}\thanks{LMB - UMR6623 CNRS, Université Marie \& Louis Pasteur. 16 route de Gray, 25030 Besançon Cedex, France. \href{mailto:zied.ammari@univ-fcomte.fr}{zied.ammari@umlp.fr}}\;, Julien~\textsc{Malartre}\thanks{Université Sorbonne Paris Nord, Laboratoire Analyse, Géométrie et Applications, LAGA,
CNRS, UMR 7539, F‐93430, Villetaneuse, France. \href{mailto:malartre@math.univ-paris13.fr}{malartre@math.univ-paris13.fr}}\;, Maher~\textsc{Zerzeri}\thanks{Université Sorbonne Paris Nord, Laboratoire Analyse, Géométrie et Applications, LAGA,
CNRS, UMR 7539, F‐93430, Villetaneuse, France. \href{mailto:zerzeri@math.univ-paris13.fr}{zerzeri@math.univ-paris13.fr}}}

\date{March 6, 2026}

\maketitle

\abstract{
\noindent We study the propagation of coherent states in self-interacting bosonic quantum field theories in the semi-classical (mean-field) regime. Relying on Hepp's method and a detailed analysis of the associated classical and quantum field dynamics, non-linear and linear respectively, we construct an asymptotic expansion of arbitrary order for the quantum evolution of coherent states. The results are first established for the spatially cutoff $\ppde$ model, under standard assumptions ensuring essential self-adjointness of the Hamiltonian and well-posedness of the classical flow, and are then extended to a class of non-polynomial analytic interactions. This work refines and generalizes earlier results, which identified only the leading-order term of the expansion.} 

	
 
 \tableofcontents
	
	\section*{Introduction}\label{sect.intro}
	
	Bohr's correspondence principle establishes a profound bridge between classical and quantum mechanics, asserting that the predictions of these two theories coincide in the limit of large quantum numbers. From a mathematical perspective, this principle leads naturally to the study of Schrödinger-type equations of the form 
	\begin{equation}\label{schrodinger_type_equation_intro}
	 \left\{\begin{array}{l}
	 \rmi\, \eps\, \partial_t \psi_t = H_\eps\, \psi_t\\
	 \eval{\psi_t}_{t=0} = \varphi_\eps\,,
	 \end{array}\right.
	\end{equation} 
	in the semi-classical regime $\eps \to 0$. This is precisely the realm of semi-classical analysis (see for instance \cite{DiSj99,Zw12}), a domain that provides rigorous formulations of Bohr’s principle, often referred to as the classical limit.
	
	\smallskip\noindent
	The connection between classical and quantum dynamics becomes especially striking when the initial datum $\varphi_\eps$ in \eqref{schrodinger_type_equation_intro} is chosen to minimize the uncertainty principle. Such states, known as coherent states, were first introduced by E.~Schrödinger in 1926 \cite{Sc26} as the most classically localized quantum states. The terminology \emph{coherent states} was later coined by J.~R.~Glauber in 1963 in the context of laser physics \cite{Gl63}. Since then, coherent states have developed into a central theme of both physics and mathematics, see for example \cite{CoRo21}. Coherent states continue to attract increasing attention owing to their involvement in various mathematical and physical fields such as condensed matter, quantum optics and quantum information theory \cite{BoPePiSo21,HaJo00,LeMiRaScSe21,LaRDStLeTz25}.
	
	\medskip
	
	A decisive advance was made in 1974, when K.~Hepp \cite{He74} introduced a rigorous framework clarifying the crucial role of coherent states in the classical limit. His groundbreaking method proved remarkably fruitful, inspiring a wide range of subsequent works devoted to semi-classical dynamics with a coherent state approach. In particular, it allows one to derive a full asymptotic expansion in powers of $\eps^{1/2}$ for solutions of the Schrödinger equation \eqref{schrodinger_type_equation_intro} with coherent state initial data \cite{Ro07}. 
	
	Hepp's method has since been extended to a variety of asymptotic regimes, including the mean-field limit, where the semi-classical parameter $\eps$ corresponds to the inverse particle number, as well as to infinite-dimensional phase space analysis and quantum field models (see respectively \cite{GiVe79,RoSc09,AmBr12}, \cite{AmNi08}, and \cite{Do81}). In the latter setting, it was shown in \cite{AmZe14} that under a certain class of self-interacting bosonic quantum field Hamiltonians, coherent states remain localized along classical trajectories, while undergoing a squeezing governed by a unitary Bogoliubov transformation. 
	
	The aim of the present paper is to further develop this line of research by providing a complete asymptotic description of coherent-state evolution in the classical limit, using Hepp's method in the spirit of D.~Robert's analysis for Schrödinger operators \cite{Ro07}. We focus in particular on the spatially cutoff $\ppde$ model, chosen for its fundamental role in theoretical physics (see \cite{GlJa85} and also \cite{DeGe00}). For this model, the classical field equation is given by a non-linear Klein-Gordon equation \eqref{NLKG_Pphi2} which is globally well-posed for initial data in the energy space, see Theorem \ref{thm.existence_solution_mild_globale_eq_classique_avec_partie_libre_Pphi2}. More generally, we study certain models of self-interacting bosonic quantum field Hamiltonians with analytic interactions, including a variant of the H\o egh-Krohn model \cite{HK69}.
	
	\medskip
	
	 Following the framework of \cite{AmZe14}, we adopt a dynamical point of view to study the classical limit. A different approach, based on a variational perspective, is discussed in \cite{Ai12, Ar96}. It is also worth mentioning that an alternative approach was developed in \cite{AmNi08}, based on an extension of Wigner (or semi-classical) measures to infinite-dimensional phase spaces. To date, this method has been applied mainly to many-body Hamiltonians in the mean-field regime, to quantum electrodynamics and to Yukawa and Polaron type models in different scaling (see e.g.~\cite{AmFa17,AmFaHi25,Fa25,CoFaOl23}). Its possible extension to self-interacting quantum field theory models lies beyond the scope of the present work and is left for future investigation.
	 
	 Using Hepp’s method and a systematic expansion of the quantum dynamics around the associated classical non-linear field equation, we construct an asymptotic expansion of arbitrary order for the time evolution of coherent states in the classical limit. The leading order is governed by the classical Hamiltonian flow, while subleading terms are described by a time-dependent quadratic (Bogoliubov) dynamics and higher-order Wick-ordered corrections. For each fixed time, we obtain norm-accurate expansions with explicit control of the remainder. These results refine and generalize previous work of the first and third authors \cite{AmZe14} as well as the work of K.~Hepp \cite{He74}, where only the leading-order term of the semi-classical expansion was established. 
	 
In particular, we clarify the study of the non-linear Klein-Gordon $\ppde$ field equation, by proving the existence and uniqueness of global mild solutions to \eqref{premiere_apparition_eq_classique_avec_partie_libre_Pphi2} for each initial data in the physical space $L^2(\R)$, see Corollary \ref{cor.existence_sol_mild_global_L2_Pphi2}. As a consequence, we show the propagation of coherent states along the classical trajectories departing from any initial data in $L^2(\R)$ and for arbitrary finite time below the Ehrenfest time, see Theorem \ref{thm.Pphi2_asymptotique_complete}. The above result requires an additional assumption on the spatial cutoff function $g$, see \eqref{premiere_apparition_spatially_cutoff_Pphi2_Hamiltonian}, which is taken furthermore to be in the Sobolev space $H^1(\R)$ instead of $L^2(\R)$. For further details, see Subsection \ref{subsect.The_classical_field_ equation_Pphi2} and Remark \ref{rem.existence_unicite_solution_H^s_Pphi2}. 

For the more general classical non-linear field equation \eqref{eq_classique_avec_partie_libre} with analytic non-linearity, formulated on an abstract Hilbert space $\hi$, we remark that the argument in \cite{AmZe14}, used to construct global mild solutions to \eqref{eq_classique_avec_partie_libre} on the $\hi$-space, is incorrect even though the conclusion itself may still hold. Nevertheless, we still have existence and uniqueness of local mild solutions to \eqref{eq_classique_avec_partie_libre} for any initial data in $\hi$, see Theorem \ref{thm.existence_solution_mild_eq_classique_avec_partie_libre}. In this setting, we accordingly prove the propagation of coherent states, with higher order approximation, along the classical trajectories during their lifespans and below the Ehrenfest time, see Theorem \ref{thm.general_asymptotique_complete}.

	\medskip\noindent
	\textbf{Overview of the paper.} In Section \ref{sect.prelim}, we briefly review the formalism of quantum field theory and Wick quantization, and we state our two main results: Theorem \ref{thm.general_asymptotique_complete} addresses a broad class of analytic interactions as described in \cite{AmZe14}, while Theorem \ref{thm.Pphi2_asymptotique_complete} provides more detailed and explicit results tailored specifically to the $\ppde$ model. Section \ref{sect.pphi2} begins with a short introduction to the $\ppde$ model and a global well-posedness result for its associated classical field equation. Next, we establish two Theorems \ref{thm.quadratic_dynamics_Ammari_Breteaux_chaospropag} and \ref{thm.Breteaux_conjug_Wick_propag_quad} on the dynamics of time-dependent quadratic Hamiltonians. The first result (Theorem \ref{thm.quadratic_dynamics_Ammari_Breteaux_chaospropag}) shows the existence and uniqueness of unitary propagators for time-dependent quadratic operators, it is a consequence of a general commutator method proved in \cite[Appendix C]{AmBr12}. The second Theorem \ref{thm.Breteaux_conjug_Wick_propag_quad} concerns the quantum quadratic evolution of Wick observable, which general form is due to S.~Breteaux \cite{Br12}. We complete the proof of Theorem \ref{thm.Pphi2_asymptotique_complete} in Subsection \ref{subsect.Proof_Thm_Pphi2}. Section \ref{sect.analytic_interactions} is devoted to the proof of Theorem \ref{thm.general_asymptotique_complete}. Since many arguments are similar to the case $\ppde$, we only highlight the innovative aspects. In Appendix \ref{app.weyl_op_nb_domain}, we establish number–type estimates for Weyl operators, which allows us to show that these operators preserve the domains of all powers of the number operator. Such information provides precise control over the approximation remainder. 
	

	\section{Background and main results}\label{sect.prelim}
	
	The symmetric Fock space is the natural Hilbert space where the Hamiltonians of free Bose fields are well described and expressed in a very explicit manner. Furthermore, we have at hand a wave representation of the Canonical Commutation Relations, see Theorem \ref{thm.rpz_ondulatoire_espace_Fock} below, which enables us to describe quantum field Hamiltonians either in particle representation or in wave representation. With the latter description, we will see that the interaction terms considered in this work take a particularly convenient form, which simplifies their study.
	
	\medskip
	
In order to make the paper as self-contained as possible, we introduce in the next subsection some basic notions of quantum field theory, following \cite{AmZe14}. 
	
\subsection{Formalism of Quantum Field Theory}\label{subsect.formalism_QFT}
	
	Let $\hi$ be a separable, complex Hilbert space which inner product $\braket{\cdot}{\cdot}$ is antilinear with respect to the left variable, with associated norm $\normhi{u} = \sqrt{\braket{u}{u}}$. We endow $\hi$ with a conjugation $\conjug : \hi \to \hi$, that is, a norm-preserving, antilinear involution on $\hi$. We define \begin{equation}\label{def_hilbert_invariant_conjugaison}
		\hi_0 = \{z \in \hi, \quad \conjug z = z\}\,,
	\end{equation}
	the real subspace of $\hi$.
	
	\bigskip
	
	The bosonic Fock space over $\hi$ is defined as the direct Hilbert sum 
	\begin{equation}\label{def_espace_fock_symetrique}
		\focks = \bigoplus_{n=0}^\infty \tensymn{\hi}{n}\,,
	\end{equation}
	where $\tensymn{\hi}{n}$ denotes the range of $\tensn{\hi}{n}$ by the orthogonal projection 
	\begin{equation}\label{def_op_symetrisation_rang_n}
		\symop[n] : u_1 \ot \cdots \ot u_n \longmapsto \frac{1}{n!} \sum_{\sigma \in \mathfrak{S}_n} u_{\sigma(1)} \ot \cdots \ot u_{\sigma(n)}\,,
	\end{equation}
	where $\mathfrak{S}_n$ is the symmetric group of degree 
$n$. The space $\tensymn{\hi}{n}$, seen as a subspace of $\focks$, is called the $n$-particle sector. For $n=0$, this gives the vacuum sector $\tensymn{\hi}{0} = \C\, \Omega$ where \begin{equation}\label{def_vacuum_state}
		\Omega = (1,0,0,\dots) \in \focks
	\end{equation}
	is called the vacuum vector. Thanks to these identifications, any $\psi = \left(\psi^{(n)}\right)_{n \in \N} \in \focks$ 	can be written 
	\begin{equation}\label{notation_serie_element_espace_fock}
		\psi = \bigoplus_{n = 0}^\infty \psi^{(n)}\quad \mathrm{with}\quad \big\Vert \psi\big\Vert_{\focks}^2=\sum_{n=0}^{+\infty}\Vert \psi^{(n)}\Vert_{\tensn{\hi}{n}}^2 < +\infty\,.
	\end{equation}
	
	\noindent The polarization formula 
	\begin{equation}\label{formule_polarisation_produit_tensoriel_symetrique}
		\symop[n](u_1 \ot \cdots \ot u_n) = \frac{1}{2^n\, n!} \sum_{\rho_1,\dots,\rho_n \in \{\pm 1\}} \rho_1 \cdots \rho_n \, \left[\tensn{\left(\sum_{j=1}^n \rho_j u_j\right)}{n}\right]
	\end{equation}
	yields, for all $n \in \N$, 
	\begin{equation}\label{espace_des_etats_totalement_symetriques_engendre_les_produits_tensoriels_identiques}
\bar{\Vect\{\tensn{w}{n}, \quad w \in \hi\}} = \tensymn{\hi}{n}\,.
	\end{equation}
	
	\medskip
The Hilbert space $\focks$ is endowed with the inner product 
	\[\braket{\varphi}{\psi}_{\focks} = \sum_{n \in \N}\, \braket{\varphi^{(n)}}{\psi^{(n)}}_{\tensymn{\hi}{n}}\,.
	\] 
	We denote by $\norm{\cdot}$ the associated norm.
	
	Finally, the finite particle space, defined as the algebraic direct sum 
	\begin{equation}\label{def_espace_fock_symetrique_algebrique}
		\focksalg = \bigoplus_{n \in \N}^{\mathrm{alg}} \tensymn{\hi}{n}\,,
	\end{equation}
	is dense in $\focks$.
	
	\medskip

	We now define the standard operators of quantum field theory. Let $\eps \in (0,1]$ be a semi-classical parameter. For all $u \in \hi$, the $\eps$-dependent creation and annihilation operators associated with $u$ are respectively defined on $\focksalg$ by the expressions 
	\begin{equation}\label{def_op_crea_tenseur_symetrique}
		\eval{\crea[\eps]{u}}_{\tensymn{\hi}{n}} = \sqrt{\eps(n+1)}\, \symop[n+1]\Big(\ket{u} \ot (\tensn{\1}{n})\Big)
	\end{equation}
	and 
	\begin{equation}\label{def_op_anni_tenseur_symetrique}
		\eval{\anni[\eps]{u}}_{\tensymn{\hi}{n}} = \sqrt{\eps n}\, \symop[n-1]\Big(\bra{u} \ot (\tensn{\1}{n-1}) \Big)\,,
	\end{equation}
	where $\bra{u} : v \in \hi \mapsto \braket{u}{v} \in \C$ and $\ket{u} = [\,\bra{u}\,]^{*} : \C \to \hi$. These two operators satisfy the canonical commutation relations (CCR): 
	\begin{equation}
		\comm{\anni[\eps]{u}}{\anni[\eps]{v}} =\comm{\crea[\eps]{u}}{\crea[\eps]{v}} = 0 \text{ and } \comm{\anni[\eps]{u}}{\crea[\eps]{v}} = \eps\, \braket{u}{v}\, \1, \, \forall\, u,v \in \hi\,.
	\end{equation}
	
	\medskip
	
	For all $u\in \hi$, the Segal field operators 
	\begin{equation}\label{def_op_champ}
		\fieldop[\eps]{u} = \frac{1}{\sqrt{2}}\, \Big(\crea[\eps]{u} + \anni[\eps]{u}\Big)
	\end{equation} are essentially self-adjoint on $\focksalg$, see for example \cite[Theorem X.41]{ReSi75}, which enables us to define (by Stone's Theorem) the Weyl operators
	\begin{equation}\label{def_op_Weyl}
		\weylop[\eps]{u} = \e^{\rmi \, \fieldop[\eps]{u}}\,.
	\end{equation}
	These operators satisfy, for all $u_1,u_2 \in \hi$, 
	\begin{equation}\label{relation_definissant_rpz_CCR}
		\weylop[\eps]{u_1} \weylop[\eps]{u_2} = \e^{\rmi\frac{\eps}{4} \symform{u_1}{u_2}}\, \weylop[\eps]{u_1+u_2}\,, 
	\end{equation}
	where $\symform{u_1}{u_2} = - 2\, \Im\braket{u_1}{u_2}$ is the canonical symplectic form.
	
	In this context, the set of coherent states are the total family of vectors in the Bosonic Fock space $\focks$ given by
	\begin{equation}\label{def_etats_coherents_QFT}
		\left\{\weylop[\eps]{-\rmi\frac{\sqrt{2}}{\eps}u}\Omega, \quad u \in \hi\right\},
	\end{equation}
	where $\Omega$ is the vacuum vector \eqref{def_vacuum_state}. Explicitly, we have 
	\begin{equation}
	 \weylop[\eps]{-\rmi\frac{\sqrt{2}}{\eps}u}\Omega=\e^{-\frac{1}{2\eps}\normhi{u}^2}\sum_{n=0}^{+\infty}\eps^{-\frac{n}{2}}\, \frac{\tensn{u}{n}}{\sqrt{n!}},\quad \forall\, u\in \hi\,.
	\end{equation}
	
	\medskip
	
	For any self-adjoint operator $A : D(A)\subset\hi \longrightarrow \hi$, the corresponding lifting operator is defined by 
	\[
	\eval{\liftop{A}}_{\tensymn{D(A)}{n,\rm alg}} = \tensn{A}{n}\,,
	\] 
	where $\tensymn{D(A)}{n,\rm alg}$ is the algebraic symmetric $n$-tensor product; and the second quantization of $A$ is the operator $\dG[\eps]{A}$ given by 
	\[
	\eval{\dG[\eps]{A}}_{\tensymn{D(A)}{n,\rm alg}} = \eps \sum_{j=1}^n\, (\tensn{\1}{j-1})\, \ot \underbrace{A}_{j^{\mathrm{th}} \mathrm{position}} \ot\, (\tensn{\1}{n-j})\,,
	\] 
which describes the Hamiltonian of free Bose fields in the particle representation. In particular, the $\eps$-dependent number operator is defined by $\nbop[\eps]=\dG[\eps]{\1}$. We also note that $\liftop{\conjug}$ defines a conjugation on the Bosonic Fock space $\focks$ whenever $\conjug$ is a conjugation on $\hi$.
	
\bigskip

We conclude this subsection with a crucial theorem, often referred to as the wave representation of Fock space, which arises from representation of Gaussian random processes indexed by real Hilbert spaces (see for instance \cite[Proposition 2.1.1]{GlJa85} and also \cite[Theorem I.1]{Si74}).

\begin{theorem}\label{thm.rpz_ondulatoire_espace_Fock}
There exists a probability space $(Q,\mathfrak{T},\mu)$ and a unitary map \[\mathcal{R}:\focks \to L^2(Q,\dd\mu)\]
fulfilling the following conditions:
		\begin{enumerate}
			\item[{\textbf{1.}}] $\mathcal{R} \Omega = 1$,
			\item[{\textbf{2.}}] $\mathcal{R} \mathfrak{M} \mathcal{R}^{*} = L^\infty(Q,\dd\mu)$, where $\mathfrak{M}$ denotes the $W^*$-algebra generated by the Weyl operators $\weylop[\eps]{u}$, $u \in \hi_0$.
		\end{enumerate}
		Furthermore, for all $V \in \focks$, we have \begin{equation}\label{Fock_conjugue_donne_L2_conjugue_via_rpz_ondulatoire}
			\mathcal{R}\liftop{\conjug}V = \bar{\mathcal{R}V}
		\end{equation}
		and, for all $u \in \hi_0$, $\eps > 0$, 
		\begin{equation}\label{formule_clef_pour_definition_Wick_non_poly}
			\mathcal{R}^{*} \Big(\mathcal{R}(\liftop{\sqrt{\eps}}V)\Big) \mathcal{R}(\weylop[\eps]{u} \Omega) = \weylop[\eps]{u} \liftop{\sqrt{\eps}} V\,.
		\end{equation}
	\end{theorem}
	
\noindent The wave representation of Fock space appears very useful to study the interaction terms considered in this work, as we will see that these terms correspond to unbounded multiplication operators by measurable functions in $L^2(Q,\dd\mu)$.
	
\subsection{Polynomial Wick quantization}\label{subsect.Wick_poly}
	
	In this subsection, we introduce the notion of polynomial symbol following \cite{AmNi08}.
	
	\medskip
	
	A map $b : \hi \to \C$ is said to be a monomial symbol of order $(p,q) \in \N^2$ if there exists a bounded operator $\tilde{b} \in \mathcal{L}(\tensymn{\hi}{p},\tensymn{\hi}{q})$ such that, for all $z \in \hi$, 
	\begin{equation}\label{identification_polynomes_Wick_operateurs_bornes}
		b(z) = \braket{\tensn{z}{q}}{\tilde b(\tensn{z}{p})}_{\tensymn{\hi}{q}}\,.
	\end{equation}
	In this case, we write $b \in \wickpol[p,q]$ and we call $m = p+q$ the total order of $b$. We define $\wickpol$ as the space of linear combinations of monomial symbols, called polynomial symbols. For all $m \in \N$, $\wickhom{m}$, resp. $\wickleq{m}$, denotes the space of polynomial symbols of total order $m$, resp. at most $m$.
	
	\medskip
	
	For all $b \in \wickpol[p,q]$, $z,h \in \hi$, we can write 
	\begin{equation}\label{dev_limite_derivees_Frechet_polynome_Wick}
		b(z+h) = \sum_{\substack{0 \leq j \leq q\\0 \leq k \leq p}}\, \frac{1}{j!}\, \frac{1}{k!}\, \braket{\tensn{h}{j}}{\gatgatbar{k}{j}{b}(z) (\tensn{h}{k})}\,,
	\end{equation}
	where 
	\begin{multline}\label{def_gatgatbar}
		\gatgatbar{k}{j}{b}(z) =\,\frac{p!}{(p-k)!}\, \frac{q!}{(q-j)!}\, \symop[j]\Big(\bra{\tensn{z}{q-j}} \ot (\tensn{\1}{j})\Big)\,\tilde{b}\, \symop[p]\Big(\ket{\tensn{z}{p-k}} \ot (\tensn{\1}{k})\Big)\\ \in \mathcal{L}(\tensymn{\hi}{k},\tensymn{\hi}{j})\,.
	\end{multline}
	This formula provides a notion of derivative for polynomial symbols. In particular, for all $k \in \N$, we have 
	\begin{equation}\label{def_gat}
		\gat{k}{b}(z) = \frac{p}{(p-k)!}\, \bra{\tensn{z}{q}}\, \tilde{b}\, \symop[p]\Big((\tensn{z}{p-k}) \ot (\tensn{\1}{k})\Big) \in \big(\tensymn{\hi}{k}\big)^*
	\end{equation}
	and 
	\begin{equation}\label{def_gatbar}
		\gatbar{k}{b}(z) = \frac{q!}{(q-k)!}\, \symop[k]\left(\bra{\tensn{z}{q-k}} \ot (\tensn{\1}{k})\right)\,\tilde{b}\,(\tensn{z}{p}) \in \tensymn{\hi}{k}\,. 
	\end{equation}
	
	\medskip
	
	The Wick quantization of a monomial symbol $b \in \wickpol[p,q]$ is defined on $\focksalg$ by the formulas \begin{equation}\label{def_quantification_Wick_poly_produit_tensoriel_symetrique}
		\eval{\Wick[\eps]{b}}_{\tensymn{\hi}{n}} = \1_{[p,+\infty)}(n)\, \frac{\sqrt{n!\,(n+q-p)!}}{(n-p)!}\,\eps^{\frac{p+q}{2}} \, \symop[n+q-p]\left(\tilde{b} \ot (\tensn{\1}{n-p})\right)\,,
	\end{equation}
	where $\symop[n+q-p]$ is the symmetrization operator given by \eqref{def_op_symetrisation_rang_n}. The expression $\Wick[\eps]{b}$ is extended by linearity to any polynomial symbol $b \in \wickpol$, and $b$ is called the Wick symbol of the operator $\Wick[\eps]{b}$. Let us remark the homogeneity property 
	\begin{equation}\label{prop.homogeneite_quantif_Wick}
	 \Wick[\eps]{b} = \eps^{\frac{j}{2}}\, \Wick[1]{b} \text{ for all } b \in \wickhom{j}\,.
	\end{equation}
	
	\medskip
	
	For all $\varphi, \psi \in \focksalg$ and for all polynomial symbol $b$, we have 
	\begin{equation}
	 \braket{\varphi}{\Wick[\eps]{b}\, \psi} = \braket{\Wick[\eps]{\bar{b}}\, \varphi}{\psi}\,,
	\end{equation}
	where $\bar{b}$ is the polynomial symbol defined by $\bar{b}(z) = \bar{b(z)}$ for all $z \in \C$. This identity shows that $\Wick[\eps]{\bar{b}}^{*}$ is an extension of $\Wick[\eps]{b}$: in particular, the operator $\Wick[\eps]{b}$ is closable. We keep the same notation for its closure.
	
	Identifying $b$ and $\tilde{b}$ through the identity \eqref{identification_polynomes_Wick_operateurs_bornes}, we give below the most important examples of Wick quantizations. \[\begin{array}{c|c}
		\tilde{b} & \Wick[\eps]{b}\\\hline
		\lambda \in \C & \lambda \1\\
		\bra{u} & \anni[\eps]{u}\\
		\ket{u} & \crea[\eps]{u}\\
		\sqrt{2} \, \Re\braket{u}{z} & \fieldop[\eps]{u}\\
		\1 \in \mathcal{L}(\hi) & \nbop[\eps]\\ 
		A \in \mathcal{L}(\hi) & \dG[\eps]{A}
	\end{array}\]
	
	\medskip
	
	We now state the main properties of Wick quantization, see \cite{AmNi08} for the proofs.
	
	\begin{prop}\label{prop.quantification_Wick_polynomiale}
		Let $b \in \wickpol[p,q]$.
		\begin{enumerate}
			\item (adjoint) On $\focksalg$, we have 
			\begin{equation}\label{adjoint_Wick_poly}
				\Wick[\eps]{b}^{*} = \Wick[\eps]{\bar{b}}\,.
			\end{equation}
			\item For any self-adjoint operator $A$ on $\hi$ and for all $t \in \R$, we have, on $\focksalg$, 
			\begin{equation}\label{relation_Wick_poly_second_quantifie}
				\e^{\rmi\frac{t}{\eps}\dG[\eps]{A}}\,\Wick[\eps]{b}\,\e^{-\rmi\frac{t}{\eps}\dG[\eps]{A}} = \Wick[\eps]{b(\e^{-\rmi t A}\,\bullet)}\,,
			\end{equation}
			where $b\left(\e^{-\rmi t A}\, \bullet\right) : z \in \C \mapsto b\left(\e^{-\rmi t A}\, z\right)$ is a polynomial symbol.
			As a consequence, for all $u \in \hi$, we have 
			\begin{equation}\label{relation_op_Weyl_second_quantifie}
				\e^{\rmi\frac{t}{\eps}\dG[\eps]{A}}\,\weylop[\eps]{u}\,\e^{-\rmi\frac{t}{\eps}\dG[\eps]{A}} = \weylop[\eps]{\e^{\rmi t A}\,u}\,.
			\end{equation}
			\item The space 
			\begin{equation}\label{def_bigcore_Wick_poly}
				\wickbigcore = \Vect\left\{\weylop[\eps]{u}\varphi, \quad u \in \hi,\, \varphi \in \focksalg\right\} 
			\end{equation}
			is a core for the operators $\Wick[\eps]{b}$, $b \in \wickpol$.
			\item (translation) On $\wickbigcore$, we have, for all $u \in \hi$, \begin{equation}\label{translation_Wick_poly}
				\weylop[\eps]{-\rmi\frac{\sqrt{2}}{\eps}u}^{*}\, \Wick[\eps]{b} \, \weylop[\eps]{-\rmi\frac{\sqrt{2}}{\eps}u} = \Wick[\eps]{b(\bullet + u)}\,,
			\end{equation}
			where $b\left(\bullet+u\right) : z \in \C \mapsto b(z+u)$ is a polynomial symbol.
			\item (composition) For all $b_1 \in \wickpol[p_1,q_1]$, $b_2 \in \wickpol[p_2,q_2]$, 
			\begin{equation}\label{formule_compo_Wick_poly}
				\Wick[\eps]{b_1} \Wick[\eps]{b_2} = \sum_{k = 0}^{\min(p_1,q_2)} \frac{\eps^k}{k!}\, \Wick[\eps]{\gatdotgatbar{k}{b_1}{b_2}}\,,
			\end{equation}
			where $\gatdotgatbar{k}{b_1}{b_2}(z) = \braket{\gat{k}{b_1}(z)}{\gatbar{k}{b_2}(z)}_{(\tensymn{\hi}{k})^*,\tensymn{\hi}{k}}$. We denote by $b_1 \compwick[\eps] b_2$ the Wick symbol of the operator $\Wick[\eps]{b_1}\Wick[\eps]{b_2}$.
		\end{enumerate}
	\end{prop}
	
\smallskip\noindent	
We end this subsection by stating the well-known number estimates, see \cite[Proposition~3.13]{DeGe00}.
	
	\begin{prop}
		For all $b \in \wickpol[p,q]$ and $\alpha,\beta \in \R$ such that $\alpha + \beta \geq p + q$, we have 
		\begin{equation}\label{number_estimate_Wick_poly}
			\norm{(\nbop[1]+1)^{-\frac{\alpha}{2}}\, \Wick[\eps]{b}\, (\nbop[1]+1)^{-\frac{\beta}{2}}} \leq\, \eps^{\frac{p+q}{2}}\, C(p,q,\alpha,\beta)\, \norm{\tilde{b}},
		\end{equation}
		where $C(p,q,\alpha,\beta)$ is a positive, $\eps$-independent constant. 
	\end{prop}
		
	\subsection{Main results}\label{subsect.Main_results}
	
	We start by stating our main result concerning the propagation of coherent states for spatially cutoff $\ppde$ Hamiltonians. We recall some standard facts about the construction of
	such models in Subsection \ref{subsect.intro_Pphi2_model}. Let us consider, for the one-particle Hilbert space $\hi=L^2(\R,\dd k)$, the spatially cutoff $\ppde$ Hamiltonian 
		\begin{equation}
		\label{premiere_apparition_spatially_cutoff_Pphi2_Hamiltonian}
			H_\eps = \dG[\eps]{\omega} + \int_\R g(x) \wickdots{P\left(\sqrt{2}\, \fieldop[\eps]{\frac{e^{-ikx}}{\sqrt{\omega(k)}}}\right)} \dd x\,,
		\end{equation} 
		where 
		\begin{itemize}
		 \item $\R \ni k\longmapsto \omega(k) = \sqrt{m_0^2 + k^2}$ is the dispersion relation, with $m_0 > 0$,
		 \item $P = \displaystyle\sum_{j=0}^{2n} \beta_j\, X^j$ is a real polynomial with even degree and positive dominant coefficient $\beta_{2n} > 0$,
		 \item $g \in L^1(\R) \cap H^1(\R)$ is a non-negative, even cutoff function.
		\end{itemize}
	
Here, for all $u\in \hi$, $\fieldop[\eps]{u}$ is the Segal field operator given by \eqref{def_op_champ} and 
\begin{equation}\label{ordredeWick}
	 \wickdots{\fieldop[\eps]{v}^j} = 2^{-\frac{j}{2}}\, \sum_{\ell = 0}^j\, \binom{j}{\ell}\, \crea[\eps]{v}^\ell\, \anni[\eps]{v}^{j-\ell}
	\end{equation}
is the Wick ordering of $\fieldop[\eps]{u}^j$. For more details on Wick ordering, one can consult \cite[Section~5]{DeGe00}.

The self-adjointness of the Hamiltonian \eqref{premiere_apparition_spatially_cutoff_Pphi2_Hamiltonian} is considered standard, see \cite{GlJa72,Se70,HKSi72}.

\medskip

We can now state our main result on the complete asymptotic evolution of coherent states in the classical limit $\eps\to 0$ for the $\ppde$ model. We use the convention 
\begin{equation}\label{Fourier_non_unit}
 \hat{g}(k) = \int_{\R} g(x)\, \e^{-\rmi\, k\, x} \dd x
\end{equation}
for the Fourier transform, sometimes denoted by $\Four$ for the sake of readability.

\bigskip\bigskip

\begin{theorem}\label{thm.Pphi2_asymptotique_complete}
Let $\varphi_0 \in \hi= L^2(\R,\dd k)$ be fixed. There exists a sequence $\Big(b_k(t,\cdot)\Big)_{k \in \N}\in \wickpol^\N$ of time-dependent polynomial symbols such that for all $\psi \in \displaystyle\bigcap_{\ell \in \N} D\left(\nbop[1]^{\ell/2}\right)$, for all $N \in \N$ and all $T > 0,$ there exists a finite $\eps$-independent bound $C(T,\psi,N) > 0$ such that for all $t \in [-T,T]$, the inequality
\begin{multline}\label{asymptotique_complete_CSdyn_Pphi2}
			 \norm{\e^{-\rmi\frac{t}{\eps}H_\eps}\, \weylop[\eps]{-\rmi\frac{\sqrt{2}}{\eps}\varphi_0}\psi-\sum_{k = 0}^N\, \eps^{\frac{k}{2}} \, \e^{\rmi\frac{\delta_\circ(t)}{\eps}}\, \weylop[\eps]{-\rmi\frac{\sqrt{2}}{\eps}\varphi_t}U_2^\circ(t,0)\Wick[1]{b_k(t)}\psi}\\ 
			 \leq C(T,\psi,N) \cdot \eps^{\frac{N+1}{2}}
		\end{multline}
holds uniformly in $\eps \in (0,1]$. Here, $\varphi_t \in C^0\Big(\R,L^2(\R)\Big)$ denotes the unique global mild solution, constructed in Corollary \ref{cor.existence_sol_mild_global_L2_Pphi2}, of the classical Klein-Gordon field equation 
		\begin{equation}\label{premiere_apparition_eq_classique_avec_partie_libre_Pphi2}
			\left\{\begin{array}{l}
				\rmi \, \dot{\varphi_t} = \omega\, \varphi_t + \mathscr{P}(\varphi_t)\\
				\eval{\varphi_t}_{t = 0} = \varphi_0
			\end{array}\right.
		\end{equation}
	such that
		\begin{multline*}
		 \mathscr{P}(\varphi_t) = \sum_{j=1}^{2n}\,j\, \beta_j\, \int_{\R^{j-1}} \prod_{m=1}^{j-1} \left(\frac{\bar{\varphi_t(k_m)} + \varphi_t(-k_m)}{\sqrt{\omega(k_m)}}\right)\,\times\\
		 \hat{g}(k_1 + \cdots + k_{j-1} + \bullet) \, \omega^{-1/2}(\bullet)\, \dd k_1 \cdots \dd k_{j-1} \in L^2(\R,\dd k)\,. 
		\end{multline*}
	 The function $\delta_\circ$ is given by 
		\begin{multline}\label{expression_explicite_torsion_Pphi2}
			\delta_\circ(t) =- \beta_0 \norm{g}_1\, t
			 + \sum_{j=1}^{2n}\, \frac{j-2}{2}\, \beta_j\, \times\\ 
			 \int_0^t \int_{\R^j} \hat{g}(k_1 + \cdots + k_j)\, \prod_{\ell=1}^j \left(\frac{\bar{\varphi_s(k_\ell)} + \varphi_s(-k_\ell)}{\sqrt{\omega(k_\ell)}}\right)\, \dd k_1 \cdots \dd k_j\, \dd s, 
		\end{multline}
		and $U_2^\circ(t,s)$ are time-dependent Bogoliubov transforms characterized as the unique unitary propagator of a time-dependent quadratic Hamiltonian given in Corollary \ref{cor.dynamique_quad_modele_Pphi2}. The polynomial symbols $b_k(t,\cdot) \in \wickleq{3k}$ are uniquely determined by the recursive formulas \eqref{expression_explicite_polynomes_de_Wick_de_l'asymptotique_complete}.
	\end{theorem}

		\begin{remark}\label{rem.Pphi2_asymptotique_complete}
		Let us highlight a few points:
			\begin{enumerate}
				\item In the spatial variable, the classical field equation \eqref{premiere_apparition_eq_classique_avec_partie_libre_Pphi2} corresponds to the non-linear Klein-Gordon equation 
				\begin{equation}\label{NLKG_Pphi2}
				 \left(\Box + m_0^2\right) \zeta_t = \mathscr{Q}(\zeta_t)\,,
				\end{equation}
				where $\Box = \partial_t^2 - \partial_x^2$ is the d'Alembert operator and 
			\begin{multline}\label{non-linearity_NLKG_Pphi2}
				 \mathscr{Q}(\zeta_t)(x) = - \frac{1}{2\, \pi} \sum_{j=1}^{2n} j\, \beta_j\, 2^{j/2}\, \times\\ \int_{\R^j} \left(\prod_{m=1}^{j-1} \hat{\zeta_t}(-k_m)\right)\, \hat{g}(k_1 + \cdots + k_j)\, \e^{\rmi\, k_j\, x}\, \dd k_1 \dots \dd k_j\,.
				\end{multline}
				\item In the case of a sub-quadratic polynomial interaction, that is, $n \leq 1$ where $2n$ is the degree of the polynomial $P$ in \eqref{premiere_apparition_spatially_cutoff_Pphi2_Hamiltonian}, we have the exact formula 	\begin{equation}\label{formule_exacte_propag_quad_etats_coherents}
	 \e^{-\rmi\frac{t}{\eps}H_\eps}\, \weylop[\eps]{-\rmi\frac{\sqrt{2}}{\eps}\varphi_0}\psi = \e^{\rmi\frac{\delta_\circ(t)}{\eps}}\, \weylop[\eps]{-\rmi\frac{\sqrt{2}}{\eps}\varphi_t}U_2^\circ(t,0)\,\psi\,,
	 \end{equation}
	 see Remark \ref{rem.subquad_case_Pphi2}.
				\item For each $N \in \N$, \eqref{asymptotique_complete_CSdyn_Pphi2} even holds for all $\psi \in D\left(\nbop[\eps]^{L/2}\right)$ with ${L = 3N+2n+3}$.
				\item The approximation of the coherent state dynamics \eqref{asymptotique_complete_CSdyn_Pphi2} is relevant for time $T$ less than the Ehrenfest time which is of order $\displaystyle\frac{N}{\lambda} \log(\displaystyle\frac{1}{\eps})$ with $\lambda$ a given Lyapunov exponent coming from Grönwall inequality. 
				\item Under the weaker assumption $g \in L^1(\R) \cap L^2(\R)$, existence and uniqueness of local $L^2$-valued mild solutions to \eqref{premiere_apparition_eq_classique_avec_partie_libre_Pphi2} still hold. In this setting, a local version of the above theorem applies; a more general formulation is given below in Theorem 
				 \ref{thm.general_asymptotique_complete}. 
			\end{enumerate}
	\end{remark}
	
		The interaction term of the Hamiltonian \eqref{premiere_apparition_spatially_cutoff_Pphi2_Hamiltonian} takes the form $\Wick[\eps]{F_V}$ with \[F_V(z) = \sum_{j=0}^{2n}\, \braket{\frac{\tensn{(z+\conjug z)}{j}}{\sqrt{j!}}}{V^{(j)}}\,,\] where $\conjug z(k) = \bar{z(-k)}$ defines a conjugation on $L^2(\R,\dd k)$ and $V = \big(V^{(j)}\big)_{j \in \N} \in \focksalg$ given by \eqref{interaction_Pphi2_de_la_forme_Fv}, see below. In fact, thanks to the formula \eqref{formule_clef_pour_definition_Wick_non_poly}, the polynomial Wick quantization defined by \eqref{def_quantification_Wick_poly_produit_tensoriel_symetrique} can be extended to symbols of the form \begin{equation}\label{def_potentiel_interaction_Fv}
		F_V(z) = \sum_{j = 0}^\infty\, \braket{\frac{\tensn{(z+\conjug z)}{j}}{\sqrt{j!}}}{V^{(j)}}\,,
	\end{equation}
	where $V = \Big(V^{(j)}\Big)_{j \in \N} \in \focks$ satisfies $\liftop{\conjug}V= V$. Furthermore, using the notations of Theorem \ref{thm.rpz_ondulatoire_espace_Fock}, the operator $\Wick[\eps]{F_V}$ corresponds via $\mathcal{R}$ to the multiplication operator by $\mathcal{R}\left(\liftop{\sqrt{\eps}}V\right)$ on $L^2(Q,\dd \mu)$. This enables us to extend Theorem \ref{thm.Pphi2_asymptotique_complete}.
	
	\medskip
	Let $A : D(A) \to \hi$ be a self-adjoint operator such that 
	
	\begin{hyp}\label{hyp.A1}
	$A \conjug = \conjug A \text{ and } A \geq m_0\, \1_\hi \text{ for some } m_0 > 0$.
	\end{hyp}

	\medskip
	Let $V \in \focks$ such that $\liftop{\conjug}V = V$. Following the setting of \cite[Theorem 2.16]{HKSi72}, we assume that $V$ satisfies:
	\begin{hyp}\label{hyp.A2}
		$\mathcal{R}(\liftop{\sqrt{\eps}}V) \in \displaystyle\bigcup_{q > 2} L^q(Q,\dd\mu)$ 
 		 and $\e^{-t \mathcal{R}(\liftop{\sqrt{\eps}}V)} \in L^1(Q,\dd\mu) \text{ for all } t > 0 \text{ and } \eps\in (0,1]$.
	\end{hyp}
	Here, $\mathcal{R}$ is the transform given by Theorem \ref{thm.rpz_ondulatoire_espace_Fock}.
	
	\medskip
	Under the assumptions \ref{hyp.A1} and \ref{hyp.A2}, the Hamiltonian 
	\begin{equation}\label{def_hamiltonien_asymptotique_complete_etat_coherent_QFT}
		H_\eps = \dG[\eps]{A} + \Wick[\eps]{F_V}
	\end{equation}
	is essentially self-adjoint on $\focksalg$ with domain $D(\dG[\eps]{A}) \cap D(\Wick[\eps]{F_V})$, see Theorem \ref{thm.hamiltonien_CSdyn_essentiellement_autoadjoint}. We recall that $\focksalg$ is given by \eqref{def_espace_fock_symetrique_algebrique}.
	
	\medskip
	To state our second main result on the complete asymptotic description of the evolution of coherent states in the classical limit for models of self-interacting bosonic quantum field Hamiltonians (i.e., cases of analytic interactions), we need the following additional assumption:
	\begin{hyp}
\label{hyp.dom_superexpo}
			$\exists\, \alpha > 0, \quad \exists\, \lambda > 1, \quad V \in D\left(\e^{\alpha\, \liftop{\lambda}}\right).$
		\end{hyp}

	\begin{theorem}\label{thm.general_asymptotique_complete}
		Assume that conditions \ref{hyp.A1}, \ref{hyp.A2} and \ref{hyp.dom_superexpo} are fulfilled.
		Let $\varphi_0 \in \hi$ be fixed. There exists a sequence $(b_k(t,\cdot))_{k \in \N} \in \wickpol^\N$ and $T_0 = T(\varphi_0) \in (0,+\infty]$ such that 
		for~all ${\psi \in \focksalg}$, for all $N \in \N$ and all $T\in (0,T_0),$ there exists a finite $\eps$-independent bound $C(T,\psi,N) > 0$ such that for all $t \in [-T,T]$, the inequality
		\begin{multline}\label{asymptotique_complete}
			\norm{\e^{-\rmi\frac{t}{\eps}H_\eps} \weylop[\eps]{-\rmi\frac{\sqrt{2}}{\eps}\varphi_0}\psi-\sum_{k = 0}^N \eps^{\frac{k}{2}} \, \e^{\rmi\frac{\delta(t)}{\eps}}\, \weylop[\eps]{-\rmi\frac{\sqrt{2}}{\eps}\varphi_t}\,U_2(t,0)\,\Wick[1]{b_k(t)}\psi}\\ 
			\leq C(T,\psi,N) \cdot \eps^{\frac{N+1}{2}}
		\end{multline}
		holds uniformly in $\eps \in (0,1/3]$. Here, $\varphi_t$ is the unique maximal mild solution on $(-T_{\rm min},T_{\rm max})$ of the classical field equation 
		\begin{equation}\label{eq_classique_avec_partie_libre}
			\left\{\begin{array}{l}
				\rmi \, \dot{\varphi_t} = A\, \varphi_t + \gatbar{}{F_V}(\varphi_t)\\
				\eval{\varphi_t}_{t = 0} = \varphi_0\,,
			\end{array}\right.
		\end{equation}
		given by Theorem \ref{thm.existence_solution_mild_eq_classique_avec_partie_libre} and $T_0=\min(T_{\min},T_{\max})$ denotes its minimal existence time. The function $\delta$ satisfies 
		\begin{equation}\label{conditions_fonction_torsion_propagation_etats_coherents}
		\left\{\begin{array}{l}
			\dot{\delta}(t) = \Re \braket{\varphi_t}{\gatbar{}{F_V}(\varphi_t)} - F_V(\varphi_t)\\
			\delta(0) = 0\,,
		\end{array}\right.
	\end{equation}
		and $U_2(t,s)$ are time-dependent Bogoliubov transforms characterized as the unique unitary propagator generated by a time-dependent quadratic Hamiltonian given in Theorem \ref{thm.dynamique_quad_applique_à_notre_potentiel}. The polynomial symbols $b_k(t,\cdot) \in \wickleq{3k}$ are uniquely determined by the recursive formulas \eqref{expression_explicite_polynomes_de_Wick_de_l'asymptotique_complete}.
	\end{theorem}
	
	\begin{remark}\label{rem.classical_action}
	 Explicitly, the function $\delta$ is given by 
	 \begin{equation}\label{expression_explicite_torsion}
			\delta(t) = \int_0^t \sum_{j=0}^\infty\, \frac{j-2}{2}\, \braket{\frac{\tensn{(\varphi_s + \conjug \varphi_s)}{j}}{\sqrt{j!}}}{V^{(j)}} \dd s\,.
		\end{equation}
		In the finite-dimensional case $\hi = \C^d$, this quantity is closely linked to the classical action. Indeed, combining \eqref{conditions_fonction_torsion_propagation_etats_coherents} with \eqref{eq_classique_avec_partie_libre}, we have formally \[\dot{\delta}(t) = \frac{1}{2}\, \symform{\varphi_t}{\dot{\varphi}_t} - h(\varphi_t)\,,\] 
		where $\symform{u_1}{u_2} = -2\, \Im\braket{u_1}{u_2}$ is the canonical symplectic form appearing in \eqref{relation_definissant_rpz_CCR} and \[h(z) = \braket{z}{A\, z} + F_V(z)\] is the energy functional associated with the Hamiltonian \eqref{def_hamiltonien_asymptotique_complete_etat_coherent_QFT}. Let us identify $\C^d$ with $\R^{2d}$ via the coordinates 
		\begin{equation}\label{identification_Cd_R2d}
		 z = \frac{q + \rmi\, p}{\sqrt{2}} \in \C^d \leftrightarrow (q,p) \in \R^{2d}\,.
		\end{equation}
		Identifying $h(z)$ with the classical Hamiltonian $h(q,p)$ and writing $\varphi_t = \displaystyle\frac{q_t + \rmi\, p_t}{\sqrt{2}}$ the solution of the classical field equation \eqref{eq_classique_avec_partie_libre}, we obtain \[\dot{\delta}(t) = L(q_t,\dot{q}_t,p_t) - \dv{t} \left(\frac{1}{2}\, q_t \cdot p_t \right)\,,\] where $L(q_t,\dot{q}_t,p_t) = \dot{q}_t \cdot p_t - h(q_t,p_t)$ is the Lagrangian of the system. Finally, we have 
		\begin{equation}\label{relation_torsion_action_classique}
		 \delta(t) = S_t(q,p) - \frac{1}{2} (q_t \cdot p_t - q_0 \cdot p_0)\,,
		\end{equation} where $S_t(q,p) = \displaystyle\int_0^t L(q_s,\dot{q}_s,p_s)\, \dd s$ is the classical action of the system. The expression \eqref{relation_torsion_action_classique} can be found, for example, in \cite[Chapter 4, page 94]{CoRo21}.
	\end{remark}
	
	\begin{remark} Let us highlight some implications obtained from the assumption \ref{hyp.dom_superexpo}:
	
	\begin{itemize}
	 \item Explicitly, the assumption \ref{hyp.dom_superexpo} means that the term defining the interaction satisfies 
 \[\sum_{j=0}^\infty\, \e^{2\alpha \lambda^j} \norm{V^{(j)}}^2 <\infty\,.\] 
	 \item The assumption \ref{hyp.dom_superexpo} implies $\mathcal{R}(V)\in \underset{q\geq 2}{\bigcap}L^q(Q,\dd\mu)$. Therefore, by hypercontractive estimates we get $\mathcal{R}\big(\liftop{\sqrt{\eps}}V\big)\in \underset{p\geq 2}{ \bigcap}L^p(Q,\dd\mu)$, see Theorem \ref{thm.hypercontractive_estimates_Simon} below.

	 \item The assumption \ref{hyp.dom_superexpo} implies that $V \in D\left(\liftop{\sqrt{2}}\right)$. Hence, if we define 
		\begin{equation}
			V_\ell^{(j)}(z_0) = \symop[\ell]\big(\bra{\tensn{(z_0 + \conjug z_0)}{j-\ell}} \ot (\tensn{\1}{\ell})\big)V^{(j)} \in \tensymn{\hi}{\ell}
		\end{equation}
		for all $z_0 \in \hi$, the sequence 
		\begin{equation}\label{def_potentiels_dev_Taylor_propagation_QFT}
			\big(V_\ell(z_0)\big)_{\ell \in \N} = \left(\sum_{j=\ell}^\infty\, \sqrt{\frac{\ell!}{j!}}\,\binom{j}{\ell}\, V_\ell^{(j)}(z_0)\right)_{\ell \in \N}
		\end{equation} 
		defines an element of $\focks$. This enables us to write, for all $z \in \hi$, 
		\begin{equation}\label{developpement_Fv_autour_d'un_point}
			F_V(z+z_0) = \sum_{\ell=0}^\infty \braket{\frac{\tensn{(z+\conjug z)}{\ell}}{\sqrt{\ell!}}}{V_\ell(z_0)} = \sum_{\ell=0}^\infty F_{V_\ell(z_0)}(z)\,.
		\end{equation}
		Thus, we have the expansion 
		\begin{equation}\label{developpement_Fv_autour_d'un_point_ordre_2}
			F_V(z+z_0) = F_V(z_0) + 2 \, \Re \braket{z}{\gatbar{}{F_V}(z_0)} + F_{V_2(z_0)}(z) + F_{R_3(z_0)}(z)\,,
		\end{equation}
		where $R_3(z_0) = \displaystyle\bigoplus_{\ell=3}^\infty V_\ell(z_0)$ and 
		\begin{equation}\label{def_gatbar_Fv}
			\gatbar{}{F_V}(z_0) = \sum_{j=1}^\infty\, \frac{j}{\sqrt{j!}}\, \left(\bra{\tensn{(z_0 + \conjug z_0)}{j-1}} \ot \1\right)\, V^{(j)} \in \hi\,.
		\end{equation}
	 \end{itemize}
	\end{remark}
	
	\bigskip
\begin{example}
 A particular case of analytic interaction is a variant of the H\o egh-Krohn model with spatial and momentum cutoffs and analytic interaction, see \cite[Subsection~3.4, Paragraph~2]{AmZe14}. On the one-particle space $\hi = L^2(\R^d,\dd k)$, let us consider the dispersion relation 
	\begin{equation}\label{relation_dispersion_modèle_Pphi2}
		\omega(k) = \sqrt{m_0^2 + \abs{k}^2} \text{ with } m_0 > 0\,.
	\end{equation}
	Let $u(k) = \sqrt{2}\, \omega(k)^{-1/2}$, and, for $\kappa \geq 1$ fixed, 
	\begin{equation}
		u_\kappa(k) = u(k) \, \hat{\chi}\left(\frac{k}{\kappa}\right)\,,
	\end{equation} where $\chi \in C_c^\infty(\R^d)$, $0 \leq \chi \leq 1$ is radial, of total mass 1 and supported in the unit ball. Here, $\hat{\chi}(k) = \displaystyle\int_{\R^d} \chi(x)\, \e^{- \rmi\, k \cdot x} \dd x$ denotes the Fourier transform of $\chi$. Finally, let $V(\lambda) = \displaystyle\sum_{j \in \N} a_j\, \lambda^j$, $a_j \in \R$, be an analytic function. For $r \geq 1$ fixed, the spatially cutoff H\o egh-Krohn model associated to $V$ corresponds to the Hamiltonian 
	\begin{equation}\label{def_spatially_cutoff_HK_Hamiltonian}
	 H_\eps^{\mathrm{HK}} = \dG[\eps]{\omega} + \sum_{j=0}^\infty a_j \int_{\abs{x} \leq r} \wickdots{\left(\fieldop[\eps]{\tau_x\, u_\kappa}\right)^j} \dd x\,,
	\end{equation}
	where $\tau_x u(k) = \e^{-\rmi kx}\, u(k)$ is the translation of vector $x$ in the Fourier variable and $\wickdots{\fieldop[\eps]{v}^j}$ denotes the Wick ordering of $\fieldop[\eps]{v}^j$, see \eqref{ordredeWick}. For this model, the interaction term takes the form $\Wick[\eps]{F_V}$ where 
	\begin{equation}
	\left\{\begin{array}{l}
	 V^{(j)} = \sqrt{j!}\, 2^{-\frac{j}{2}}\, a_j\, f_j\\
	 f_j(k_1,\dots,k_j) = \displaystyle\prod_{m=1}^j u_\kappa(k_m)\, \Four\left(\1_{\abs{x} \leq r}\right)(k_1 + \cdots + k_j) \in L^2_{\mathrm{sym}}(\R^{dj}) \simeq \tensymn{L^2(\R^d)}{j}\,.
	\end{array}\right.
	\end{equation}
	Let us note that we have the estimate 
	\begin{align*}
	 \norm{V^{(j)}}^2 \leq r^{2d}\, V_d^2 \, j! \left(\frac{(2 \pi \kappa)^d \norm{\chi}_2^2}{m_0}\right)^j a_j^2\,,
	 \end{align*}
	 where $V_d$ denotes the volume of the euclidean unit ball in $\R^d$. Hence, if the coefficients $a_j$ satisfy 
	 \begin{equation}
	 \sum_{j=0}^\infty \e^{2 \alpha' \lambda^j} \, j! \, a_j^2 < +\infty \text{ for some $\alpha'>0$ and $\lambda > 1$\,,}
	 \end{equation}
	 the assumption \ref{hyp.dom_superexpo} holds with constants $0 < \alpha < \alpha'$ and $\lambda > 1$.
\end{example}

	\noindent
	\textbf{Outline of the proof of Theorems \ref{thm.Pphi2_asymptotique_complete} and \ref{thm.general_asymptotique_complete}.}
	First, recall that the interaction term of the $\ppde$ model takes the form Wick$(F_V)$, where $V$ is given by \eqref{interaction_Pphi2_de_la_forme_Fv}. In this form, the quantities $\mathscr{P}$, $\delta_\circ(\bullet)$ and $U_2^\circ(t,s)$ involved in the $\ppde$ case are equal to
$\gatbar{}{F_V}$, $\delta(\bullet)$ and $U_2(t,s)$, respectively.
		
Expanding $F_V$ around the classical orbit $\varphi_t$ satisfying the classical field equation \eqref{eq_classique_avec_partie_libre}, we obtain an expression of the form 
		\begin{equation}\label{asymp_Taylor_outline_proof}
			F_V(\bullet+\varphi_t) = F_V(\varphi_t) + 2\, \Re\braket{\varphi_t}{\gatbar{}{F_V}(\varphi_t)} + F_{V_2(t)} + F_{R_3(t)}\,,
		\end{equation}
		where $V_\ell(t) = V_\ell(\varphi_t)$ (see \eqref{def_potentiels_dev_Taylor_propagation_QFT}) and $R_3(t) = \displaystyle\bigoplus_{\ell = 3}^\infty V_\ell(t)$. 
		
		Let $\psi \in \focksalg$ and $N \in \N$ be fixed. For $\delta \in C^1(\R,\R)$ satisfying \eqref{conditions_fonction_torsion_propagation_etats_coherents} (see also Remark~\ref{rem.classical_action}) and for $\big(b_j(t)\big)_{j \in \N} \in \wickpol^\N$ to be chosen later, let us define 
		\begin{equation*}
		\psi_N(t) = \sum_{j = 0}^N \eps^{\frac{j}{2}}\, \Wick[1]{b_j(t)}\psi
		\end{equation*}
and
		\begin{equation}\label{premiere_apparition_Theta(t)}
			\Theta_N(t) = \e^{\rmi\frac{t}{\eps}H_\eps}\, \e^{\rmi\frac{\delta(t)}{\eps}} \,\weylop[\eps]{-\rmi\frac{\sqrt{2}}{\eps}\varphi_t}U_2(t,0) \, \psi_N(t)\,,
		\end{equation}
		such that the left-hand side of \eqref{asymptotique_complete} is equal to $\norm{\Theta_N(t) - \Theta_N(0)}$.
		The formal differentiation of \eqref{premiere_apparition_Theta(t)} yields 
		\begin{equation}\label{expression_propre_derivee_Theta(t)}
			\dot{\Theta}_N(t) = \frac{\rmi}{\eps}\, \e^{\rmi\frac{t}{\eps}H_\eps}\, \e^{\rmi\frac{\delta(t)}{\eps}}\, \weylop[\eps]{-\rmi\frac{\sqrt{2}}{\eps}\varphi_t} \cdot \left(\Wick[\eps]{F_{R_3(t)}} \, U_2(t,0)\, \psi_N(t) - \rmi \, \eps\, U_2(t,0) \dot{\psi}_N(t)\right).
		\end{equation}
		Writing $R_3(t) = \displaystyle\bigoplus_{\ell = 3}^{N+2-j} V_{\ell}(t) + R_{N+3-j}(t)$ for each $0 \leq j \leq N$, applying the homogeneity property \eqref{prop.homogeneite_quantif_Wick} and picking suitable Wick symbols $b_j(t)$, we obtain 
		\begin{equation*}
		\dot{\Theta}_N(t) = \frac{\rmi}{\eps}\, \e^{\rmi\frac{t}{\eps}H_\eps}\, \e^{\rmi\frac{\delta(t)}{\eps}} \, \weylop[\eps]{-\rmi\frac{\sqrt{2}}{\eps}\varphi_t} \cdot \left(\sum_{j = 0}^N \eps^{\frac{j}{2}} \, \Wick[\eps]{F_{R_{N+3-j}(t)}} \, U_2(t,0) \, \Wick[1]{b_j(t)} \psi\right)\,,
		\end{equation*}
		 hence 
		\begin{align*}
			\norm{\Theta_N(t) - \Theta_N(0)} \leq &\int_0^{\abs{t}} \norm{\dot{\Theta}_N(s)} \dd s\\
			\leq &\frac{1}{\eps} \int_0^{\abs{t}} \sum_{j = 0}^N\, \eps^{\frac{j}{2}}\, \norm{\Wick[\eps]{F_{R_{N+3-j}(s)}}\, U_2(s,0)\, \Wick[1]{b_j(s)} \psi} \dd s\,.
		\end{align*}
		
		Finally, \eqref{inegalite_cruciale_pour} provides a domination of the form 
		\begin{equation*}
		 \norm{\Wick[\eps]{F_{R_{N+3-j}(s)}}\, U_2(s,0)\, \Wick[1]{b_j(s)} \psi} \leq\, C'(s,\psi,N,j) \cdot \eps^{\frac{N+3-j}{2}}\,,
		\end{equation*}
		which shows that the left hand-side of \eqref{asymptotique_complete} is dominated by 
		$\displaystyle
		 C(t,\psi,N) \cdot \eps^{\frac{N+1}{2}}
		$
		with $C(t,\psi,N)$ uniform on compact sets in $t$.
		
		\medskip
		
		Shortly speaking, the main technical challenge lies in controlling domain issues arising from the unboundedness of the various bosonic field operators and in providing a rigorous justification of the scheme described above.
		Let us now emphasize the main difference between the two proofs: the proof of Theorem \ref{thm.Pphi2_asymptotique_complete} mainly relies on the number estimates \eqref{number_estimate_Wick_poly}, which are not valid for non-polynomial interactions. Therefore, in order to prove Theorem \ref{thm.general_asymptotique_complete}, we must generalize these estimates by replacing the powers of $(\nbop[1]+1)$ with an analytic function in the variable $(\nbop[1]+1)$ with appropriately chosen coefficients having super-exponential decay rate. 
	
	\bigskip
	
	\section{The spatially cutoff \texorpdfstring{$\ppde$}{P(φ)₂} model}\label{sect.pphi2}
	
	This section is devoted to the $\ppde$ model. We start by recalling standard facts about the construction of the spatially cutoff $\ppde$ Hamiltonians in Subsection \ref{subsect.intro_Pphi2_model}. The study of the classical non-linear Klein-Gordon field equation associated with this model is presented in Subsection \ref{subsect.The_classical_field_ equation_Pphi2}. The existence and uniqueness of the unitary propagator describing the quadratic dynamics associated with the $\ppde$ model are proven in Subsection \ref{subsect.The_quadratic_dynamic_Pphi2}. We conclude the section with the proof of Theorem \ref{thm.Pphi2_asymptotique_complete} in three steps, see Paragraphs \ref{par.Diff_Pphi2}, \ref{par.Choice_Pphi2} and \ref{par.End_Proof_Thm_Pphi2} in Subsection \ref{subsect.Proof_Thm_Pphi2}.

	\subsection{Introduction of the model}\label{subsect.intro_Pphi2_model}
	
	We follow the construction of the spatially cutoff $\ppde$ model presented in \cite[Section~6.1]{DeGe00}. Let us consider the one-particle Hilbert space $\hi = L^2(\R,\dd k)$ endowed with its usual norm $\norm{\cdot}_2$ and with the conjugation 
	\begin{equation}\label{conjug_sur_L2_Fourier_pour_Pphi2}
		\conjug z (k) = \bar{z(-k)}\,.
	\end{equation}
	
The dispersion relation is given by 
\begin{equation}\label{relation_dispersion_modele_Pphi2}
		\R\ni k\longmapsto \omega(k) = \sqrt{m_0^2 + k^2}\quad \text{ with } m_0 > 0\,,
	\end{equation}
	and then the kinetic energy is defined by $\dG[\eps]{\omega}$.
Let us set $u(k)= \sqrt{2} \, {\omega(k)^{-1/2}}$, and let $P$ be a real polynomial given by 
	\begin{equation}\label{def_polynome_modele_Pphi2}
		P = \displaystyle\sum_{j=0}^{2n}\, \beta_j\, X^j \text{ with } \beta_{2n} > 0\,,
	\end{equation} 
which is bounded from below. The spatially cutoff $\ppde$ model corresponds to the Hamiltonian 
	\begin{align}\label{def_spatially_cutoff_Pphi2_Hamiltonian}
		H_\eps = &\dG[\eps]{\omega} + \int_\R g(x) \wickdots{P\left(\fieldop[\eps]{\tau_x u}\right)} \dd x\\
		= &\dG[\eps]{\omega} + \sum_{j=0}^{2n} \beta_j \int_\R g(x) \wickdots{\fieldop[\eps]{\tau_x u}^j} \dd x\nonumber\,,
	\end{align} 
	where $\tau_x u(k) = \e^{-\rmi kx}\, u(k)$ is the translation of vector $x$ in the Fourier variable, $g \in L^1(\R) \cap L^2(\R)$ is a non-negative, even cutoff function, and $\wickdots{\fieldop[\eps]{v}^j}$ denotes the Wick ordering of $\fieldop[\eps]{v}^j$, which explicit expression is given by \eqref{ordredeWick}.
	
\smallskip	
Since $u$ does not belong to $L^2(\R,\dd k)$, the formal notation 
	\begin{equation}\label{def_spatially_cutoff_interaction_Pphi2}
		 I_\infty = \int_\R g(x) \wickdots{P\big(\fieldop[\eps]{\tau_x u}\big)} \dd x
	\end{equation} expresses the fact that the interaction term $I_\infty$ is defined as the limit of a family of interactions \[\text{\enquote{$I_\infty = \displaystyle\underset{\kappa \to \infty}{\lim} I_\kappa$}\,.}\] Precisely, let $\chi$ be a Schwartz, non-negative function of total mass 1 and let us define, for all $\kappa \geq 1$, 
	\[
	u_\kappa(k) = u(k)\, \hat{\chi}\left(\displaystyle\frac{k}{\kappa}\right)\,,
	\] 
	where $\hat{\chi}$ denotes the Fourier transform of $\chi$. The interactions 
	\[I_\kappa = \int_\R g(x) \wickdots{P\big(\fieldop[\eps]{\tau_x u_\kappa}\big)} \dd x = \sum_{j=0}^{2n} \sum_{\ell = 0}^j\, 2^{-\frac{j}{2}}\, \beta_j\, \binom{j}{\ell}\, \underbrace{\int_\R g(x)\, \crea[\eps]{\tau_x u_\kappa}^\ell\, \anni[\eps]{\tau_x u_\kappa}^{j-\ell}\, \dd x}_{\egaldef I_{j,\ell,\kappa}}\] are well-defined on $\focksalg$. Furthermore, for all $0 \leq j \leq 2n$ and $0 \leq \ell \leq j$, $I_{j,\ell,\kappa}$ is a Wick operator with symbol 
	\begin{equation}\label{def_bjlkappa_Pphi2}
		b_{j,\ell,\kappa}(z) = \int_\R g(x)\, \braket{z}{\tau_x u_\kappa}^\ell\, \braket{\tau_x u_\kappa}{z}^{j-\ell}\, \dd x\,.
	\end{equation}
	Let us note that we have 
	\begin{equation}
		b_{j,\ell,\kappa}(z) = \braket{(\tensn{z}{\ell}) \ot (\tensn{\conjug z}{j-\ell})}{w_{j,\kappa}}_{L^2(\R,\dd k)} \,,
	\end{equation}
	where $w_{0,\kappa} = w_0 \defegal \norm{g}_1$ and, for $j \geq 1$,
	\begin{align}
		w_{j,\kappa}(k_1,\cdots,k_j) = \hat{g}(k_1 + \cdots + k_j) \prod_{m=1}^j u_\kappa(k_m)
	\end{align}
	is a Schwartz kernel belonging to $L^2(\R^j)$ according to \cite[Lemma 6.1]{DeGe00}. The same lemma tells us that, when $\kappa \to +\infty$, $(w_{j,\kappa})_{\kappa \geq 1}$ converges in $L^2(\R^j)$ towards 
	\begin{equation}\label{def_noyau_Schwartz_Pphi2_Derezinski_Gerard}
	 w_j(k_1,\dots,k_j) = \hat{g}(k_1 + \cdots + k_j) \prod_{m=1}^j u(k_m) \in L^2(\R^j)\,.
	\end{equation}
	
	\noindent Thus, if we define 
	\begin{equation}\label{def_bjl_Pphi2}
	 b_{j,\ell}(z) = \braket{(\tensn{z}{\ell}) \ot (\tensn{\conjug z}{j-\ell})}{w_j}_{L^2(\R^j)} \in \wickpol[j-\ell,\ell]\,,
	\end{equation}
	the interaction term \eqref{def_spatially_cutoff_interaction_Pphi2} can be rigorously defined as 
	\begin{equation}\label{def_rigoureuse_spatially_cutoff_interaction_Pphi2}
		I_\infty = \Wick[\eps]{b^\circ}\,, \text{ where } b^\circ = \sum_{j=0}^{2n} \sum_{\ell = 0}^j\, 2^{-\frac{j}{2}}\, \beta_j \, \binom{j}{\ell} \, b_{j,\ell} \in \wickpol\,.
	\end{equation}

	\begin{remark}
		In the wave representation given by Theorem \ref{thm.rpz_ondulatoire_espace_Fock}, each operator $\mathcal{R}\,I_\kappa\,\mathcal{R}^{*}$ is an unbounded multiplication operator on $L^2(Q,\dd \mu)$ associated with a function $G_\kappa \in L^2(Q,\dd \mu)$. In fact, we have 
		\[
		G_\kappa \in \displaystyle\bigcap_{1 \leq p < + \infty} L^p(Q,\dd \mu)\] and there exists a function $G \in \displaystyle\bigcap_{1 \leq p < + \infty} L^p(Q,\dd \mu)$ such that \[\norm{G_\kappa - G}_{L^p(Q,\dd \mu)} \xrightarrow[\kappa \to \infty]{} 0 \text{ for all } 1 \leq p < + \infty,
		\] 
		see \cite[Lemma 3.16]{HKSi72}. The multiplication operator associated with $G$ is $\mathcal{R}\, I_\infty\, \mathcal{R}^{*}$.
	\end{remark}
	
	\begin{theorem}[{\cite[Theorem 3.15]{HKSi72}}]
		The Hamiltonian $H_\eps = \dG[\eps]{\omega} + I_\infty$ with $I_\infty$ given by \eqref{def_spatially_cutoff_interaction_Pphi2} is essentially self-adjoint on $\focksalg$ with domain $D(H_\eps) = D(\dG[\eps]{\omega}) \cap D(I_\infty)$.
	\end{theorem}
	
	\subsection[The classical field equation]{The classical field equation for the $\ppde$ model}\label{subsect.The_classical_field_ equation_Pphi2}
	
	Formally, the Hamiltonian \eqref{def_spatially_cutoff_Pphi2_Hamiltonian} takes the form $H_\eps = \Wick[\eps]{h}$, where 
	\begin{equation}\label{def_energy_functional_Pphi2_model}
		h(z) = \braket{z}{\omega\, z} + \sum_{j=0}^{2n} \sum_{\ell = 0}^j\, 2^{-\frac{j}{2}}\, \beta_j\, \binom{j}{\ell}\, b_{j,\ell}(z) \egaldef \norm{\omega^{1/2}\, z}_2^2 + b^\circ(z)\,,
	\end{equation} which provides the classical field equation \begin{equation}\label{eq_classique_avec_partie_libre_Pphi2}
		\left\{\begin{array}{l}
			\rmi \, \dot{\varphi}_t = \omega\, \varphi_t + \gatbar{}{b^\circ}(\varphi_t)\\
			\eval{\varphi_t}_{t=0} = \varphi_0\,.
		\end{array}\right.
	\end{equation}
	
	A map $\varphi_t \in C^0\big(I,L^2(\R,\dd k)\big)$ is said to be a mild solution of \eqref{eq_classique_avec_partie_libre_Pphi2} on $I$ if it satisfies, for all $t \in I$, the integral equation 
	\begin{equation}\label{eq_classique_avec_partie_libre_formulation_integrale_Pphi2}
		\varphi_t = \e^{-\rmi t \omega} \, \varphi_0 - \rmi \, \int_0^t \e^{-\rmi(t-\tau)\omega}\, \gatbar{}{b^\circ}(\varphi_\tau)\, \dd \tau\,.
	\end{equation}
	
	For $s\geq 0$, let us define $\Hi^{s}_\omega = D\left(\omega^{s}\right)$ equipped with its graph norm 
	\[
	\norm{\cdot}_{\Hi^{s}_\omega}^2 = \norm{\omega^{s}\ \cdot}_2^2 + \norm{\cdot}_2^2\,.
	\]
	The aim of this subsection is to show the existence of a unique global mild solution for the Cauchy problem \eqref{eq_classique_avec_partie_libre_Pphi2} with any initial data $\varphi_0\in\hi = L^2(\R,\dd k)$. First, we show the existence of a unique maximal mild solution in $\Hi^{1/2}_\omega$ by passing to the interaction representation, see for instance \cite[Section X.12, page 283]{ReSi75}. This method consists in performing a time-dependent change of unknown, obtained by conjugating $\varphi_t$ with the propagator $\e^{\rmi t\omega}$, that removes the free part $\omega$ in \eqref{eq_classique_avec_partie_libre_Pphi2} so that only the transformed interaction drives the evolution, see \eqref{eq_classique_sans_partie_libre_Pphi2}.
	\begin{prop}\label{prop.estimee_champ_vecteurs_Pphi2}
		Let $g\in L^1(\R) \cap L^2(\R)$. Then, the vector field \[\begin{array}{ccccc}
			X & : & \R \times \Hi^{1/2}_\omega & \longrightarrow & \Hi^{1/2}_\omega \\
			& {} & (t,z) & \longmapsto & \e^{\rmi t \omega}\,\gatbar{}{b^\circ}\big(\e^{-\rmi t \omega}\, z)
		\end{array}\] is well-defined and continuous. Furthermore, for all $t \in \R$ and $z_1,z_2 \in \Hi^{1/2}_\omega$, we have 
		\begin{equation}\label{estimee_champ_vecteurs_Pphi2_espace_energie}
			\norm{X(t,z_1)-X(t,z_2)}_{\Hi^{1/2}_\omega} \leq C\Big(\max\big(\norm{z_1}_{\Hi^{1/2}_\omega},\norm{z_2}_{\Hi^{1/2}_\omega}\big)\Big) \norm{z_1-z_2}_{\Hi^{1/2}_\omega}\,,
		\end{equation}
		where $C$ is a polynomial function.
	\end{prop}
	
	The following statement is a consequence of the Cauchy-Lipschitz theorem.
	\begin{cor}\label{cor.existence_unicite_solution_H^s_Pphi2}
Let $g\in L^1(\R) \cap L^2(\R)$. Then, for any initial data $\varphi_0 \in \Hi^{1/2}_\omega$, the Cauchy problem \begin{equation}\label{eq_classique_sans_partie_libre_Pphi2}
			\left\{\begin{array}{l}
				\rmi \, \dot{\tilde\varphi}_t = \e^{\rmi t \omega}\,\gatbar{}{b^\circ}\big(\e^{-\rmi t \omega}\,\tilde\varphi_t\big)\\	
				\eval{\tilde\varphi_t}_{t=0} = \varphi_0
			\end{array}\right.
		\end{equation} admits a unique maximal solution $\tilde{\varphi}_t \in C^1\left(I,\Hi^{1/2}_\omega\right)$, where $I$ is an open interval containing the origin. Moreover, the map 
		\[\varphi_t = \e^{-\rmi t \omega}\, \tilde{\varphi}_t \in C^0\left(I,\Hi^{1/2}_\omega\right)\]
		is the unique maximal, $\Hi^{1/2}_\omega$-valued mild solution of \eqref{eq_classique_avec_partie_libre_Pphi2}, satisfying \eqref{eq_classique_avec_partie_libre_formulation_integrale_Pphi2}.
	\end{cor}

	\begin{proof}[Proof of the estimate \eqref{estimee_champ_vecteurs_Pphi2_espace_energie}]
		It suffices to estimate, for all $j \in \N$ and $1 \leq \ell \leq j$, the quantities \[
		\norm{\gatbar{}{b_{j,\ell}}(z_1) - \gatbar{}{b_{j,\ell}}(z_2)}_2 \text{ for } z_1,z_2 \in L^2(\R,\dd k)\] and \[\norm{\omega^{1/2}\Big(\gatbar{}{b_{j,\ell}}(z_1) - \gatbar{}{b_{j,\ell}}(z_2)\Big)}_2 \text{ for } z_1,z_2 \in \Hi^{1/2}_\omega\,,
		\]
		where $b_{j,\ell} \in \wickpol[j-\ell,\ell]$ is given by \eqref{def_bjl_Pphi2}.
		\smallskip
		\begin{itemize}
			\item The expression \eqref{def_gatbar} provides \[\gatbar{}{b_{j,\ell}}(z) = \ell\, \left(\bra{\tensn{z}{\ell-1}} \ot \1\right)\, \tilde{b}_{j,\ell}\,(\tensn{z}{j-\ell})\,,\] whence 
 \begin{multline*}
 \norm{\gatbar{}{b_{j,\ell}}(z_1)-\gatbar{}{b_{j,\ell}}(z_2)}_2 \leq \,\ell \, \norm{\left(\bra{\tensn{z_1}{\ell-1}-\tensn{z_2}{\ell-1}} \ot \1\right)\, \tilde{b}_{j,\ell}\,(\tensn{z_1}{j-\ell})}_2 \\+ \ell\, \norm{\left(\bra{\tensn{z_2}{\ell-1}} \ot \1\right)\, \tilde{b}_{j,\ell}\,(\tensn{z_1}{j-\ell}-\tensn{z_2}{j-\ell})}_2\,.
 \end{multline*}
			Furthermore, the identity 
			\begin{equation}\label{identite_remarquable_a^n_moins_b^n_version_produit_tensoriel}
				\tensn{z_1}{k} - \tensn{z_2}{k} = \sum_{\ell=1}^{k} \symop[k]\left((z_1 - z_2) \ot (\tensn{z_1}{k-\ell}) \ot (\tensn{z_2}{\ell-1})\right) \text{ for all } k \in \N^*
			\end{equation}
			leads to 
			\begin{equation}\label{majoration_difference_puissances_produit_tensoriel}
			 \norm{\tensn{z_1}{k} - \tensn{z_2}{k}}_2 \leq k\, \max\left(\norm{z_1}_2,\norm{z_2}_2\right)^{k-1}\, \norm{z_1-z_2}_2\,.
			\end{equation}
			Finally, we obtain 
			\begin{align*}
				\norm{\gatbar{}{b_{j,\ell}}(z_1)-\gatbar{}{b_{j,\ell}}(z_2)}_2 \leq\, &\ell \norm{\tilde{b}_{j,\ell}} \times\\
				&\Big((\ell - 1) \max\left(\norm{z_1}_2,\norm{z_2}_2\right)^{\ell-2}\norm{z_1}_2^{j-\ell}\\
				& \phantom{\Big(}+ (j-\ell)\max\left(\norm{z_1}_2,\norm{z_2}_2\right)^{j-\ell-1} \norm{z_2}_2^{\ell-1}\Big) \cdot \norm{z_1-z_2}_2\,,
			\end{align*}
			which, summing over $j$ and $\ell$, yields an estimate of the form 
			\begin{equation}\label{estimee_L2_champ_vecteurs_Pphi2_espace_energie}
			\norm{X(t,z_1)-X(t,z_2)}_2 \leq C\Big(\max\big(\norm{z_1}_2,\norm{z_2}_2\big)\Big) \norm{z_1-z_2}_2\,.
		\end{equation}

			\item For all $z \in \Hi^{1/2}_\omega$, let $f_z(x) = \braket{\omega^{1/2}\,z}{\omega^{-1/2}\,\tau_x u}$. The symbol $b_{j,\ell}$ defined by \eqref{def_bjl_Pphi2} satisfies 
		\begin{align*}
			 \gatbar{}{b_{j,\ell}}(z) = &\ell \int_\R g(x)\, f_z(x)^{\ell-1}\, \bar{f_z(x)}^{\,j-\ell}\, \tau_x u_\kappa \, \dd x \in L^2(\R, \dd k)\\
			 = &\ell \, u \, \Four\left(g \, f_z^{\ell-1} \, \bar{f_z}^{j-\ell}\right)\,,
			\end{align*}
			where we recall that $\Four$ denotes the Fourier transform defined by \eqref{Fourier_non_unit}. Thus, the estimate 
		\begin{equation}\label{majoration_norme_infty_pour_champ_vecteurs_Pphi2}
				\norm{f_z}_\infty \leq \norm{z}_{\Hi^{1/2}_\omega} \norm{\omega^{-1/2}\, u}_2
			\end{equation} 
			yields 
			\begin{align*}
				\norm{\omega^{1/2}\, \gatbar{}{b_{j,\ell}}(z)}_2 \leq &\, \ell\, \norm{\omega^{1/2}\, u}_\infty \, \norm{\Four\left(g \, f_z^{\ell-1} \, \bar{f_z}^{j-\ell}\right)}_2\\
				\leq &2\,\sqrt{\pi} \, \ell \, \norm{g \, f_z^{\ell-1} \, \bar{f_z}^{j-\ell}}_2\\
				\leq &2\,\sqrt{\pi} \, \ell \, \norm{g}_2 \norm{f_z}_\infty^{j-1}\\
				\leq &2\,\sqrt{\pi} \, \ell \, \norm{g}_2\, \norm{z}_{\Hi^{1/2}_\omega}^{j-1} \norm{\omega^{-1/2}\, u}_2^{j-1} < +\infty\,.
			\end{align*}
			Now, for all $z_1, z_2 \in \Hi^{1/2}_\omega$, $\Big(\gatbar{}{b_{j,\ell}}(z_1) - \gatbar{}{b_{j,\ell}}(z_2)\Big)$ is a linear combination of terms of the form 
			\[
			\ell \int_\R g(x)\, f_{z_1-z_2}(x)^\natural\, \left(f_{z_1}(x)^\natural\right)^m \, \left(f_{z_2}(x)^\natural\right)^n \tau_x\, u\, \dd x\,,
			\] 
			where $m,n \in \N$ and $f_z(x)^\natural \in \Big\{f_z(x), \bar{f_z(x)}\Big\}$, which gives an estimate of the form 
			\[
			\norm{\omega^{1/2}\, \gatbar{}{b_{j,\ell}}(z_1) - \omega^{1/2}\, \gatbar{}{b_{j,\ell}}(z_2)}_2 \leq C_{j,\ell}\Big(\norm{z_1}_{\Hi^{1/2}_\omega}, \norm{z_2}_{\Hi^{1/2}_\omega}\Big) \norm{g}_2 \norm{z_1 - z_2}_{\Hi^{1/2}_\omega}\,.
			\]
		\end{itemize}
	\end{proof}
	
		\begin{remark}\label{rem.blowup_criterion_Pphi2}
		 According to \eqref{estimee_champ_vecteurs_Pphi2_espace_energie}, the vector field $X$ is continuous and locally bounded on $\Hi^{1/2}_\omega$. Hence, we have at hand the following blowup criterion, see \cite[Theorem 2, page 9]{Re76}: if $T \defegal \sup I$ is finite, then the map $t \longmapsto \norm{\varphi_t}_{\Hi^{1/2}_\omega}$ is not bounded in the neighborhood of $T$.
	\end{remark}
	
	\begin{lemma}
		Let	$g\in L^1(\R) \cap L^2(\R)$. Then, there exists $c \geq 0$ such that, for all $z \in \Hi^{1/2}_\omega$, 
		\begin{equation}\label{energie_Pphi2_controle_norme_espace_energie}
			h(z) \geq \norm{\omega^{1/2}\,z}_2^2 - c\,.
		\end{equation}
	\end{lemma}
	
	\begin{proof}
		Let us recall the the polynomial $P$ defined by \eqref{def_polynome_modele_Pphi2} has even degree $2n$ and positive dominant coefficient. Hence, $P$ is bounded from below. Using the expressions \eqref{def_bjlkappa_Pphi2} and \eqref{def_rigoureuse_spatially_cutoff_interaction_Pphi2} of $b_{j,\ell,\kappa}(z)$ and $b^\circ(z)$ respectively, we obtain, for all $z \in L^2(\R,\dd k)$, 
		\begin{align*}
			b^\circ(z) = &\underset{\kappa \to \infty}{\lim}\ \sum_{j=0}^{2n}\, 2^{-\frac{j}{2}}\, \beta_j\, \sum_{\ell=0}^j\, \binom{j}{\ell}\, b_{j,\ell,\kappa}(z)\\
			= &\underset{\kappa \to \infty}{\lim}\ \int_\R g(x)\, P\Big(\sqrt{2}\,\Re\braket{z}{\tau_x u_\kappa}\Big)\, \dd x\\
			\geq &- \norm{g}_1\, \abs{\underset{\R}{\inf}\ P} > -\infty\,.
		\end{align*}
	\end{proof}
	
		\bigskip
		
	 Under a stronger assumption on the spatial cutoff $g\in L^1(\R) \cap H^1(\R)$, we now show that the solution $\varphi_t$ given by Corollary \ref{cor.existence_unicite_solution_H^s_Pphi2} is globally defined for any initial data $\varphi_0 \in \Hi^{3/2}_\omega$. The proof relies on the combination of global well-posedness on $\Hi^{1}_\omega$, Lemma \ref{lemma.glob_sol_mild_Pphi2_in_H^3/2}, conservation of the energy functional, Lemma \ref{lemma.sol_mild_Pphi2_conserve_energie}, and the blowup criterion stated in Remark \ref{rem.blowup_criterion_Pphi2}.
	
	\begin{lemma}\label{lemma.glob_sol_mild_Pphi2_in_H^3/2}
 Let	$g\in L^1(\R) \cap H^1(\R)$. Then for any initial data $\varphi_0 \in \Hi^{3/2}_\omega$, the Cauchy problem \eqref{eq_classique_avec_partie_libre_Pphi2} admits a unique global mild solution $\varphi_t \in C^0\big(\R,\Hi^{3/2}_\omega\big)\cap 
C^1\big(\R,\Hi^{1/2}_\omega\big)$. Furthermore, the energy functional $h$ defined by \eqref{def_energy_functional_Pphi2_model} is constant along the curve $\varphi_t$.
	\end{lemma}
		\begin{proof}
	This is a consequence of \cite[Theorem~1, page~5 and Theorem~2, page~9]{Re76}. We apply these results in the following setting: $\Hi^{1/2}_\omega$ is the underlying Hilbert space, $\omega:\Hi^{3/2}_\omega \to \Hi^{1/2}_\omega$ is the self-adjoint operator on $\Hi^{1/2}_\omega$ and the 
		$\gatbar{}b^\circ(z)$ is the non-linearity. The two main required estimates are: 
			\begin{equation}\label{premiere_estimee_champ_vecteurs_Pphi2_H^3/2}
	 \norm{\gatbar{}{b^\circ}(z)}_{\Hi^{3/2}_\omega} \leq C\left(\norm{z}_{\Hi^{1/2}_\omega}\right) \norm{z}_{\Hi^{3/2}_\omega}
	\end{equation}
	and
			\begin{equation}\label{estimee_champ_vecteurs_Pphi2_H^3/2}
	 \norm{\gatbar{}{b^\circ}(z_1)-\gatbar{}{b^\circ}(z_2)}_{\Hi^{3/2}_\omega} \leq C\left(\norm{z_1}_{\Hi^{3/2}_\omega},\norm{z_2}_{\Hi^{3/2}_\omega}\right) \norm{z_1-z_2}_{\Hi^{3/2}_\omega}\,.
	\end{equation}
	
	We proceed in two steps: 
	\begin{enumerate}
	 \item 
Recall that $u = \sqrt{2} \, {\omega^{-1/2}}$ and for $0 \leq j \leq 2n$, $0 \leq \ell \leq j$, 
\begin{equation}\label{dzbar_bjl_fourier_Pphi2}
\gatbar{}{b_{j,\ell}(z)=\ell \, u \, \Four\bigg(g \,\Four(\bar z u)^{\ell-1} \,
\overline{\Four(\bar z u)}^{j-\ell}\bigg)\in L^2(\R,\dd k).
\end{equation}
Let us note that the function $f_z = \Four(\bar{z}\, u)$ belongs to $H^{1}(\R)$ whenever $z \in \Hi^{1/2}_\omega$, with 
	\begin{align*}
	 \norm{f_z}_{H^{1}(\R)}
	 &\lesssim_{m_0} \norm{z}_{\Hi^{1/2}_\omega}\,.
	\end{align*}
The symbol $\lesssim_{m_0}$ indicates that the constant involved in the inequality depends upon the mass of the field $m_0$.

	 Using the expression \eqref{dzbar_bjl_fourier_Pphi2}, we have 
	\begin{align*}
	 \norm{\omega^{3/2}\,\gatbar{}{b_{j,\ell}}(z)}_2 &= \sqrt{2}\, \ell\, \norm{\omega\,\Four\left(g\, f_z^{\ell-1}\, \bar{f_z}^{j-\ell}\right)}_2\\
	 &\lesssim_{m_0} \norm{g\, f_z^{\ell-1}\, \bar{f_z}^{j-\ell}}_{H^{1}(\R)}\,.
	\end{align*}
	Since $H^{1}(\R)$ is an algebra, this yields 
	\begin{align*}
	 \norm{g\, f_z^{\ell-1}\, \bar{f_z}^{j-\ell}}_{H^{1}(\R)}
	 &\leq \norm{g}_{H^{1}(\R)}\, \norm{f_z}_{H^{1}(\R)}^{j-1}\\
	 &\lesssim_{m_0} \norm{g}_{H^{1}(\R)}\, \norm{z}_{\Hi^{1/2}_\omega}^{j-1}\,,
	\end{align*}
	whence finally 
	\[
	\norm{\gatbar{}{b_{j,\ell}}(z)}_{\Hi^{3/2}_\omega} \lesssim_{m_0} \norm{g}_{H^{1}(\R)}\, \norm{z}_{\Hi^{1/2}_\omega}^{j-1}\,.
	\]
	This concludes the proof of \eqref{premiere_estimee_champ_vecteurs_Pphi2_H^3/2}. To prove \eqref{estimee_champ_vecteurs_Pphi2_H^3/2}, we proceed like above to dominate \[\norm{\gatbar{}{b_{j,\ell}}(z_1) - \omega\,\gatbar{}{b_{j,\ell}}(z_2)}_{\Hi^{3/2}_\omega}\,.\] The later estimate allows to apply \cite[Theorem~1, page~5]{Re76} and hence obtain existence and uniqueness of a local mild solution $\varphi_t \in C^0(I,\Hi^{3/2}_\omega) \cap C^1(I,\Hi^{1/2}_\omega)$ to \eqref{eq_classique_avec_partie_libre_Pphi2} with initial data $\varphi_0 \in \Hi^{3/2}_\omega$.
	
	\item Since the symbol $b^\circ$ appearing in the energy functional $h$ (see \eqref{def_energy_functional_Pphi2_model}) is real-valued, we have 
	\begin{equation}\label{conserv_energy_strong_sol_Pphi2}
	 \dv{t} \, h(\varphi_t) = 2 \, \Re \braket{\dot{\varphi}_t}{\omega \varphi_t + \gatbar{}{b^\circ}(\varphi_t)} = 2 \, \Re \braket{\dot{\varphi}_t}{\rmi \, \dot{\varphi}_t} = 0\,,
	\end{equation}
	which proves that $h$ is constant along $\varphi_t$. In particular, it follows from \eqref{energie_Pphi2_controle_norme_espace_energie} that $\varphi_t$ is bounded in $\Hi^{1/2}_\omega$: this fact and the estimate \eqref{premiere_estimee_champ_vecteurs_Pphi2_H^3/2} enable us to apply \cite[Theorem~2, page~9]{Re76} and conclude that $\varphi_t$ is globally defined.
	\end{enumerate}
	}
	
		\end{proof}
	
\begin{lemma}\label{lemma.sol_mild_Pphi2_conserve_energie}
 Let $g\in L^1(\R) \cap H^1(\R)$. Then, the energy functional $h$ given by \eqref{def_energy_functional_Pphi2_model} is constant along any mild solution $\varphi_t$ of the Cauchy problem \eqref{eq_classique_avec_partie_libre_Pphi2} with initial data $\varphi_0 \in \Hi^{1/2}_\omega$.
	\end{lemma}

		\begin{proof}
		For $\varphi_0 \in \Hi^{1/2}_\omega$ fixed, let $\Big(\varphi_0^{(n)}\Big)_{n \in \N} \in \left(\Hi^{3/2}_\omega\right)^\N$ such that $\norm{\varphi_0^{(n)} - \varphi_0}_{\Hi^{1/2}_\omega} \xrightarrow[n \to \infty]{} 0$ and let ${\varphi_t^{(n)} \in C^0\left(\R,\Hi^{3/2}_\omega\right)}$ be the global solution of \eqref{eq_classique_avec_partie_libre_Pphi2} with initial data $\varphi_0^{(n)}$. 
		For all $n \in \N$, the identity \eqref{conserv_energy_strong_sol_Pphi2} yields 
		\begin{equation}\label{conserv_energy_init_data_H^3/2}
			h\left(\varphi_t^{(n)}\right) = h\left(\varphi_0^{(n)}\right)\,,\quad \forall\, t \in \R\,.
		\end{equation}
		Furthermore, following \cite[Theorem 14, page 52]{Re76} and using \eqref{estimee_champ_vecteurs_Pphi2_espace_energie}, we get 
		\begin{align*}
			\norm{\varphi_t - \varphi_t^{(n)}}_{\Hi^{1/2}_\omega} \leq &\norm{\varphi_0 - {\varphi_0^{(n)}}}_{\Hi^{1/2}_\omega} + \int_0^{\abs{t}} \norm{X\left(\tau,\tilde{\varphi}_\tau\right)-X\left(\tau,\tilde\varphi^{(n)}_\tau\right)}_{\Hi^{1/2}_\omega} \dd \tau\\
			\leq &\norm{\varphi_0 - {\varphi_0^{(n)}}}_{\Hi^{1/2}_\omega} + C\left(t,\norm{\varphi_0}_{\Hi^{1/2}_\omega}\right) \int_0^{\abs{t}} \norm{\varphi_\tau-\varphi^{(n)}_\tau}_{\Hi^{1/2}_\omega} \dd \tau\,,
		\end{align*} whence by Grönwall lemma \[\norm{\varphi_t - \varphi_t^{(n)}}_{\Hi^{1/2}_\omega} \leq \norm{\varphi_0 - \varphi_0^{(n)}}_{\Hi^{1/2}_\omega} \cdot M\left(t,\norm{\varphi_0}_{\Hi^{1/2}_\omega}\right) \xrightarrow[n \longrightarrow \infty]{} 0\,.
		\]
		Since the energy functional $h : \Hi^{1/2}_\omega \to \R$ is continuous with respect to the norm $\norm{\cdot}_{\Hi^{1/2}_\omega}$, letting $n \longrightarrow + \infty$ in \eqref{conserv_energy_init_data_H^3/2}, we finally obtain 
		\[
		h(\varphi_t) = h(\varphi_0)\,.
		\]
	\end{proof}

 We can finally state the main result of this subsection.

	\begin{theorem}\label{thm.existence_solution_mild_globale_eq_classique_avec_partie_libre_Pphi2}
	Let $g\in L^1(\R) \cap H^1(\R)$. For all initial data $\varphi_0 \in \Hi^{1/2}_\omega$, the Cauchy problem \eqref{eq_classique_avec_partie_libre_Pphi2} admits a unique global mild solution $\varphi_t \in C^0\big(\R,\Hi^{1/2}_\omega\big)$ satisfying \eqref{eq_classique_avec_partie_libre_formulation_integrale_Pphi2}.
	\end{theorem}
	
	\begin{proof}
		Let $\varphi_t \in C^0\big(I,\Hi^{1/2}_\omega\big)$ be the unique maximal mild solution of the equation \eqref{eq_classique_avec_partie_libre_Pphi2}. The curve 
		${\tilde{\varphi}_t = \e^{\rmi t \omega}\varphi_t\, \in C^1\left(I,\Hi^{1/2}_\omega\right)}$ is the unique maximal solution of 
			 \eqref{eq_classique_sans_partie_libre_Pphi2} according to Corollary \ref{cor.existence_unicite_solution_H^s_Pphi2}. Suppose by contradiction that $T \defegal \sup I$ is finite. In view of the blowup criterion given in Remark \ref{rem.blowup_criterion_Pphi2}, the map $t \longmapsto \norm{\omega^{1/2}\,\tilde\varphi_t}_2$ is not bounded in the neighborhood of $T$. Up to extraction, we thus have $\norm{\omega^{1/2}\, \tilde\varphi_t}_2 \xrightarrow[t \to T]{} +\infty$. Combining this result with \eqref{energie_Pphi2_controle_norme_espace_energie} and Lemma \ref{lemma.sol_mild_Pphi2_conserve_energie}, we obtain \[
		h(\varphi_0) + c = h(\varphi_t) + c \geq \norm{\omega^{1/2}\, \varphi_t}_2^2=\norm{\omega^{1/2}\, \tilde\varphi_t}_2^2 \xrightarrow[t \to T]{} +\infty\,,
		\] 
		which is a contradiction. Then $\sup I = + \infty$. A time reversal argument shows that $\inf I = - \infty$. 	
	\end{proof}

	\begin{cor}\label{cor.existence_sol_mild_global_L2_Pphi2}
	 Let $g\in L^1(\R) \cap H^1(\R)$. For all initial data $\varphi_0 \in \hi = L^2(\R,\dd k)$, the Cauchy problem \eqref{eq_classique_avec_partie_libre_Pphi2} admits a unique global mild solution $\varphi_t \in C^0\big(\R,\hi\big)$ satisfying \eqref{eq_classique_avec_partie_libre_formulation_integrale_Pphi2}.
	\end{cor}
	
	\begin{proof}
	Let $\varphi_0 \in \hi$ be fixed and let $\varphi_t \in C^0(I_0,\hi)$ be a local mild solution of \eqref{eq_classique_avec_partie_libre_Pphi2} with initial data $\varphi_0$, which is well-defined thanks to the estimate \eqref{estimee_L2_champ_vecteurs_Pphi2_espace_energie} and the Cauchy-Lipschitz theorem applied to the equation \eqref{eq_classique_sans_partie_libre_Pphi2}. We assume that $T_0 = \sup I_0$ is finite. Our aim is to extend the solution $\varphi_t$ to any interval $[0,T_0+\eta]$ with $\eta > 0$, which will yield globalness by a time reversal argument. 
	
	Let ${\left(\varphi_0^{(n)}\right)_{n \in \N} \in \left(\Hi^{1/2}_\omega\right)^{\N}}$ such that $\varphi_0^{(n)} \xrightarrow[n \to \infty]{} \varphi_0$ in $\hi$, and for each $n \in \N$, let $\varphi_t^{(n)}$ be the global mild solution of \eqref{eq_classique_avec_partie_libre_Pphi2} with initial data $\varphi_0^{(n)}$, given by Theorem \ref{thm.existence_solution_mild_globale_eq_classique_avec_partie_libre_Pphi2}. Let $\eta > 0$. Using \eqref{estimee_L2_champ_vecteurs_Pphi2_espace_energie} and following the proof of Lemma \ref{lemma.sol_mild_Pphi2_conserve_energie} with $\norm{\cdot}_{\Hi^{1/2}_\omega}$ replaced by $\norm{\cdot}_2$, we obtain that the sequence $\left(\eval{\varphi_t^{(n)}}_{[0,T_0+\eta]}\right)_{n \in \N}$ is Cauchy in the Banach space $C^0([0,T_0+\eta],\hi)$. Its limit satisfies the equation \eqref{eq_classique_avec_partie_libre_formulation_integrale_Pphi2} on $[0,T_0+\eta]$ and coincides with $\varphi_t$ on $[0,T_0)$, which extends $\varphi_t$ to $[0,T_0+\eta]$ and concludes the proof.
	\end{proof}

	\begin{remark}\label{rem.existence_unicite_solution_H^s_Pphi2}
	Let us make a few comments about the results of this subsection, which can be extended providing the same regularity on the cutoff function $g$:

	\begin{itemize}	 \item Let $0 < s \leq 1/2$. Replacing the estimate \eqref{majoration_norme_infty_pour_champ_vecteurs_Pphi2} by \[\norm{f_z}_\infty \leq \norm{z}_{\Hi^s_\omega} \, \norm{\omega^{-s}\, u}_2\] for $z \in \Hi^s_\omega$, we can adapt the proof of \eqref{estimee_champ_vecteurs_Pphi2_espace_energie} to the case $\Hi^s_\omega$ whenever ${g \in L^1(\R) \cap L^2(\R)}$. As a consequence, for any initial data $\varphi_0 \in \Hi^s_\omega$, $0 < s \leq 1/2$, the equation \eqref{eq_classique_avec_partie_libre_Pphi2} admits a unique maximal mild solution $\varphi_t \in C^0(I_s, \Hi^s_\omega)$. Furthermore, assuming $g\in L^1(\R) \cap H^1(\R)$ and adapting the proof of Corollary \ref{cor.existence_sol_mild_global_L2_Pphi2}, we see that $I_s = \R$.
\item For all $1/2 < s \leq 3/2$, a similar adaptation of the proof of Lemma \ref{lemma.glob_sol_mild_Pphi2_in_H^3/2} yields global well-posedness of \eqref{eq_classique_avec_partie_libre_Pphi2} for any initial data $\varphi_0 \in \Hi^{s}_\omega$, provided that $g \in L^1(\R) \cap H^1(\R)$.
\end{itemize}

	\end{remark}

	\subsection[The quadratic dynamics]{The quadratic dynamics for the $\ppde$ model}\label{subsect.The_quadratic_dynamic_Pphi2}
	
	In order to study the propagation of coherent states in the classical limit, we need some information about the dynamics of time-dependent quadratic Hamiltonians. Although the results of this subsection concern the $\ppde$ model in the Hilbert space $\hi = L^2(\R,\dd k)$, we will see later that they can easily be extended to a more general case: for this reason, we prefer to state them in a rather abstract way, following \cite[Section 4]{AmZe14}.
	
	\medskip
	For all $k \in \N$, we denote by $\mathcal{D}_{+,k}$ the space $D\left(\nbop[1]^{k/2}\right)$ equipped with the inner product \[\braket{\varphi}{\psi}_{\mathcal{D}_{+,k}} = \sum_{n = 0}^\infty (n+1)^k \braket{\varphi^{(n)}}{\psi^{(n)}}\,,\] and by $\mathcal{D}_{-,k}$ the completion of $\focks$ with respect to the inner product 
	\[
	\braket{\varphi}{\psi}_{\mathcal{D}_{-,k}} = \sum_{n = 0}^\infty (n+1)^{-k} \braket{\varphi^{(n)}}{\psi^{(n)}}\,.
	\] 
	This provides the Hilbert rigging 
	\[
	\mathcal{D}_{+,k} \subset \focks \subset \mathcal{D}_{-,k}\,.
	\] 
	We also define 
	\[
	\mathcal{D}_{+,\infty} = \displaystyle\bigcap_{k \in \N} \mathcal{D}_{+,k} \,.
	\]
	
	Let $\Big(\mathfrak{Q}_t\Big)_{t \in I}$ be a family of self-adjoint operators on $\focks$. A family of unitary operators $\left(U(t,s)\right)_{t,s \, \in \, I}$ is said to be a unitary propagator of the problem 
	\begin{equation}\label{abstract_Cauchy_problem_quadratic_dynamics}
		\left\{\begin{array}{l}
			\rmi\, \partial_t u_t = \mathfrak{Q}_t\, u_t\\
			\eval{u_t}_{t=0} = u_0 \in \mathcal{D}_{+,k}
		\end{array}\right.
	\end{equation}
	if it satisfies:
	\begin{enumerate}
		\item $U(t,t) = \1$ for all $t \in I$,
		\item $U(t,r) \circ U(r,s) = U(t,s)$ for all $r,s,t \in I$,
		\item For each $s \in I$, we have $U(\bullet,s) \in C^0\big(I,\mathcal{L}(\mathcal{D}_{+,k})\big) \cap C^1\big(I,\mathcal{L}(\mathcal{D}_{+,k},\mathcal{D}_{-,k})\big)$ and, for all $u_0 \in \mathcal{D}_{+,k}$, 
		\[
		\rmi\, \partial_t U(t,s)\, u_0 = \mathfrak{Q}_t\, U(t,s)\, u_0 \,.
		\]
	\end{enumerate}
	Here, $C^k(I,\mathfrak{B})$ denotes the space of $k$-continuously differentiable $\mathfrak{B}$-valued functions where $\mathfrak{B}$ is endowed with the strong operator topology.
	
	\subsubsection{Existence of the unitary propagator}\label{par.exist_propag_quadra_Pphi2}
	
	We now return to the $\ppde$ model introduced in Subsection \ref{subsect.intro_Pphi2_model}. From now on, we assume that the spatial cutoff $g$ belongs to $L^1(\R) \cap H^1(\R)$. Let us note that the Wick symbol of the interaction term \eqref{def_rigoureuse_spatially_cutoff_interaction_Pphi2} takes the form 
	\begin{multline}\label{interaction_Pphi2_de_la_forme_Fv}
		b^\circ(z) = F_V(z) = \sum_{j=0}^{2n} \braket{\frac{\tensn{(z + \conjug z)}{j}}{\sqrt{j!}}}{V^{(j)}}\\ \text{ with } V^{(j)} = \left\{\begin{array}{l}
			\beta_0\, \norm{g}_1 \text{ if } j = 0\\
			\sqrt{j!}\,2^{-j/2}\, \beta_j\, w_j \in L^2_{\mathrm{\mathrm{sym}}}(\R^j) \simeq \tensymn{L^2(\R)}{j} \text{ if } 1 \leq j \leq 2n\\
			0 \text{ if } j > 2n\,,
		\end{array}\right.
	\end{multline} 
	where $n$ and $(\beta_j)_{0 \leq j \leq 2n}$ appear in the polynomial \eqref{def_polynome_modele_Pphi2}, $\conjug z (k) = \bar{z(-k)}$ is a conjugation on $L^2(\R,\dd k)$, and 
	\[
	w_j(k_1,\dots,k_j) = 2^{j/2}\, \hat{g}(k_1 + \cdots + k_j) \prod_{m=1}^j \omega(k_m)^{-1/2} \in L^2_{\mathrm{\mathrm{sym}}}(\R^j) \simeq \tensymn{L^2(\R)}{j}
	\] is the Schwartz kernel given by \eqref{def_noyau_Schwartz_Pphi2_Derezinski_Gerard} (see also \cite[Lemma 6.1]{DeGe00}). 
	
	For all $t \in \R$, let us define 
	\begin{equation}\label{def_potentiel_2_Fv}
		V_2(t) = \frac{1}{\sqrt{2}} \sum_{j=2}^{2n}\, j\,(j-1) \,2^{-\frac{j}{2}}\, \beta_j \, \symop[2]\Big(\bra{\tensn{(\varphi_t + \conjug \varphi_t)}{j-2}} \ot (\tensn{\1}{2})\Big)\,w_j \in L^2_{\mathrm{\mathrm{sym}}}(\R^2)\,,
	\end{equation}
	which corresponds to the quadratic polynomial 
	\begin{equation}
	 F_{V_2(t)}(z) = \frac{1}{\sqrt{2}} \braket{\tensn{(z + \conjug z)}{2}}{V_2(t)}
	\end{equation}
	appearing in the Taylor expansion \eqref{asymp_Taylor_outline_proof}. We recall that $\varphi_t \in C^0(\R,L^2(\R,\dd k))$ is the mild solution of \eqref{eq_classique_avec_partie_libre_Pphi2} given by Corollary \ref{cor.existence_sol_mild_global_L2_Pphi2}. The aim of this paragraph is to prove the existence of a unique unitary propagator for the operators $\dG[\eps]{\omega} + \Wick[\eps]{F_{V_2(t)}}$. Again, by passing to the interaction representation, we rather consider the Cauchy problem \eqref{abstract_Cauchy_problem_quadratic_dynamics} associated with the operators
	\begin{equation}\label{def_WickQ_t_Pphi2}
		\mathfrak{Q}_t = \eps^{-1}\,\Wick[\eps]{F_{\tilde{V}_2(t)}^{\conjug_t}}\,,
	\end{equation} 
	where
	\begin{equation}\label{def_Fv_quad_tordu_Pphi2}
		F_{\tilde{V}_2(t)}^{\conjug_t}(z) = \frac{1}{\sqrt{2}} \braket{\tensn{(\e^{-\rmi t \omega}\,z + \e^{\rmi t \omega} \, \conjug z)}{2}}{V_2(t)}\,.
	\end{equation}
	It is worth noticing that the operators $\mathfrak{Q}_t $, given in \eqref{def_WickQ_t_Pphi2}, are $\eps$-independent since they are quadratic in terms of creation-annihilation operators.
	
	\medskip
	
	We will use the result \cite[Corollary C.4]{AmBr12}, stated below in a weaker version that fits our needs. 
	
	\begin{theorem}\label{thm.quadratic_dynamics_Ammari_Breteaux_chaospropag}
		Let $I \subset \R$ be a closed interval and $k \in \N^*$ be fixed. Let $(\mathfrak{Q}_t)_{t \in I}$ be a family of self-adjoint operators on $\focks$ such that $\mathfrak{Q}_t \in \mathcal{L}\big(\mathcal{D}_{+,k},\mathcal{D}_{-,k}\big)$. We assume that 
		\begin{enumerate}
			\item the map $t \longmapsto \norm{\mathfrak{Q}_t}_{\mathcal{L}(\mathcal{D}_{+,k},\mathcal{D}_{-,k})}$ is continuous, 
			\item there exists a continuous function $f_k : I \longrightarrow \R_+$ such that, for all $\varphi,\psi \in \focksalg$ and $t \in I$, 
			\begin{equation}\label{estimee_commut_corollary_C4_Ammari_Breteaux_chaospropag}
				\abs{\braket{(\nbop[1]+1)^k \varphi}{\mathfrak{Q}_t \psi} - \braket{\mathfrak{Q}_t \varphi}{(\nbop[1]+1)^k \psi}} \leq f_k(t)\, \norm{\varphi}_{\mathcal{D}_{+,k}}\, \norm{\psi}_{\mathcal{D}_{+,k}}\,.
			\end{equation}
			Then, the Cauchy problem \eqref{abstract_Cauchy_problem_quadratic_dynamics} admits a unique unitary propagator $(U(t,s))_{t,s \, \in \, I}$. 
		\end{enumerate}
	\end{theorem}
	
	\noindent
	Let us remark that we have 
	\[
	F_{\tilde{V}_2(t)}^{\conjug_t}(z) = \frac{1}{\sqrt{2}} \braket{\tensn{(z+\conjug_t z)}{2}}{\tilde{V}_2(t)}
	\] 
	where $\conjug_t z = \e^{2\rmi t \omega}\conjug z$ and $\tilde{V}_2(t) = \e^{\rmi\frac{t}{\eps} \dG[\eps]{\omega}}\, V_2(t)$. The assumption $\liftop{\conjug}V_2(t) = V_2(t)$ combined with the identity $\conjug\, \e^{\rmi t \omega} = \e^{-\rmi t \omega}\, \conjug$ implies 
	\begin{equation}
	 \liftop{\conjug_t}\tilde{V}_2(t) = \tilde{V}_2(t)\,.
	\end{equation}
	Thus, applying Theorem \ref{thm.rpz_ondulatoire_espace_Fock} with $\conjug$ replaced by $\conjug_t$ and following Proposition \ref{prop.Wick_non_poly_est_essentiellement_autoadjoint}, we obtain that the operators $\mathfrak{Q}_t$ given in \eqref{def_WickQ_t_Pphi2} are essentially self-adjoint. We keep the same notation for their closure. Let us now verify that these operators fulfill the assumptions of Theorem \ref{thm.quadratic_dynamics_Ammari_Breteaux_chaospropag}.
	
	\begin{enumerate}
		\item For all $T > 0$ and $s,t \in [-T,T]$, the estimate 
		\begin{equation}
			\norm{\tensn{(\varphi_t + \conjug \varphi_t)}{j-2} - \tensn{(\varphi_s + \conjug \varphi_s)}{j-2}}_2 \leq 2\, (j-2)\, \norm{\varphi_t - \varphi_s}_2 \left(\underset{\abs{\tau} \leq T}{\sup} \left(2 \norm{\varphi_\tau}_2\right) \right)^{j-3}\,,
		\end{equation}
		 holding for all $j \geq 3$, yields (see \eqref{majoration_difference_puissances_produit_tensoriel})
		\begin{equation*}
			\norm{V_2(t)-V_2(s)}_2 \leq \sqrt{2}\, \norm{\varphi_t-\varphi_s}_2\, \sum_{j=3}^{2n}\, j\,(j-1)\,(j-2)\, 2^{-\frac{j}{2}} \, \abs{\beta_j} \left(\underset{\abs{\tau} \leq T}{\sup} (2 \norm{\varphi_\tau}_2) \right)^{j-3} \norm{w_j}_2\,,
		\end{equation*}
		which shows that the map $t \longmapsto \norm{V_2(t)}_2$ is continuous. Let us recall that $w_j$ denotes the Schwartz kernel given by \eqref{def_noyau_Schwartz_Pphi2_Derezinski_Gerard}. Now, let us write 
		\begin{equation}\label{decomp_symbole_quad_tordu}
		 	F_{\tilde{V}_2(t)}^{\conjug_t} = b_{2,0}(t) + b_{0,2}(t) + b_{1,1}(t)\,,
		\end{equation}
		where the polynomial symbols $b_{p,q}(t) \in \wickpol[p,q] $, $p+q=2$ are associated with the operators 
		\[
		\tilde{b}_{0,2}(t) = \frac{1}{\sqrt{2}}\, \ket{\tilde{V}_2(t)}\,, \quad \tilde{b}_{1,1}(t) = \sqrt{2}\, \Big(\1 \ot \bra{\conjug_t \, \bullet}\Big)\tilde{V}_2(t)\quad \mathrm{ and }\quad \tilde{b}_{2,0}(t) = \frac{1}{\sqrt{2}}\, \bra{\tilde{V}_2(t)}\,.
		\] 
		For all $\psi \in \focksalg$, we have by explicit computation 
		\begin{align}\label{explicit_computation_continuity_Qt_quadratic_dynamics}
			\norm{\Wick[\eps]{F_{\tilde{V}_2(t)}^{\conjug_t}} \psi}_{\mathcal{D}_{-,k}}^2 \leq &\,3\, \Big( \norm{\Wick[\eps]{b_{2,0}(t)} \psi}_{\mathcal{D}_{-,k}}^2\\
			&\phantom{3\,\Big(}+ \norm{\Wick[\eps]{b_{0,2}(t)} \psi}_{\mathcal{D}_{-,k}}^2 \nonumber\\
			&\phantom{3\,\Big(}+ \norm{\Wick[\eps]{b_{1,1}(t)} \psi}_{\mathcal{D}_{-,k}}^2\Big)\nonumber\\
			\leq &\,9\, \eps^2\, \norm{V_2(t)}_2^2\, \norm{\psi}_{\mathcal{D}_{+,k}}^2\,.\nonumber
		\end{align}
		
		Similarly, for all $t,s \in \R$, we have 
		\begin{equation*}
			\norm{\Wick[\eps]{F_{\tilde{V}_2(t)}^{\conjug_t}}\psi - \Wick[\eps]{F_{\tilde{V}_2(s)}^{\conjug_s}} \psi}_{\mathcal{D}_{-,k}}^2 \leq\, 9 \, \eps^2\, \norm{V_2(t)-V_2(s)}_2^2\, \norm{\psi}_{\mathcal{D}_{+,k}}^2.
		\end{equation*}
		Thus, the map $t \longmapsto  \norm{\mathfrak{Q}_t}_{\mathcal{L}(\mathcal{D}_{+,k},\mathcal{D}_{-,k})}$ is continuous.
		\item We now establish an estimate of the form \eqref{estimee_commut_corollary_C4_Ammari_Breteaux_chaospropag}. Using the formula 
		\[
		\left[(\nbop[1]+1)^k, \Wick[\eps]{b}\right] \psi^{(n)} = \left((n+q-p+1)^k-(n+1)^k\right)\Wick[\eps]{b}\psi^{(n)}
		\] for $b \in \wickpol[p,q]$ and $\psi^{(n)} \in \tensymn{\hi}{n}$, and using the decomposition \eqref{decomp_symbole_quad_tordu}, an explicit computation shows that
		\begin{equation*}
		 \braket{(\nbop[1]+1)^k \varphi}{\Wick[\eps]{F_{\tilde{V}_2(t)}^{\conjug_t}}\psi} - \braket{\Wick[\eps]{F_{\tilde{V}_2(t)}^{\conjug_t}} \varphi}{\Wick[\eps]{(\nbop[1]+1)^k\psi}}
		\end{equation*}
		is equal to 
		\begin{equation*}
		 \sum_{n=0}^\infty \big((n+3)^k-(n+1)^k\big) \Big(\braket{\varphi^{(n+2)}}{\Wick[\eps]{b_{0,2}(t)}\psi^{(n)}} - \braket{\varphi^{(n)}}{\Wick[\eps]{b_{2,0}(t)}\psi^{(n+2)}}\Big)
		\end{equation*}
		for all $\varphi,\psi \in \focksalg$. Since $(n+3)^k - (n+1)^k \leq 2\, k\, (n+3)^{k-1}$, we obtain 
		\begin{multline*}
			\abs{\braket{(\nbop[1]+1)^k \varphi}{\Wick[\eps]{F_{\tilde{V}_2(t)}^{\conjug_t}}\psi} - \braket{\Wick[\eps]{F_{\tilde{V}_2(t)}^{\conjug_t}} \varphi}{\Wick[\eps]{(\nbop[1]+1)^k\psi}}}\\ \leq \sqrt{2}\, \eps\, k\, \norm{V_2(t)} \sum_{n=0}^\infty (n+3)^{k-1} \sqrt{(n+1)(n+2)} \left(\norm{\varphi^{(n+2)}}\norm{\psi^{(n)}} + \norm{\varphi^{(n)}}\norm{\psi^{(n+2)}}\right)\\
			\leq \sqrt{2}\, \eps\, k\, \norm{V_2(t)} \sum_{n=0}^\infty (n+3)^k\, \left(\norm{\varphi^{(n+2)}}\norm{\psi^{(n)}} + \norm{\varphi^{(n)}}\norm{\psi^{(n+2)}}\right)\,.
		\end{multline*}
		Dominating $(n+3)^k$ by $3^{k/2}\, (n+1)^{k/2}\, (n+3)^{k/2}$ and applying Cauchy-Schwarz inequality, we obtain  
		\begin{multline}\label{estimee_forme_quad_propag_quad_Pphi2}
			\abs{\braket{(\nbop[1]+1)^k \varphi}{\mathfrak{Q}_t\,\psi} -\braket{\mathfrak{Q}_t\,\varphi}{\Wick[\eps]{(\nbop[1]+1)^k}\psi}}\\ 
			\leq 2\, \sqrt{2}\, k\, 3^{k/2} \norm{V_2(t)} \norm{\varphi}_{\mathcal{D}_{+,k}}\norm{\psi}_{\mathcal{D}_{+,k}}.
		\end{multline}
	\end{enumerate}
	
	We can finally state our main result concerning the dynamics of a time-dependent quadratic Hamiltonian for the $\ppde$ model.
	
	\begin{cor}\label{cor.dynamique_quad_modele_Pphi2}
		For all initial data $u_0 \in \mathcal{D}_{+,k}$ with $k \in \N^*$, the non-autonomous Cauchy problem 
		\begin{equation}\label{Cauchy_problem_sans_partie_libre_quadratic_dynamics_Pphi2}
			\left\{\begin{array}{l}
				\rmi \,\eps\, \partial_t \tilde{u}_t = \Wick[\eps]{F_{\tilde{V}_2(t)}^{\conjug_t}}\, \tilde{u}_t\\
				\eval{\tilde{u}_t}_{t=0} = u_0 \in \mathcal{D}_{+,k}
			\end{array}\right.
		\end{equation}
		admits a unique unitary propagator $\big(\tilde{U}_2^\circ(t,s)\big)_{t,s \, \in \, \R}$. Furthermore, we have the estimate 
		\begin{equation}\label{estimee_norme_propag_quad_propagation_etats_coherents}
			\norm{\tilde{U}_2^\circ(t,s)}_{\mathcal{L}(\mathcal{D}_{+,k})} \leq \e^{2\, \sqrt{2}\, k\, 3^{k/2}\, \abs{\int_s^t \norm{V_2(\tau)} \dd \tau}} \,.
		\end{equation}
		Finally, if we set $U_2^\circ(t,s) = \e^{-\rmi\frac{t}{\eps}\dG[\eps]{\omega}}\, \tilde{U}_2^\circ(t,s)\, \e^{\rmi\frac{s}{\eps}\dG[\eps]{\omega}}$, the map $t \longmapsto U_2^\circ(t,s)\, u_s$ is a mild solution of the Cauchy problem 
		\begin{equation}\label{Cauchy_problem_avec_partie_libre_quadratic_dynamics_Pphi2}
			\left\{\begin{array}{l}
				\rmi \, \eps\, \partial_t u_t = \left(\dG[\eps]{\omega} + \Wick[\eps]{F_{V_2(t)}}\right) u_t\\
				\eval{u_t}_{t=0} = u_0 \in \mathcal{D}_{+,k} \,.
			\end{array}\right.
		\end{equation}
	\end{cor}
	
	\begin{proof}
	 The existence and uniqueness of $\tilde{U}_2^\circ(t,s)$ follow from Theorem \ref{thm.quadratic_dynamics_Ammari_Breteaux_chaospropag}. The proof of the operator norm estimate \eqref{estimee_norme_propag_quad_propagation_etats_coherents} is analogous to the one of \eqref{estimee_forme_quad_propag_quad_Pphi2}, since the differentiation of the quantity 
	 \[
	 u(t) = \norm{(\nbop[1]+1)^{k/2}\, \tilde{U}_2^\circ(t,0)\, \psi}^2
	 \]
	 in the sense of quadratic form yields 
	 \[
	 \rmi \, \eps\, \dot{u}(t) = \braket{\tilde{U}_2^\circ(t,0)\,\psi}{\comm{(\nbop[1]+1)^k}{\Wick[\eps]{F_{\tilde{V}_2(t)}^{\conjug_t}}}\, \tilde{U}_2^\circ(t,0)\,\psi}\,.
	 \]
	 The Grönwall lemma ends the proof.
	\end{proof}
	
	\subsubsection{Quadratic propagation of a polynomial observable}\label{par.quad_propag_Wick_observable}
	
	We now state a result due to S. Breteaux \cite{Br12} concerning the propagation of a polynomial observable under a time-dependent quadratic Hamiltonian. 
	
	Let us write \[F_{\tilde{V}_2(t)}^{\conjug_t}(z) = \braket{z}{\alpha_t \, z} + \Im\braket{\beta_t}{\tensn{z}{2}}\] where $\beta_t = -\rmi \, \sqrt{2}\, \tilde{V}_2(t) \in C^0(\R,\tensymn{\hi}{2})$, and $\alpha_t = \sqrt{2}\,
 \Big(\1 \ot \bra{\conjug_t \, \bullet}\Big)\, \tilde{V}_2(t)$ defines a family of bounded, strongly continuous self-adjoint operators. Applying \cite[Theorem~X.69]{ReSi75}, we obtain that the assumptions \textbf{H1} and \textbf{H2} of \cite[Section~3, page~82]{Br12} are satisfied by $F_{\tilde{V}_2(t)}^{\conjug_t}$.
	
	The classical evolution associated with this operator is given by 
	\begin{equation}\label{eq_classique_quad_sans_partie_libre}
		\left\{\begin{array}{l}
			\rmi \, \dot{\tilde{z}}_t = \gatbar{}{F_{\tilde{V}_2(t)}^{\conjug_t}}(\tilde{z}_t)\\
			\eval{\tilde{z}_t}_{t = 0} = z_0 \in \hi\,,
		\end{array}\right.
	\end{equation}
	where 
	\[
	\gatbar{}{F_{\tilde{V}_2(t)}^{\conjug_t}}(u) = \sqrt{2}\, \big(\bra{(u+\conjug_t u)} \ot \1\big)\tilde{V}_2(t)
	\] 
	is $\R$-linear with respect to $u$. The inequality 
	\[
	\normhi{\gatbar{}{F_{\tilde{V}_2(t)}^{\conjug_t}}(u)-\gatbar{}{F_{\tilde{V}_2(t)}^{\conjug_t}}(v)} \leq 2\, \sqrt{2}\, \norm{\tilde{V}_2(t)} \normhi{u-v}
	\] shows that the classical flow $\tilde{\upphi}_t$ of the Cauchy problem \eqref{eq_classique_quad_sans_partie_libre} is globally defined (recall that $\normhi{\cdot}$ denotes the norm on $\hi$). Each map $\tilde{\upphi}_t : \hi \to \hi$ is a symplectomorphism, that is, a $\R$-linear, continuous automorphism which preserves the symplectic form $\symform{z_1}{z_2} = -2\, \Im\braket{z_1}{z_2}$.
	In fact, we have a unique decomposition \[\tilde{\upphi}_t = \tilde{L}_t + \tilde{A}_t\] where $\tilde{L}_t, \tilde{A}_t \in \mathcal{L}(\hi)$ are respectively $\C$-linear and $\C$-antilinear and $[\tilde{A}_t]^{*} \tilde{A}_t$ is a trace class operator, see \cite[Appendix A]{Br12}. Let $\Lambda_t$ be the time-dependent operator on $\wickpol$ defined by 
	\begin{equation}
		\Lambda_t\, b(z) = -2\, \Tr{\left([\tilde{A}_t]^{*} \tilde{A}_t \gatgatbar{}{}{b}(z)\right)} + \braket{v_t}{\gatbar{2}{b}(z)}_{\tensn{\hi}{2}} + \braket{\gat{2}{b}(z)}{v_t}_{(\tensn{\hi}{2})^*,\tensn{\hi}{2}}\,,
	\end{equation}
	where $v_t \in \tensn{\hi}{2}$ is the unique vector such that $\braket{z_1 \ot z_2}{v_t}_{\tensn{\hi}{2}} = \braket{z_1}{\tilde{L}_t \, [\tilde{A}_t]^{*}\, z_2}_{\hi}$ for all $z_1,z_2 \in \hi$. Applying \cite[Theorem 3.2]{Br12}, we obtain the following result.
	
	\begin{theorem}\label{thm.Breteaux_conjug_Wick_propag_quad}
		Let $\upphi_t = \e^{-\rmi t \omega}\,\tilde{\upphi}_t$ be the classical flow of 
		\begin{equation}\label{eq_classique_quad_avec_partie_libre}
		 \left\{\begin{array}{l}
			\rmi \, \dot{z}_t = \omega\, z_t + \gatbar{}{F_{V_2(t)}}(z_t)\\
			\eval{z_t}_{t = 0} = z_0 \in \hi\,,
		\end{array}\right.
		\end{equation} and let $U_2^\circ(t,s)$ be the quantum propagator of \eqref{Cauchy_problem_avec_partie_libre_quadratic_dynamics_Pphi2}. For all $t \in \R$ and $b_t \in \wickleq{m}$, $m \geq 1$, we have 
		\begin{equation}
			U_2^\circ(t,0)^{*}\ \Wick[1]{b_t}\, U_2^\circ(t,0) = \Wick[1]{\hat{b}_t}
		\end{equation}
		with 
		\begin{equation}\label{formule_symbole_distordu_propagation_quad}
			\hat{b}_t = \sum_{k=0}^{\ent{\frac{m}{2}}}\, \frac{1}{2^k k!}\, \Lambda_t^k (b_t \circ \upphi_t) \in \wickleq{m}\,.
		\end{equation}
	\end{theorem}
	
	\subsection[Complete asymptotic expansion]{Complete asymptotic expansion for the $\ppde$ model}\label{subsect.Proof_Thm_Pphi2}

	In this subsection, we prove our first main result Theorem \ref{thm.Pphi2_asymptotique_complete} in three steps. First, we differentiate the expression \eqref{premiere_apparition_Theta(t)} in a rigorous way. Then, we choose the parameters $\delta_\circ(t)$ and $b_j(t)$ appearing in \eqref{premiere_apparition_Theta(t)} to fit our needs. Eventually, we conclude the proof of Theorem \ref{thm.Pphi2_asymptotique_complete}.
	
	\medskip
	Let $\psi \in \mathcal{D}_{+,\infty} = \displaystyle\bigcap_{k \in \N} D\left(\nbop[1]^{k/2}\right)$ and $N \in \N$ be fixed. 
	
	\subsubsection{Differentiation of \eqref{premiere_apparition_Theta(t)}}\label{par.Diff_Pphi2}

	For $\delta_\circ \in C^1(\R,\R)$ and $\big(b_j(t)\big)_{j \in \N} \in \wickpol^\N$ to be chosen later, let us define, for all $t \in \R$, the quantity 
	
	\begin{equation}
		\Theta_N(t) = \e^{\rmi\frac{t}{\eps}H_\eps}\, \e^{\rmi\frac{\delta_\circ(t)}{\eps}}\, \weylop[\eps]{-\rmi \, \frac{\sqrt{2}}{\eps}\varphi_t}U_2^\circ(t,0) \, \psi_N(t)\,,
	\end{equation}
	
	where 
	\begin{itemize}
		\item $H_\eps$ is the spatially cutoff $\ppde$ Hamiltonian \eqref{def_spatially_cutoff_Pphi2_Hamiltonian},
		\item $\varphi_t$ is the mild solution of the classical field equation \eqref{eq_classique_avec_partie_libre_Pphi2}, 
		\item $U_2^\circ(t,0)$ is the quadratic propagator given by Corollary \ref{cor.dynamique_quad_modele_Pphi2}
	\end{itemize}
	and 
	\begin{equation}\label{ansatz_psiN_derivation_Theta(t)}
		\psi_N(t) = \sum_{j = 0}^N \eps^{\frac{j}{2}} \, \Wick[1]{b_j(t)}\psi\,.
	\end{equation}
	
	In order to differentiate $\Theta_N(t)$, using \eqref{relation_op_Weyl_second_quantifie}, we write 
	\begin{equation}\label{reecriture_Theta(t)}
		\Theta_N(t) = \e^{\rmi\frac{t}{\eps}H_\eps}\, \e^{-\rmi\frac{t}{\eps} \dG[\eps]{\omega}}\, \e^{\rmi\frac{\delta_\circ(t)}{\eps}}\, \weylop[\eps]{-\rmi\frac{\sqrt{2}}{\eps}\tilde{\varphi}_t} \tilde{U}_2^\circ(t,0) \, \psi_N(t)\,,
	\end{equation}
	where $\tilde{\varphi}_t = \e^{\rmi t \omega} \, \varphi_t$ and $\tilde{U}_2^\circ(t,0) = \e^{\rmi\frac{t}{\eps} \dG[\eps]{\omega}}\, U_2^\circ(t,0)$. Applying Corollary \ref{cor.dynamique_quad_modele_Pphi2} and the number estimates \eqref{number_estimate_Wick_poly}, we obtain that the derivative of $\tilde{U}_2^\circ(t,0) \, \psi_N(t)$, given by 
	\[
	\partial_t\left(\tilde{U}_2^\circ(t,0) \, \psi_N(t)\right) = -\frac{\rmi}{\eps}\, \Wick[\eps]{F_{\tilde{V}_2(t)}^{\conjug_t}} \tilde{U}_2^\circ(t,0) \, \psi_N(t) + \tilde{U}_2^\circ(t,0) \, \dot{\psi}_N(t)
	\] 
	with $F_{\tilde{V}_2(t)}^{\conjug_t} \in \wickpol$ defined by \eqref{def_Fv_quad_tordu_Pphi2}, belongs to $\mathcal{D}_{+,\infty}$. Moreover, since $\tilde{\varphi}_t \in C^1(\R,\hi)$, the quantity $\weylop[\eps]{-\rmi\displaystyle\frac{\sqrt{2}}{\eps}\tilde{\varphi}_t}$ is differentiable on $\mathcal{D}_{+,1}$ by \cite[Lemma 3.1]{GiVe79} (see also \cite[Lemma 6.1, point (iii)]{AmBr12}), with the formula 
	\begin{equation}\label{formule_derivation_op_Weyl}
		\rmi \, \eps\, \dv{t} \, \weylop[\eps]{-\rmi\frac{\sqrt{2}}{\eps} \tilde{\varphi}_t} = \weylop[\eps]{-\rmi\frac{\sqrt{2}}{\eps} \tilde{\varphi}_t} \Wick[\eps]{\Re\braket{\tilde{\varphi}_t}{\rmi \dot{\tilde\varphi}_t} + 2 \, \Re\braket{z}{\rmi \, \dot{\tilde\varphi}_t}}.
	\end{equation}
	
	It remains to differentiate the operators $\e^{\rmi\frac{t}{\eps}H_\eps}\, \e^{-\rmi\frac{t}{\eps} \dG[\eps]{\omega}}$.
	
	\begin{lemma}
		The operator $B_t \defegal \e^{\rmi\frac{t}{\eps}H_\eps}\, \e^{-\rmi\frac{t}{\eps} \dG[\eps]{\omega}}$ is strongly differentiable on $\mathcal{D}_{+,2n}$ with the formula 
		\begin{equation}\label{derivee_forte_propag_perturbation}
			\dot{B}_t = \frac{\rmi}{\eps}\, B_t\, \Wick[\eps]{F_{\tilde{V}(t)}^{\conjug_t}},
		\end{equation}
		where 
		\[
		F_{\tilde{V}(t)}^{\conjug_t}(z) = \sum_{j=0}^{2n} \braket{\frac{\tensn{(\e^{-\rmi t \omega}\,z + \e^{\rmi t \omega} \, \conjug z)}{j}}{\sqrt{j!}}}{V^{(j)}}.
		\] 
	\end{lemma}
	
	\begin{proof}
		The number estimates \eqref{number_estimate_Wick_poly} yield
		\[
		\norm{\Wick[\eps]{b}(\nbop[1]+1)^{-\frac{j}{2}}} \leq C(j)\, \eps^{\frac{j}{2}} \sum_{p+q = j} \norm{\tilde{b}_{p,q}}
		\] 
		for all $b = \displaystyle\sum_{p+q=j} b_{p,q} \in \wickhom{j}$. Applying this result to the symbols $F_{V^{(j)}}$ with $V^{(j)}$ given by \eqref{interaction_Pphi2_de_la_forme_Fv} and summing over $0 \leq j \leq 2n$, we obtain 
		\begin{equation}\label{baby_proposition_3.17_Ammari_Zerzeri}
			\norm{\Wick[\eps]{F_V} (\nbop[1]+1)^{-n}} \leq C(n) \norm{\liftop{\sqrt{\eps}}V}.
		\end{equation}
		
We conclude the proof using similar arguments as in \cite[Lemma 3.1]{GiVe79}.
	\end{proof}
	
	Now, using \eqref{estimee_norme_propag_quad_propagation_etats_coherents} and the estimate \eqref{controle_norme_methode_commut}, we obtain that the space $\mathcal{D}_{+,\infty}$ is preserved by the quadratic propagator $\tilde{U}_2^\circ(t,0)$ and by the Weyl operators. In particular, the vectors 
	\begin{equation}
		\chi_N(t) = \weylop[\eps]{-\rmi\frac{\sqrt{2}}{\eps}\tilde{\varphi}_t} \tilde{U}_2^\circ(t,0) \, \psi_N(t)
	\end{equation} 
	belong to $\mathcal{D}_{+,\infty}$, as well as their time derivatives. Therefore, we can differentiate the expression \eqref{reecriture_Theta(t)}: using Proposition \ref{prop.quantification_Wick_polynomiale} and the identities 
	\begin{equation}\label{auxiliary_identities_to_differentiate_Theta(t)}
		\braket{\tilde{\varphi}_t}{\rmi \dot{\tilde\varphi}_t} = \braket{\varphi_t}{\gatbar{}{F_V(\varphi_t)}} \text{ and } \braket{z}{\rmi \, \dot{\tilde\varphi}_t} = \braket{\e^{-\rmi \, \omega t}z}{\gatbar{}{F_V}(\varphi_t)}\,,
	\end{equation}
	we obtain 
	\begin{equation}\label{premiere_expression_derivee_Theta(t)}
		\dot{\Theta}_N(t) = \frac{\rmi}{\eps}\, \e^{\rmi\frac{t}{\eps} H} \, \e^{\rmi\frac{\delta_\circ(t)}{\eps}}\, \weylop[\eps]{-\rmi\frac{\sqrt{2}}{\eps}\varphi_t} \left(\Wick[\eps]{f_t} \, U_2^\circ(t,0) \, \psi_N(t) - \rmi \, \eps \, U_2^\circ(t,0) \dot{\psi}_N(t)\right),
	\end{equation}
	for all $\eps \in (0,1]$, where 
	\begin{equation}\label{expression_auxiliaire_dans_Wick_pour_derivation_Theta(t)}
		f_t(z) = F_V(z + \varphi_t) + \dot{\delta}_\circ(t) - \Re\braket{\varphi_t}{\gatbar{}{F_V}(\varphi_t)} - 2 \, \Re\braket{z}{\gatbar{}{F_V}(\varphi_t)} - F_{V_2(t)}(z)\,.
	\end{equation}	
	
	\subsubsection{Choice of the parameters}\label{par.Choice_Pphi2}
	
	We now turn to the choice of the parameters $\delta_\circ(t)$ and $b_j(t)$ appearing in \eqref{premiere_apparition_Theta(t)}. For all ${0 \leq j \leq 2n}$, $0 \leq \ell \leq j$ and $t \in \R$, let us define 
	\begin{equation}
		V_\ell^{(j)}(t) = \symop[\ell]\big(\bra{\tensn{(\varphi_t + \conjug \varphi_t)}{j-\ell}} \ot (\tensn{\1}{\ell})\big)V^{(j)} \in \tensymn{\hi}{\ell}
	\end{equation}
	and 
	\begin{equation}\label{def_potentiels_dev_Taylor_propagation_QFT_Pphi2}
		V_\ell(t) = \sum_{j=\ell}^{2n}\, \sqrt{\frac{\ell!}{j!}}\,\binom{j}{\ell}\, V_\ell^{(j)}(t)\,.
	\end{equation}
	Then, we have the expansion 
	\begin{equation}\label{developpement_Fv_autour_d'un_point_ordre_2_Pphi2}
		F_V(z+\varphi_t) = F_V(\varphi_t) + 2 \, \Re \braket{z}{\gatbar{}{F_V}(\varphi_t)} + F_{V_2(t)}(z) + \sum_{\ell = 3}^{2n} F_{V_\ell(t)}(z)\,.
	\end{equation}
	
	If we choose the function $\delta_\circ$ satisfying \eqref{conditions_fonction_torsion_propagation_etats_coherents}, the function $f_t$ defined by \eqref{expression_auxiliaire_dans_Wick_pour_derivation_Theta(t)} coincides with 
	\[
	F_{R_3(t)} \defegal \displaystyle\sum_{\ell = 3}^{2n} F_{V_\ell(t)}\,,
	\] 
	which leads to 
	\begin{multline}\label{expression_propre_derivee_Theta(t)_Pphi2}
	    \dot{\Theta}_N(t) = \frac{\rmi}{\eps}\, \e^{\rmi\frac{t}{\eps}H_\eps}\, \e^{\rmi\frac{\delta_\circ(t)}{\eps}}\, \weylop[\eps]{-\rmi\frac{\sqrt{2}}{\eps}\varphi_t}\times\\\left(\Wick[\eps]{F_{R_3(t)}} \, U_2^\circ(t,0) \, \psi_N(t) - \rmi \, \eps \, U_2^\circ(t,0) \dot{\psi}_N(t)\right)\,.
	\end{multline}
	
\medskip\noindent
Let us remark that homogeneity property \eqref{prop.homogeneite_quantif_Wick} yields \[\Wick[\eps]{F_{R_3(t)}} = \displaystyle\sum_{\ell = 3}^{2n} \eps^{\frac{\ell}{2}}\, \Wick[1]{F_{V_\ell(t)}}.\] Using the explicit expression \eqref{ansatz_psiN_derivation_Theta(t)} of $\psi_N(t)$ and Theorem \ref{thm.Breteaux_conjug_Wick_propag_quad} (see also \cite[Theorem 3.2]{Br12}), we thus obtain 
	\begin{align*}
		\Wick[\eps]{F_{R_3(t)}} \, U_2^\circ(t,0) \, \psi_N(t) = &\sum_{j = 0}^N \sum_{\ell = 3}^{2n} \eps^{\frac{j+\ell}{2}} \Wick[1]{F_{V_\ell(t)}} \, U_2^\circ(t,0) \, \Wick[1]{b_j(t)} \psi\\
		= &\sum_{k = 1}^{N+2n-2} \sum_{\substack{3 \leq \ell \leq 2n\\0 \leq j \leq N\\j+\ell = k+2}} \eps^{\frac{k+2}{2}} \, U_2^\circ(t,0) \, \Wick[1]{\hat{F}_{V_\ell(t)}} \Wick[1]{b_j(t)} \psi \\
		= &U_2^\circ(t,0) \sum_{k = 1}^{N+2n-2} \eps^{\frac{k+2}{2}} \sum_{j = 0}^{\min(N,k-1)} \Wick[1]{\hat{F}_{V_{k+2-j}(t)} \compwick[1] b_j(t)} \psi\,,
	\end{align*}
	where $a \compwick[1] b$ denotes the symbol of the operator $\Wick[1]{a}\Wick[1]{b}$ (see \eqref{formule_compo_Wick_poly}).  Hence, the last factor in \eqref{expression_propre_derivee_Theta(t)_Pphi2}, namely: $\Wick[\eps]{F_{R_3(t)}} \, U_2^\circ(t,0) \, \psi_N(t) - \rmi \, \eps \, U_2^\circ(t,0) \dot{\psi}_N(t)$, is equal to 
	\begin{align*}
	    &U_2^\circ(t,0) \sum_{k = 1}^{N+2n-2}\hskip-1pt \eps^{\frac{k+2}{2}}\hskip-4pt \sum_{j = 0}^{\min(N,k-1)} \Wick[1]{\hat{F}_{V_{k+2-j}(t)} \compwick[1] b_j(t)} \psi - \rmi\, \eps\, U_2^\circ(t,0)\, \sum_{k=0}^N \eps^{\frac{k}{2}} \Wick[1]{\partial_t b_k(t)}\, \psi\\
	    = &U_2^\circ(t,0) \left(\sum_{k = 1}^{N+2n-2} \eps^{\frac{k+2}{2}} \sum_{j = 0}^{\min(N,k-1)} \Wick[1]{\hat{F}_{V_{k+2-j}(t)} \compwick[1] b_j(t)} \psi - \rmi\, \sum_{k=0}^N \eps^{\frac{k+2}{2}} \Wick[1]{\partial_t b_k(t)}\, \psi\right)\,.
	\end{align*}
	Ordering the sum in terms of powers of $\eps$ yields 
	\begin{align*}
	    U_2^\circ(t,0) \Big(&- \rmi\, \eps\, \Wick[1]{\partial_t b_0(t)} \psi\\
	    &+ \sum_{k=1}^N \eps^{\frac{k+2}{2}} \bigg(- \rmi\, \Wick[1]{\partial_t b_k(t)} + \sum_{j = 0}^{k-1} \Wick[1]{\hat{F}_{V_{k+2-j}(t)} \compwick[1] b_j(t)} \bigg) \psi\\
	    &+ \sum_{k = N+1}^{N+2n-2} \eps^{\frac{k+2}{2}} \sum_{j = 0}^N \Wick[1]{\hat{F}_{V_{k+2-j}(t)} \compwick[1] b_j(t)} \psi\Big)\,.
	\end{align*}
	
	\smallskip\noindent 
	We now suppose $n \geq 2$ (see Remark \ref{rem.subquad_case_Pphi2} below for the sub-quadratic case $n \leq 1$). Our aim is to choose the polynomial symbols $b_k(t)$ such that the first term and the first sum vanish. Picking $\big(b_k(t)\big)_{k\in \N}$ defined by 
	\begin{equation}\label{relation_de_recurrence_verifiee_les_polynomes_de_Wick_de_l'asymptotique_complete}
		\left\{\begin{array}{l}
			b_0 \equiv 1\\
			- \rmi\, \partial_t b_k(t) + \displaystyle\sum_{j=0}^{k-1} \hat{F}_{V_{k+2-j}(t)} \compwick[1] b_j(t) = 0 \text{ with } b_k(0) = 0\,,\quad \forall\, k\in \N^*\,,
		\end{array}\right.
	\end{equation}
	namely, 
	\begin{equation}\label{expression_explicite_polynomes_de_Wick_de_l'asymptotique_complete}
	\left\{
	\begin{array}{l}
	     b_0 \equiv 1\\
	     b_k(t) = -\rmi\, \displaystyle\sum_{j=0}^{k-1} \int_0^t \hat{F}_{V_{k+2-j}(s)} \compwick[1] b_j(s)\, \dd s\,,\quad \forall\,  k\in \N^*\,,
	\end{array}
	\right.
	\end{equation}
		we finally obtain 
	\begin{multline}\label{expression_finale_derivee_Theta(t)_Pphi2}
	 \dot{\Theta}_N(t) = \frac{\rmi}{\eps} \, \e^{\rmi\frac{t}{\eps}H_\eps}\, \e^{\rmi\frac{\delta_\circ(t)}{\eps}}\,\weylop[\eps]{-\rmi\frac{\sqrt{2}}{\eps}\varphi_t}\times\\
	 \sum_{k=N+1}^{N+2n-2} \eps^{\frac{k+2}{2}} \sum_{j = 0}^N \Wick[1]{F_{V_{k+2-j}(t)}} \, U_2^\circ(t,0) \, \Wick[1]{b_j(t)} \psi\,.
	\end{multline}
	
	\subsubsection{End of the proof}\label{par.End_Proof_Thm_Pphi2}
	
	We can eventually prove the estimate \eqref{asymptotique_complete_CSdyn_Pphi2}. Let us note that the left-hand side of the estimate is equal to $\norm{\Theta_N(t) - \Theta_N(0)}$. Hence, this quantity is dominated by $\displaystyle\int_0^{\abs{t}} \norm{\dot{\Theta}_N(s)} \dd s$, which is controlled by 
	\begin{equation}
		\eps^{\frac{N+1}{2}} \sum_{k=N+1}^{N+2n-2} \sum_{j = 0}^N \int_0^{\abs{t}} \norm{\Wick[1]{F_{V_{k+2-j}(s)}} \, U_2^\circ(s,0) \Wick[1]{b_j(s)} \psi} \dd s
	\end{equation}
	in virtue of \eqref{expression_finale_derivee_Theta(t)_Pphi2}. Applying the estimates \eqref{number_estimate_Wick_poly} and \eqref{estimee_norme_propag_quad_propagation_etats_coherents} with $\eps = 1$, we obtain 
	\begin{equation}\label{baby_lemma_4.9_Ammari_Zerzeri}
		\norm{\Wick[1]{F_{V_{k+2-j}(s)}} \, U_2^\circ(s,0) \Wick[1]{b_j(s)} \psi} \leq C'(s,k,j,\psi)
	\end{equation}
	for every $k,j$, where 
	\begin{multline}\label{cst_auxiliaire_asymptotique_complete}
	    C'(s,k,j,\psi) = C(k,j) \norm{V_{k+2-j}(s)} \times\\\norm{\Wick[1]{b_j(s)} \psi}_{+,(k+2-j)} \, \e^{\sqrt{2} \, (k+2-j)\, 3^{\frac{k+2-j}{2}} \int_0^{\abs{s}} \norm{V_2(\tau)} \dd \tau}
	\end{multline}
	is continuous with respect to $s$. Let us note that the finiteness of $C'(s,k,j,\psi)$ merely requires $L^2$-regularity of the mild solution $\varphi_s$ of \eqref{eq_classique_avec_partie_libre_Pphi2}.
	
	Finally, the left-hand side of \eqref{asymptotique_complete_CSdyn_Pphi2} is dominated by \[\eps^{\frac{N+1}{2}} \sum_{k=N+1}^{N+2n-2} \sum_{j = 0}^N \int_0^{\abs{t}} C'(s,k,j,\psi)\, \dd s = C(t,\psi,N)\, \eps^{\frac{N+1}{2}}\,,\]
	which concludes the proof of Theorem \ref{thm.Pphi2_asymptotique_complete}.
	\hfill$\square$
	
	\begin{remark}\label{rem.subquad_case_Pphi2}
	 Let us note that $V_{k+2-j}(t) = 0$ for $j \leq k - (2n-1)$, where $2n$ is the degree of the polynomial \eqref{def_polynome_modele_Pphi2}. Hence, in the sub-quadratic case $n \leq 1$, we have $b_k(t) \equiv 0$ for all $k \in \N^*$, $\psi_N(t) \equiv \psi$ and $\dot{\Theta}_N(t) \equiv 0$ for all $N \in \N$, which leads to the exact formula \eqref{formule_exacte_propag_quad_etats_coherents}.
	\end{remark}
	
	\section{Analytic interactions}\label{sect.analytic_interactions}
	
	In this section, we focus on self-interacting bosonic quantum field Hamiltonians with analytic interactions. In Subsection \ref{subsect.Non_poly}, we present roughly how to extend the Wick quantization to non-polynomial symbols along the same lines as the work of the first and third authors, see \cite{AmZe14}. As explained in Subsection \ref{subsect.Main_results}, the outline of the proof is exactly the same as that presented in the previous section for $\ppde$ Hamiltonians. However, some tools need to be adapted to the analytic interaction framework. We provide a lower bound of the lifetime of a maximal solution of the classical field equation associated with the corresponding Hamiltonian, see \eqref{borne_inf_temps_d'existence_solution_mild_equation_eikonale} in Subsection \ref{subsect.The_classical_field_ equation}. We point out that the argument given in the proof of \cite[Theorem~4.1]{AmZe14} for global well-posedness is incorrect. 
	The local existence and uniqueness of the propagator describing the associated quadratic dynamics are proven in Subsection \ref{subsect.The_quadratic_dynamic}. 
	As explained at the end of Subsection \ref{subsect.Main_results}, we need to generalize the number estimates \eqref{number_estimate_Wick_poly}. This is done in Subsection \ref{subsect.some_auxiliary} by introducing an analytic function in the variable $(\nbop[1]+1)$ with suitably chosen, exponentially decaying coefficients. After all these preparations, the proof of Theorem \ref{thm.general_asymptotique_complete} proceeds in the same way as for the $\ppde$ model, see Subsection \ref{subsect.End_Proof_Thm}.
	
	\subsection{Non-polynomial Wick operators}\label{subsect.Non_poly}
	
	We start by giving a rigorous meaning to the expression $\Wick[\eps]{F_V}$ appearing in the Hamiltonian \eqref{def_hamiltonien_asymptotique_complete_etat_coherent_QFT},
	since the function $F_V$ given by \eqref{def_potentiel_interaction_Fv} is possibly non-polynomial. 
	Let us define
	\begin{equation}\label{def_core_Wick_poly}
		\wickcore = \Vect\{\weylop[\eps]{u}\Omega,\quad u \in \hi_0\}\,.
	\end{equation}
	This space is dense in $\focks$, see \cite[Lemma 3.4]{AmZe14}.
	
\medskip\noindent
We now define $\FVclass$ as the space of maps $F : \hi \to \C$ of the form 
\[
F(z) = \sum_{j = 0}^\infty \braket{\frac{\tensn{(z+\conjug z)}{j}}{\sqrt{j!}}}{V^{(j)}} = F_V(z) \text{ where } V = \Big(V^{(j)}\Big)_{j \in \N} \in \focks\,.
\] 
Endowed with the inner product 
\[
\braket{F_{V_1}}{F_{V_2}}_{\FVclass} = \braket{V_1}{V_2}_{\focks}\,,
\] 
$\FVclass$ is a Hilbert space which is isometrically isomorphic to $\focks$. Note that the Wick symbol of the $\ppde$ interaction term given by \eqref{def_spatially_cutoff_interaction_Pphi2} belongs to $\FVclass$ in virtue of \eqref{interaction_Pphi2_de_la_forme_Fv}.
	
\medskip
	For each $V \in \focks$, the operator 
	\begin{equation}\label{def_quantification_Wick_non_polynomiale_sur_petit_coeur}
		\begin{array}{ccccc}
			\Wick[\eps]{F_V} & : & \wickcore & \longrightarrow & \focks \\
			& & \weylop[\eps]{u}\Omega & \longmapsto & \weylop[\eps]{u}\liftop{\sqrt{\eps}}V
		\end{array}
	\end{equation}
	is densely defined, and it is closable since $\Wick[\eps]{F_V} \subset \Wick[\eps]{(\bar{F_V})^{*}}$. The closure of this operator, still denoted by $\Wick[\eps]{F_V}$, is called the Wick quantization of the map $F_V \in \FVclass$, and coincides with the definition \eqref{def_quantification_Wick_poly_produit_tensoriel_symetrique} if $F_V \in \FVclass \cap \wickpol$.
	
	\begin{prop}[{\cite[Proposition 3.11]{AmZe14}}]\label{prop.Wick_non_poly_est_essentiellement_autoadjoint}
		Assume that $V \in \focks$ satisfies $\liftop{\conjug}V = V$. Then, the operator $\Wick[\eps]{F_V}$ is essentially self-adjoint. 
	\end{prop}
	
	\begin{proof}
		In the wave representation, the operator $\Wick[\eps]{F_V}$ corresponds to the multiplication by $\mathcal{R}(\liftop{\sqrt{\eps}}V)$ in view of \eqref{formule_clef_pour_definition_Wick_non_poly}. Under the assumption $\liftop{\conjug}V = V$, the function $\mathcal{R}(\liftop{\sqrt{\eps}}V)$ is real-valued by \eqref{Fock_conjugue_donne_L2_conjugue_via_rpz_ondulatoire}, so the associated multiplication operator is essentially self-adjoint. 
	\end{proof}
	
	We now state a fundamental inequality which follows from the well-known hypercontractive estimates, see for example \cite[Theorem I.17]{Si74}. Recall that the Fock space $\focks$ is unitarily equivalent to some space $L^2(Q,\dd\mu)$ where $\mu$ is a probability measure by Theorem \ref{thm.rpz_ondulatoire_espace_Fock}.
	
	\begin{theorem}\label{thm.hypercontractive_estimates_Simon}
	 Let $1 < p \leq q < +\infty$ and $0 < c_0 \leq \sqrt{\dfrac{p-1}{q-1}}$. For all $f \in L^2(Q,\dd\mu)$, we have 
	 \begin{equation}\label{hypercontractive_estimates_Simon}
	 \norm{\liftop{c_0}f}_{L^q(Q,\dd\mu)} \leq \norm{f}_{L^p(Q,\dd\mu)}\,.
	 \end{equation}
	\end{theorem}
	
	We deduce from this estimate that, for all $V \in \focks$ and $\psi \in D\left(\liftop{\sqrt{3}}\right)$, we have 
	\begin{equation}\label{hypercontractive_estimate_in_Fock_space}
	 \norm{\Wick[\eps]{F_V}\psi} \leq \norm{V} \norm{\liftop{\sqrt{3}}\psi}
	 \end{equation}
	for all $\eps \in (0,1/3]$, see \cite[Lemma 3.14]{AmZe14}. In fact, $\focksalg$ is a core for $\Wick[\eps]{F_V}$ by \cite[Proposition 3.15]{AmZe14} and, if we define $V_\kappa = (V^{(0)}, \dots, V^{(\kappa)},0,\dots)$ for all $\kappa \in \N$, the estimate \eqref{hypercontractive_estimate_in_Fock_space} yields 
	\begin{equation}
	 \norm{\Big(\Wick[\eps]{F_{V_\kappa}}-\Wick[\eps]{F_V}\Big)\psi} \leq \norm{V_\kappa-V}\norm{\liftop{\sqrt{3}}\psi}
	\end{equation}
	for all $\psi \in D\left(\liftop{\sqrt{3}}\right)$. This enables us to extend two properties of the polynomial symbols to interactions in $\FVclass$. First, we can generalize the translation property \eqref{translation_Wick_poly}.
	
	\begin{lemma}[{\cite[Lemma 3.18]{AmZe14}}]
	 Let $V \in D\left(\liftop{\sqrt{2}}\right)$. For all $u \in \hi$, we have $F_V(\bullet + u) \in \FVclass$ and 
	 \begin{equation}\label{translation_Wick_non_poly}
	 \weylop[\eps]{-\rmi\frac{\sqrt{2}}{\eps}u}^{*} \Wick[\eps]{F_V} \weylop[\eps]{-\rmi\frac{\sqrt{2}}{\eps}u} = \Wick[\eps]{F_V\big(\cdot+u\big)}
	 \end{equation}
	 for all $\eps \in (0,1/3]$.
	 
	\end{lemma}

	\bigskip
	Furthermore, we can extend the estimate \eqref{baby_proposition_3.17_Ammari_Zerzeri} to non-polynomial interactions.
	
	\begin{prop}[{\cite[Proposition 3.17]{AmZe14}}]\label{prop.3.17_ammari_zerzeri}
		Let $S(x) = \displaystyle\sum_{k=0}^\infty a_k x^k$ be an entire function on $\C$ such that $a_k > 0$ for all $k \in \N$, and denote by 
		\begin{equation}\label{fct_analytique_op_nombre}
		    \mathcal{\bf S} = \sum_{k=0}^\infty a_k (\nbop[1]+1)^k\,.
		\end{equation}
		Let $V = \Big(V^{(j)}\Big)_{j \in \N} \in \focks$ such that $\liftop{\conjug}V = V$. For all $\lambda_1 > 8e$, there exists $C > 0$ such that the inequality 
		\begin{equation}
			\norm{\Wick[\eps]{F_V} \psi} \leq 2 \norm{\liftop{\sqrt{\eps}}V} \norm{\psi} + C\left(\sum_{k=0}^{+\infty}\, \frac{(\lambda_1 \eps)^k}{a_{k+2}} \norm{V^{(k)}}^2\right)^{1/2} \norm{\mathcal{\bf S}^{1/2}\, \psi}
		\end{equation}
		holds whenever the right-hand side is finite.
	\end{prop}
	
	Eventually, we give a crucial theorem of essential self-adjointness. The general version of this result, stated by Segal in \cite{Se70}, is the outcome of remarkable works from Nelson, Glimm, Jaffe, Simon and H\o egh-Krohn, among many others.
	
	\begin{theorem}\label{thm.hamiltonien_CSdyn_essentiellement_autoadjoint}
		Under the assumptions \ref{hyp.A1} and \ref{hyp.A2}, the Hamiltonian \eqref{def_hamiltonien_asymptotique_complete_etat_coherent_QFT} is essentially self-adjoint on $\focksalg$ with domain $D(\dG[\eps]{A}) \cap D(\Wick[\eps]{F_V})$.
	\end{theorem}
	
	\subsection{The classical field equation}\label{subsect.The_classical_field_ equation}
	
	In the case of the $\ppde$ model, we showed the existence of a unique global mild solution of the classical field equation \eqref{eq_classique_avec_partie_libre_Pphi2} with initial data $\varphi_0 \in L^2(\R,\dd k)$, see Corollary \ref{cor.existence_sol_mild_global_L2_Pphi2}. In the same way, we can establish the existence of a unique maximal mild solution for \eqref{eq_classique_avec_partie_libre}, that is, a function $\varphi_t \in C^0(I,\hi)$ satisfying 
	\begin{equation}\label{eq_classique_avec_partie_libre_formulation_integrale}
		\varphi_t = \e^{-\rmi t A}\, \varphi_0 - \rmi \, \int_0^t \e^{-\rmi(t-s)A} \, \gatbar{}{F_V}(\varphi_s) \, \dd s\,,
	\end{equation} 
	where $\gatbar{}{F_V}(\varphi_t)$ is defined by \eqref{def_gatbar_Fv}. This result has been stated in \cite[Theorem 4.1]{AmZe14} with $I = \R$, however, it seems to us that the provided argument based on Bihari's lemma \cite{Bi56} is not sufficient to prove that the solutions are globally defined. Therefore, we revisit the proof and restrict ourselves to maximal solutions.
	
	\medskip
	Let $\varphi_0 \in \hi$ and $V \in D\left(\liftop{\sqrt{2}}\right)$, $V \neq 0$ be fixed. For all $u,v \in \hi$, we have 
	\begin{equation}\label{inegalite_contraction_dzbarFv}
		\normhi{\gatbar{}{F_V}(u)-\gatbar{}{F_V}(v)} \leq 2\, \norm{V}\, \mathrm{g}\big(\max(\normhi{u},\normhi{v})\big) \normhi{u-v}\,,
	\end{equation} 
	where $\mathrm{g}(x)= \sqrt{\displaystyle\sum_{k=0}^\infty\, \displaystyle\frac{4^k(k+1)(k+2)}{k!}\, x^{2k}}$ defines a positive, non-decreasing function on $\R_+$. Recall that $\normhi{\cdot}$ denotes the norm on $\hi$. Hence, the vector field 
	\[X : (t,u) \in \R \times \hi \longmapsto \e^{\rmi t A} \, \gatbar{}{F_V}\big(\e^{-\rmi t A}u\big)
	\]
	is continuous and locally Lipschitz in the second variable, which provides the existence of a unique maximal solution $\tilde{\varphi}_t \in C^1(I,\hi)$ of the non-autonomous Cauchy problem:
	\begin{equation}\label{eq_classique_sans_partie_libre}
	 \left\{\begin{array}{l}
			\rmi \, \dot{\tilde\varphi}_t = \e^{\rmi t A} \, \gatbar{}{F_V}\big(\e^{-\rmi t A}\tilde{\varphi}_t\big)\\
			\eval{\tilde{\varphi}_t}_{t = 0} = \varphi_0 \in \hi\,.
		\end{array}\right. 
	\end{equation}
	Furthermore, the integral formulation 
	\begin{equation}\label{eq_classique_sans_partie_libre_formulation_integrale}
	 \tilde{\varphi}_t = \varphi_0 - \rmi \, \int_0^t \e^{\rmi sA} \, \gatbar{}{F_V}\big(\e^{-\rmi sA} \tilde{\varphi}_s\big) \, \dd s
	\end{equation}
	gives, for all $t \in I \cap \R_+$, 
	\[
	\normhi{\tilde{\varphi}_t} \leq \normhi{\varphi_0} + \norm{V} \int_0^t \sqrt{f\Big(4\,\normhi{\tilde\varphi_s}^2\Big)}\, \dd s\,,
	\]
	where $f(x) = (1+x)\,\e^x$. Thus, considering the function $G$ defined by 
	\[
	G(u) = \frac{1}{\norm{V}} \int_0^u \frac{1}{\sqrt{f(4x^2)}} \, \dd x = \frac{1}{\norm{V}} \int_0^u \frac{\e^{-2 x^2}}{\sqrt{1+4x^2}} \, \dd x\,,
	\] 
	Bihari's lemma \cite{Bi56} yields \[\normhi{\varphi_t} \leq G^{-1}\big(G(\normhi{\varphi_0}) + t\big)\] for all $t \in I \cap \R_+$ such that $G(\normhi{\varphi_0}) + t \in \mathrm{Dom}(G^{-1})$, that is, for all $t \in I \cap [0,T_0)$ with 
	\begin{equation}\label{borne_inf_temps_d'existence_solution_mild_equation_eikonale}
		T_0 = T(\normhi{\varphi_0}) = \frac{1}{\norm{V}} \int_{\normhi{\varphi_0}}^{+\infty}
 \displaystyle\frac{\e^{-2x^2}}{\sqrt{1+4x^2}} \, \dd x \in (0,+\infty)\;.
	\end{equation}
	Eventually, we get $[0,T_0) \subset I$, and $(-T_0,0] \subset I$ by a time reversal argument. Now, setting ${\varphi_t = \e^{-\rmi t A}\, \tilde{\varphi}_t \in C^0(I,\hi)}$, we are able to restate \cite[Theorem 4.1]{AmZe14}.
	
	\begin{theorem}\label{thm.existence_solution_mild_eq_classique_avec_partie_libre}
		For any initial data $\varphi_0 \in \hi$, the Cauchy problem \eqref{eq_classique_avec_partie_libre} admits a unique maximal mild solution $\varphi_t \in C^0(I,\hi)$. Moreover, we have $(-T_0,T_0) \subset I$ where $T_0$ is given by \eqref{borne_inf_temps_d'existence_solution_mild_equation_eikonale}.
	\end{theorem}
	
	\subsection{The quadratic dynamics}\label{subsect.The_quadratic_dynamic}
	
	First, following Corollary \ref{cor.dynamique_quad_modele_Pphi2}, we obtain the result stated below. Let us recall that $I$ is the lifespan of the mild solutions in Theorem \ref{thm.existence_solution_mild_eq_classique_avec_partie_libre}.
	
	\begin{theorem}[{\cite[Corollary 4.6]{AmZe14}}]\label{thm.dynamique_quad_applique_à_notre_potentiel}
		Let $t \in I \longmapsto \varphi_t$ be the mild solution of the classical field equation \eqref{eq_classique_avec_partie_libre} given by Theorem \ref{thm.existence_solution_mild_eq_classique_avec_partie_libre}. For all $t \in I$, let us define 
		\begin{equation}\label{def_time_dependent_quadratic_potential}
			V_2(t) = \frac{1}{\sqrt{2}} \sum_{j=2}^\infty \frac{j\,(j-1)}{\sqrt{j!}} \symop[2]\big(\bra{\tensn{(\varphi_t + \conjug \varphi_t)}{j-2}} \ot (\tensn{\1}{2})\big)V^{(j)} \in \tensymn{\hi}{2}\,,
		\end{equation}
		and 
		\begin{equation}
			F_{\tilde{V}_2(t)}^{\conjug_t}(z) = \braket{\frac{\tensn{(\e^{-\rmi t A}z+\e^{\rmi t A}\,\conjug z)}{2}}{\sqrt{2}}}{V_2(t)}\,.
		\end{equation}
		Then, for all initial data $u_0 \in \mathcal{D}_{+,k}$, $k \in \N^*$, the Cauchy problem 
		\begin{equation}
			\left\{\begin{array}{l}
				\rmi\, \eps\, \partial_t u_t = \Wick[\eps]{F_{\tilde{V}_2(t)}^{\conjug_t}} u_t\\
				\eval{u_t}_{t=0} = u_0 
			\end{array}\right.
		\end{equation}
		admits a unique unitary propagator $(\tilde{U}_2(t,s))_{t,s\, \in\, I}$. Furthermore, if we define 
		\begin{equation}
			U_2(t,s) = \e^{-\rmi\frac{t}{\eps} \dG[\eps]{A}}\, \tilde{U}_2(t,s)\, \e^{\rmi\frac{s}{\eps} \dG[\eps]{A}}\,, 
		\end{equation}
		the map $t \longmapsto U_2(t,s)\, u_s$ is a mild solution of 
		\begin{equation}
			\left\{\begin{array}{l}
				\rmi\, \eps\, \partial_t u_t =\Big(\dG[\eps]{A} + \Wick[\eps]{F_{V_2(t)}}\Big) u_t\\
				\eval{u_t}_{t=s} = u_s\,.
			\end{array}\right.
		\end{equation}
	\end{theorem}
	
	\smallskip\noindent
	In the case of the $\ppde$ model, we obtained a bound on $\norm{U_2(t,s)}_{\mathcal{L}(\mathcal{D}_{+,k})}$ for all $k \in \N^*$, see \eqref{estimee_norme_propag_quad_propagation_etats_coherents}. We give below a refinement of this bound that will be useful to prove Theorem \ref{thm.general_asymptotique_complete}.
	
	\begin{prop}[{\cite[Proposition 4.5]{AmZe14}}]
		With the notations of Theorem \ref{thm.dynamique_quad_applique_à_notre_potentiel}, for all $\lambda > 1$, there exists $c > 0$ such that, for all $k \in \N^*$, the inequality
		\begin{multline}\label{proposition_4.5_Ammari_Zerzeri_propagation_etats_coherents}
		 \norm{(\nbop[1]+1)^{k/2} \, U_2(t,s)\, \psi} \leq \e^{\sqrt{2} \, k\, \lambda^k \, \abs{\int_s^t \norm{V_2(\tau)} \dd \tau}} \times\\ 
		 \left(k\,c^k\, \abs{\int_s^t \norm{V_2(\tau)} \dd \tau} \norm{\psi}^2 + \norm{(\nbop[1]+1)^{k/2}\, \psi}^2\right)^{1/2}
		\end{multline}
		holds for any $\psi \in \mathcal{D}_{+,k}$.
	\end{prop}

 Eventually, Theorem \ref{thm.Breteaux_conjug_Wick_propag_quad} still holds with $\omega$ replaced by the self-adjoint operator fulfilling~\ref{hyp.A1}.
	
	\subsection{Some auxiliary results}\label{subsect.some_auxiliary}
	
	In order to prove Theorem \ref{thm.general_asymptotique_complete}, we need to strongly differentiate the expression \eqref{premiere_apparition_Theta(t)} on a dense subspace. In the non-polynomial case, the main difficulty is to identify a convenient domain for the strong differentiation of the operators 
	\[
	B_t \defegal \e^{\rmi\frac{t}{\eps}H_\eps}\, \e^{-\rmi\frac{t}{\eps} \dG[\eps]{A}}
	\]
	appearing in \eqref{reecriture_Theta(t)}, while preserving the invariance of such domains by the Weyl operators and the unitary propagator in Corollary \ref{cor.dynamique_quad_modele_Pphi2}. In this case, we observe that the number estimates \eqref{number_estimate_Wick_poly} do not apply anymore. We start by stating some auxiliary results.
	
	\begin{prop}[\cite{AmZe14}]\label{prop.differentiabilite_premier_op_unitaire_dans_Theta(t)}
		Let $V = \Big(V^{(j)}\Big)_{j \in \N} \in \focks$ such that $\liftop{\conjug}V = V$ fulfilling \ref{hyp.dom_superexpo} with constants $\alpha > 0$, $\lambda > 1$. Let 
		\begin{equation}\label{serie_entiere_coeff_exponentiels_cas_analytique}
			S_{\alpha_0}(x) = \sum_{k = 0}^\infty \e^{-\alpha_0 \, \lambda^k} x^k 
		\end{equation} with $0 < \alpha_0 < \displaystyle\frac{2\alpha}{\lambda^2}$. Then, there exists $C' > 0$ such that, for all $\eps \in (0,1]$, 
		\begin{equation}\label{Wick_Fv_est_borne_sur_domaine_racine_serie_entiere}
			\norm{\Wick[\eps]{F_V} \mathcal{\bf S}_{\alpha_0}^{-1/2}} \leq C' \norm{\e^{\alpha \lambda^{\nbop[1]}}\liftop{\sqrt{\eps}}V}\,,
		\end{equation}
		where $\mathcal{\bf S}_{\alpha_0}$ is defined by \eqref{fct_analytique_op_nombre}.
		In particular, $B_t \defegal\e^{\frac{\rmi t}{\eps}H_\eps}\e^{-\rmi\frac{t}{\eps}\dG[\eps]{A}}$ in differentiable on ${\mathcal{D}_{\alpha_0} = D\left(\mathcal{\bf S}_{\alpha_0}^{1/2}\right)}$ with the formula 
		\begin{equation}\label{derivee_forte_propag_perturbation_bis}
			\dot{B}_t = \frac{\rmi}{\eps}\, B_t \, \Wick[\eps]{F_{V(t)}^{\conjug_t}}\,.
		\end{equation}
	\end{prop}
	
	\begin{proof}
		Let $0 < \alpha_0 < \displaystyle\frac{2 \alpha}{\lambda^2}$. Applying Proposition \ref{prop.3.17_ammari_zerzeri} with $a_k = \e^{-\alpha_0 \, \lambda^k}$ and $\lambda_1 > 8\e$ fixed, we obtain, for all $\psi \in \mathcal{D}_{\alpha_0}$ and $\eps \in (0,1/3]$, 
		\begin{align*}
			\norm{\Wick[\eps]{F_V}\psi} \leq &\,2 \norm{\liftop{\sqrt{\eps}}V} \norm{\psi} \\
			\phantom{\leq \, 2\,}&+ C \left(\sum_{j = 0}^\infty \e^{(\alpha_0 \, \lambda^2) \lambda^j} (\lambda_1 \eps)^j \norm{V^{(j)}}^2\right)^{1/2} \norm{\mathcal{\bf S}_{\alpha_0}^{1/2} \, \psi}\\
			\leq &\,2 \norm{\liftop{\sqrt{\eps}}V} \norm{\psi} + C_1 \norm{\e^{\alpha \lambda^{\nbop[1]}} \liftop{\sqrt{\eps}} V} \norm{\mathcal{\bf S}_{\alpha_0}^{1/2} \, \psi}\,.
		\end{align*}
		Since the operator $\mathcal{\bf S}_{\alpha_0}^{-1/2} : \focks \longrightarrow \mathcal{D}_{\alpha_0}$ is bounded, we obtain \eqref{Wick_Fv_est_borne_sur_domaine_racine_serie_entiere} by setting ${\psi' = \mathcal{\bf S}_{\alpha_0}^{1/2} \, \psi}$. The proof of the differentiability of $B_t$ is very similar to the one of \cite[Lemma 3.1]{GiVe79}, using \eqref{Wick_Fv_est_borne_sur_domaine_racine_serie_entiere}. 
	\end{proof}

	\begin{prop}\label{prop.op_Weyl_preserve_presque_domaine_racine_serie_entiere}
		For all $u \in \hi$ and $0 < \alpha_1 < \alpha_0$, the operator 
		\[
		\mathcal{\bf S}_{\alpha_0}^{1/2} \, \weylop[\eps]{u} \mathcal{\bf S}_{\alpha_1}^{-1/2}
		\] is bounded.
	\end{prop}
	
	\begin{proof}
		Let $u \in \hi$ be fixed. Let $n \in \N$ and $\psi^{(n)} \in \tensymn{\hi}{n}$. For all $k, \ell \in \N$, we have 
		\[
		\frac{1}{k!} \norm{(\nbop[1]+1)^\ell \fieldop[\eps]{u}^k \psi^{(n)}} \leq \frac{\norm{\psi^{(n)}}}{n!}\frac{\sqrt{(n+k)!}}{k!} \Big(\sqrt{2 \eps}\, \normhi{u}\Big)^k (n+k+1)^{\frac{\ell}{2}}\,,
		\] 
		whence 
		\[
		\sum_{k=0}^\infty \frac{1}{k!} \norm{\mathcal{\bf S}_{\alpha_0}\, \fieldop[\eps]{u}^k\, \psi^{(n)}} < +\infty\,.
		\] 
		We deduce that $\wickbigcore \subset D\left(\mathcal{\bf S}_{\alpha_0}\right)$. This enables us to write, for all $\psi \in \wickbigcore$, 
		\begin{align*}
			\norm{\mathcal{\bf S}_{\alpha_0}^{1/2} \, \weylop[\eps]{u} \psi}^2 = &\braket{\weylop[\eps]{u} \psi}{\mathcal{\bf S}_{\alpha_0} \weylop[\eps]{u} \psi}\\
			= &\sum_{k = 0}^\infty \e^{-\alpha_0 \, \lambda^k} \norm{
			(\nbop[1]+1)^{k/2} \, \weylop[\eps]{u} \psi}^2\\
			\leq &\sum_{k = 0}^\infty \e^{-\alpha_0 \, \lambda^k}\, M(\normhi{u})\, \normhi{u}^{\frac{k}{2}} 2^{\frac{k^2}{4}}\, \norm{(\nbop[1]+1)^{k/2} \, \psi}^2\\
			\leq &C\Big(\alpha_0,\alpha_1,\lambda,\normhi{u}\Big) \sum_{k = 0}^\infty \e^{-\alpha_1 \lambda^k} \norm{(\nbop[1]+1)^{k/2} \, \psi}^2\\
			= &C\Big(\alpha_0,\alpha_1,\lambda,\normhi{u}\Big) \norm{\mathcal{\bf S}_{\alpha_1}^{1/2}\, \psi}^2\,,
		\end{align*}
		where we have used \eqref{controle_norme_methode_commut} for the first inequality. This inequality extends to $\psi \in \mathcal{D}_{\alpha_1}$, which concludes the proof.
	\end{proof}
	
	\begin{prop}[{\cite[proof of Lemma~4.9]{AmZe14}}]\label{propagateur_quad_envoie_Fock_algebrique_dans_domaine_racine_serie_entiere}
		For all $\psi \in \focksalg$ and $t \in I$, we have ${U_2(t,0) \, \psi \in \mathcal{D}_{\alpha_0}}$ for all $\alpha_0 > 0$. Furthermore, we have 
		\begin{equation}\label{image_fock_algebrique_propag_quad_inclus_dans_domaine_puissance_serie_entiere}
			\norm{\mathcal{\bf S}_{\alpha_0}^{1/2} \, U_2(t,0) \, \psi} \leq \left(c \, \mathrm{g}_t'(c)\, \int_0^{\abs{t}} \norm{V_2(\tau)} \dd \tau \cdot \norm{\psi}^2 + \norm{\mathrm{g}_t(\nbop[1]+1)^{1/2}\, \psi}^2\right)^{1/2},
		\end{equation}
		with $c > 0$,
		\begin{equation}\label{def_serie_entiere_dependant_du_temps_pour_majoration_cruciale_asymptotique_complete_QFT}
			\mathrm{g}_t(r) = \sum_{k = 0}^\infty \e^{-\alpha_0 \, \lambda^k} \e^{2 \, \sqrt{2} \, k\, \lambda_0^k \int_0^{\abs{t}} \norm{V_2(\tau)} \dd \tau}\, r^k\,,\quad 1 < \lambda_0 < \lambda\,,
		\end{equation} and 
		\[
		\mathrm{g}_t'(r) = \dv{r}\, \mathrm{g}_t(r)\,.
		\]
	\end{prop}
	
	\begin{proof}
		This result follows from the estimate \eqref{proposition_4.5_Ammari_Zerzeri_propagation_etats_coherents} applied with $1 < \lambda_0 < \lambda$.
		
\medskip\noindent
		The right-hand side of \eqref{propagateur_quad_envoie_Fock_algebrique_dans_domaine_racine_serie_entiere} is finite: indeed, if $\psi \in \displaystyle\bigoplus_{n=0}^m \tensymn{\hi}{n}$, writing ${\mathrm{g}_t(r) = \displaystyle\sum_{k=0}^\infty a_k(t)\, r^k}$, we have 
		\[
		\sum_{k=0}^\infty a_k(t) \norm{(\nbop[1]+1)^{k/2} \, \psi}^2 \leq \sum_{k=0}^\infty a_k(t)\, m^k\, \norm{\psi}^2 < + \infty\,.
		\] 
	\end{proof}
	
	\begin{cor}[{\cite[Lemma 4.9]{AmZe14}}]
	 For all $\psi \in \focksalg$ and $t \in I$, the vector $U_2(t,0) \, \psi$ belongs to $D\left(\Wick[\eps]{F_V}\right)$ with an estimate of the form 	
	 \begin{multline}\label{inegalite_cruciale_pour}
			\norm{\Wick[\eps]{F_V} \, U_2(t,0) \, \psi} \leq C \norm{\e^{\alpha \lambda^{\nbop[1]}} \liftop{\sqrt{\eps}} V}\times\\
			\left(\norm{\Big(\mathrm{g}_t(\nbop[1]+1)\Big)^{1/2}\, \psi}^2 + c \, \mathrm{g}_t'(c) \int_0^{\abs{t}} \norm{V_2(\tau)} \dd \tau \norm{\psi}^2\right)^{1/2}
		\end{multline}
		for all $\eps \in (0,1/3]$ and $\psi \in \focksalg$, where $\mathrm{g}_t$ is given by \eqref{def_serie_entiere_dependant_du_temps_pour_majoration_cruciale_asymptotique_complete_QFT}.
	\end{cor}
	
	\begin{proof}
	 This is a consequence of the inequality 
	 \[
	 \norm{\Wick[\eps]{F_V} \, U_2(t,0) \, \psi} \leq \norm{\Wick[\eps]{F_V} \mathcal{\bf S}_{\alpha_0}^{-1/2}}\cdot \norm{\mathcal{\bf S}_{\alpha_0}^{1/2} \, U_2(t,0) \, \psi}
	 \] 
	 combined with the estimates \eqref{Wick_Fv_est_borne_sur_domaine_racine_serie_entiere} and \eqref{image_fock_algebrique_propag_quad_inclus_dans_domaine_puissance_serie_entiere}.
	\end{proof}
	
	\subsection{Complete asymptotic expansion}\label{subsect.End_Proof_Thm}
	
	In this subsection, we prove Theorem \ref{thm.general_asymptotique_complete}. Since the proof is very similar to the one of Theorem \ref{thm.Pphi2_asymptotique_complete}, we merely insist upon the main differences.
	
	\bigskip\noindent
	Let $V = \Big(V^{(j)}\Big)_{j \in \N} \in \focks$ such that $\liftop{\conjug}V = V$ and \ref{hyp.dom_superexpo} with constants $\alpha > 0$, $\lambda > 1$.
	Let $\psi \in \focksalg$ and $N \in \N$ be fixed. For $\delta(t)$ and $\big(b_j(t)\big)_{j \in \N}$ respectively given by \eqref{conditions_fonction_torsion_propagation_etats_coherents} and \eqref{relation_de_recurrence_verifiee_les_polynomes_de_Wick_de_l'asymptotique_complete}, we need to differentiate the quantity 
	\begin{equation*}
		\Theta_N(t) = \e^{\rmi\frac{t}{\eps}H_\eps}\, \e^{\rmi\frac{\delta(t)}{\eps}}\, \weylop[\eps]{-\rmi\frac{\sqrt{2}}{\eps}\varphi_t}U_2(t,0) \, \psi_N(t),
	\end{equation*}
	where $H_\eps = \dG[\eps]{A} + \Wick[\eps]{F_V}$, $\varphi_t$ is the mild solution of \eqref{eq_classique_avec_partie_libre}, $U_2(t,0)$ is the quadratic propagator given by Theorem \ref{thm.dynamique_quad_applique_à_notre_potentiel} and 
	\begin{equation*}
		\psi_N(t) = \sum_{j = 0}^N \eps^{\frac{j}{2}}\, \Wick[1]{b_j(t)}\psi\,.
	\end{equation*}
	
	In view of Propositions \ref{prop.op_Weyl_preserve_presque_domaine_racine_serie_entiere} and \ref{propagateur_quad_envoie_Fock_algebrique_dans_domaine_racine_serie_entiere}, the vectors 
	\begin{equation}
		\chi_N(t) = \weylop[\eps]{-\rmi\frac{\sqrt{2}}{\eps}\tilde{\varphi}_t} \tilde{U}_2(t,0) \, \psi_N(t)
	\end{equation} 
	all belong to the space 
	\[
	\displaystyle\bigcap_{0 < \lambda^2 \alpha_0 < 2 \alpha} \mathcal{D}_{\alpha_0}\,,
	\] 
	on which the operator $B_t = \e^{\rmi\frac{t}{\eps}H_\eps}\, \e^{-\rmi\frac{t}{\eps} \dG[\eps]{A}}$ is differentiable by Proposition \ref{prop.differentiabilite_premier_op_unitaire_dans_Theta(t)}. Eventually, we can differentiate the modified expression \eqref{reecriture_Theta(t)} of $\Theta_N(t)$. Using Proposition \ref{prop.quantification_Wick_polynomiale} as well as the identities \eqref{auxiliary_identities_to_differentiate_Theta(t)} and the translation property \eqref{translation_Wick_non_poly}, we obtain 
	\begin{equation}
		\dot{\Theta}_N(t) = \frac{\rmi}{\eps}\, \e^{\rmi\frac{t}{\eps} H} \, \e^{\rmi\frac{\delta(t)}{\eps}}\, \weylop[\eps]{-\rmi\frac{\sqrt{2}}{\eps}\varphi_t} \cdot \left(\Wick[\eps]{f_t} \, U_2(t,0) \, \psi_N(t) - \rmi \, \eps \, U_2(t,0) \dot{\psi}_N(t)\right)
	\end{equation}
	for all $\eps \in (0,1/3]$, where $f_t(z)$ is defined by \eqref{expression_auxiliaire_dans_Wick_pour_derivation_Theta(t)}. Now, using the expansion \eqref{developpement_Fv_autour_d'un_point_ordre_2} of $F_V$ around $z_0 = \varphi_t$ and the definition \eqref{conditions_fonction_torsion_propagation_etats_coherents} of $\delta(t)$, we get 
	\begin{multline}\label{expression_propre_derivee_Theta(t)_general}
	 \dot{\Theta}_N(t) = \frac{\rmi}{\eps}\, \e^{\rmi\frac{t}{\eps}H_\eps}\, \e^{\rmi\frac{\delta(t)}{\eps}}\, \weylop[\eps]{-\rmi\frac{\sqrt{2}}{\eps}\varphi_t} \times\\ \left(\Wick[\eps]{F_{R_3(t)}} \, U_2(t,0) \, \psi_N(t) - \rmi \, \eps \, U_2(t,0) \dot{\psi}_N(t)\right)\,,
	\end{multline}
	where 
	\begin{equation}
		V_\ell^{(j)}(t) = \symop[\ell]\big(\bra{\tensn{(\varphi_t + \conjug \varphi_t)}{j-\ell}} \ot (\tensn{\1}{\ell})\big)V^{(j)} \in \tensymn{\hi}{\ell}\,,
	\end{equation}
	\begin{equation}\label{def_potentiels_dev_Taylor_propagation_QFT_general}
		V_\ell(t) = \sum_{j=\ell}^\infty\, \sqrt{\frac{\ell!}{j!}}\,\binom{j}{\ell}\, V_\ell^{(j)}(t)
	\end{equation} 
	and $R_3(t) = \displaystyle\bigoplus_{\ell=3}^\infty V_\ell(t)$. Furthermore, for all $N \in \N$ and $0 \leq j \leq N$, we have the formula 
	\[\Wick[\eps]{F_{R_3(t)}} = \sum_{\ell = 3}^{N+2-j} \eps^{\frac{\ell}{2}}\, \Wick[1]{F_{V_\ell(t)}} + \Wick[\eps]{F_{R_{N+3-j}(t)}}
	\] 
	where $R_{N+3-j}(t) = \displaystyle\sum_{\ell = N+3-j}^\infty V_\ell(t) \in \displaystyle\bigoplus_{\ell = N+3-j}^\infty \tensymn{\hi}{\ell}$. 
	Using the explicit expression \eqref{ansatz_psiN_derivation_Theta(t)} of $\psi_N(t)$ and Theorem \ref{thm.Breteaux_conjug_Wick_propag_quad}, as well as the definition \eqref{relation_de_recurrence_verifiee_les_polynomes_de_Wick_de_l'asymptotique_complete} of the polynomial symbols $b_k(t)$, we thus obtain 
	\begin{multline}\label{expression_finale_derivee_Theta(t)_general}
	 \dot{\Theta}_N(t) = \frac{\rmi}{\eps}\, \e^{\rmi\frac{t}{\eps}H_\eps}\, \e^{\rmi\frac{\delta(t)}{\eps}}\, \weylop[\eps]{-\rmi\frac{\sqrt{2}}{\eps}\varphi_t}\times\\
	 \left(\sum_{j = 0}^N \eps^{\frac{j}{2}}\,\Wick[\eps]{F_{R_{N+3-j}(t)}}\, U_2(t,0)\, \Wick[1]{b_j(t)} \psi\right)\,.
	\end{multline}
	
	 Let us now conclude the proof of Theorem \ref{thm.general_asymptotique_complete}. The left-hand side of \eqref{asymptotique_complete} is dominated by 
	\begin{equation*}
		\int_0^{\abs{t}} \norm{\dot{\Theta}_N(s)}\, \dd s \leq \frac{1}{\eps} \int_0^{\abs{t}} \sum_{j = 0}^N \eps^{\frac{j}{2}} \, \norm{\Wick[\eps]{F_{R_{N+3-j}(s)}}\, U_2(s,0)\ \Wick[1]{b_j(s)} \psi}\, \dd s\,.
	\end{equation*}
	Since $V$ satisfies \ref{hyp.dom_superexpo} with constants $\alpha > 0$, $\lambda > 1$, the sequence $\Big(V_\ell(t)\Big)_{\ell \in \N}$ satisfies \ref{hyp.dom_superexpo} with constants $0 < \gamma < \alpha$ and $\lambda > 1$. 

\medskip\noindent
	Thus, we can apply \eqref{inegalite_cruciale_pour} to the potentials $R_{N+3-j}(s) \in \displaystyle\bigoplus_{n = N+3-j}^{+\infty} \tensymn{\hi}{n}$ and to the functions $\Wick[1]{b_j(s)}\psi \in \focksalg$, which yields 
	\begin{align*}
		\norm{\Wick[\eps]{F_{R_{N+3-j}(s)}} \, U_2(s,0) \Wick[1]{b_j(s)} \psi}^2 \leq & C\, \eps^{N+3-j} \norm{\e^{\gamma \lambda^{\nbop[1]}} R_{N+3-j}(s)}^2\times\\ &\bigg(\norm{\mathrm{\bf g}_s^{1/2}\, \Wick[1]{b_j(s)} \psi}^2 \\&+ c \,\mathrm{g}_s'(c) \int_0^{\abs{s}} \norm{V_2(\tau)} \dd \tau \norm{\Wick[1]{b_j(s)} \psi}^2\bigg)\\
		= & C'(s,\psi,N,j) \, \eps^{N+3-j} 
	\end{align*}
	with $C'(s,\psi,N,j)$ continuous with respect to $s$ and $\mathrm{g}_s$ given by \eqref{def_serie_entiere_dependant_du_temps_pour_majoration_cruciale_asymptotique_complete_QFT}. Recall that $\mathrm{\bf g}_s$ is defined like in \eqref{fct_analytique_op_nombre}. Eventually, the left-hand side of \eqref{asymptotique_complete} is dominated by \[\frac{1}{\eps} \int_0^{\abs{t}} \sum_{j = 0}^N \eps^{\frac{j}{2}} \, C(s,\psi,N,j) \cdot \eps^{\frac{N+3-j}{2}} \dd s = C(t,\psi,N) \cdot \eps^{\frac{N+1}{2}}\,,
	\]
	which concludes the proof of Theorem \ref{thm.general_asymptotique_complete}. 
	\hfill$\square$
	
	\begin{appendices}
	
	\section{Weyl operators and Number estimates}\label{app.weyl_op_nb_domain}
	
	\noindent The aim of this appendix is to establish the Number-Weyl estimates \eqref{controle_norme_methode_commut}. Such estimates are standard, see for instance \cite[Section 3.5]{DeGe99} for further discussion. However, the proof of Proposition \ref{prop.op_Weyl_preserve_presque_domaine_racine_serie_entiere} requires applying \eqref{controle_norme_methode_commut} with a sharp quantitative bound tailored to our needs. For this reason, we prefer to include a detailed proof. 
	
	\medskip
	Let us denote by 
	\begin{equation*}\label{def_right_complex_half_plane}
	 \C_+ = \{z \in \C, \Re(z) \geq 0\}
	\end{equation*}
	the closed right half-plane. For all $z \in \C_+$, let $D\Big((\nbop[1]+1)^z\Big)$ be the domain of the closure of the operator $\Big((\nbop[1]+1)^z,\focksalg\Big)$.
	
	\smallskip

	\begin{theorem}\label{thm.op_de_Weyl_preserve_domaines_puissances_nombre}
	 Let $z \in \C_+$, $u \in \hi\, \setminus\, \{0\}$ and $0 < \eps \leq 1$ be fixed. Then, the Weyl operator $\weylop[\eps]{u}$ preserves the domain $D\Big((\nbop[1]+1)^z\Big)$. Furthermore, we have an estimate of the form 
	 \begin{equation}\label{controle_norme_methode_commut}
	 \norm{(\nbop[1]+1)^z\, \weylop[\eps]{u}\, (\nbop[1]+1)^{-z}} \leq M(\normhi{u})\, \normhi{u}^{\Re(z)}\, 2^{\Re(z)^2}\,,
	 \end{equation}
	 where $M(\normhi{u})$ is an $\eps$-independent constant.
	\end{theorem}

	 \medskip
		
	\noindent We state a technical lemma that provides the estimate \eqref{controle_norme_methode_commut} for $z \in \N$. Then, we deduce the general case thanks to a complex interpolation argument.
	
	\subsection*{Proof for integers}
	
	\begin{lemma}\label{lemma.op_Tk_borne}
	 For all $u \in \hi\, \setminus\, \{0\}$, $k \in \N$ and $0 < \eps \leq 1$, the operator 
	 \begin{equation}\label{def_opérateur_Tk}
	 T_k = \left(\nbop[1] + 1 + \fieldop[\eps]{\rmi u} + \frac{\eps}{2}\, \normhi{u}^2\right)^k\, (\nbop[1] + 1)^{-k}
	 \end{equation}
	 is bounded on $\focks$. Furthermore, we have an estimate of the form 
	 \begin{equation}\label{premier_controle_norme_methode_commut}
	 \norm{T_k} \leq M'(\normhi{u})\, \normhi{u}^k\, 2^{k^2}\,,
	 \end{equation}
	 where $M'(\normhi{u})$ is an $\eps$-independent constant.
	 \end{lemma}
	 
	 \medskip
	 
	 Before jumping into the proof of this result, let us explain how to obtain Theorem \ref{thm.op_de_Weyl_preserve_domaines_puissances_nombre} in the case $z = k \in \N$. We recall that the space $\wickbigcore$ defined by \eqref{def_bigcore_Wick_poly} is a core for the Wick operators $\Wick[\eps]{b},\, b \in \wickpol$. In particular, we have $\wickbigcore \subset \fieldop[\eps]{u}$ and $\wickbigcore \subset D(\nbop[\eps]^k)$ for all $u \in \hi$, $k \in \N$. The translation property \eqref{translation_Wick_poly} yields 
	 \begin{equation}
	 \weylop[\eps]{u}^{*} \fieldop[\eps]{v} \weylop[\eps]{u} = \fieldop[\eps]{v} + \eps\, \Im\braket{u}{v} \1
	 \end{equation}
	 and 
	 \begin{equation}
	 \weylop[\eps]{u}^{*} \nbop[\eps] \weylop[\eps]{u} = \nbop[\eps] + \eps\, \fieldop[\eps]{\rmi u} + \frac{\eps^2}{2}\, \normhi{u}^2 \1
	 \end{equation}
	 on $\wickbigcore$. Hence, the operators $\weylop[\eps]{u}^{*} \fieldop[\eps]{v} \weylop[\eps]{u}$ and $\weylop[\eps]{u}^{*} \nbop[\eps] \weylop[\eps]{u}$ preserve $\focksalg$, so $\fieldop[\eps]{v}$ and $\nbop[1]$ preserve $\wickbigcore$. This enables us to write, for all $k \in \N$, 
	 \begin{equation}
	 \weylop[\eps]{u}^{*}\, (\nbop[1] + 1)^k\, \weylop[\eps]{u} = T_k\, (\nbop[1]+1)^k \text{ on $\wickbigcore$}
	 \end{equation}
	 with $T_k$ given by \eqref{def_opérateur_Tk}. Hence, for all $\psi \in \wickbigcore$ and $\varphi \in D(\nbop[\eps]^k)$ such that $\norm{\varphi} = 1$, we have 
	 \begin{equation*}
	 \abs{\braket{\weylop[\eps]{u}^{*}\, (\nbop[1] + 1)^k\, \varphi}{\psi}} \leq \norm{T_k}\, \norm{(\nbop[1] + 1)^k \psi}\,.
	 \end{equation*}
	 Extending this inequality to $\psi \in D(\nbop[\eps]^k)$ yields, for all $\chi \in \focks$, 
	 \begin{equation*}
	 \abs{\braket{(\nbop[1] + 1)^k \varphi}{\weylop[\eps]{u} (\nbop[1] + 1)^{-k}\, \chi}} \leq \norm{T_k} \norm{\chi}\,.
	 \end{equation*}
	 Finally, we obtain that $\weylop[\eps]{u}\, (\nbop[1] + 1)^{-k}\, \chi \in D\Big([(\nbop[1] + 1)^k]^{*}\Big) = D\Big((\nbop[1] + 1)^k\Big)$, and the estimate \eqref{controle_norme_methode_commut} for $z = k \in \N$ follows from \eqref{premier_controle_norme_methode_commut}.
	
	\bigskip
	
	We now turn to the proof of Lemma \ref{lemma.op_Tk_borne}, which is inspired from the number-energy estimates \cite[Section 3.5]{DeGe99} (see also \cite[Proposition 3.5]{AmFaOl23} for an application of such estimates in another context). Let us define \[A = \nbop[1] + 1 + \fieldop[\eps]{\rmi u} + \frac{\eps}{2}\, \normhi{u}^2 \text{ and } B = (\nbop[1] + 1)^{-1}\] so that $T_k = A^k B^k$. We define recursively 
			\begin{equation}
				\left\{\begin{array}{l}
					\ad[B]^0(A) = A,\\
					\forall\, j \in \N, \quad \ad[B]^{j+1}(A) = \comm{\ad[B]^j(A)}{B}\,.
				\end{array}\right.
			\end{equation}
			
			From the formulas $\comm{\fieldop[\eps]{v}}{\nbop[\eps]} = \rmi \, \eps\, \fieldop[\eps]{\rmi v}$ and \[\comm{A}{B} = B \comm{B^{-1}}{A} B\,,\] we deduce that the equality \begin{equation}\label{expression_commut_methode_des_commutateurs}
				\ad[B](A) = \comm{A}{B} = -\rmi\, B\, \fieldop[\eps]{u} B
			\end{equation}
			holds on $\wickbigcore$ and extends to $\mathcal{L}(\focks)$. This enables us to obtain by induction on $j \in \N^*$ the formula \begin{equation}\label{expression_commut_multiple_methode_des_commutateurs}
				\ad[B]^j(A) = (-\rmi)^j\, B^j\, \fieldop[\eps]{u_j} B^j \text{ where } u_1 = u \text{ and } u_{j+1} = \rmi\, u_j\,.
			\end{equation}
		We now show that the operators $T_k,\, k \in \N^*$ are bounded. For $k = 1$, the inequality \begin{equation}\label{inégalité_norme_champ_et_nombre}
		 \norm{\fieldop[\eps]{u} \psi}^2 \leq 2\, \eps\, \normhi{u}^2 \norm{(\nbop[1]+1)^{1/2}\, \psi}^2
		\end{equation} yields $\norm{T_1} \leq Q(\normhi{u})^{1/2}$ where $Q$ is an $\eps$-independent polynomial provided that $0 < \eps \leq 1$. Furthermore, the Leibniz formula 
		\begin{equation}
				A B^k = \sum_{j = 0}^k \binom{k}{j}\, B^{k-j}\, \ad[B]^j(A)
			\end{equation}
			and the formula \eqref{expression_commut_multiple_methode_des_commutateurs} allow us to write \begin{align*}
				T_{k+1} = A^{k+1} B^{k+1} = &A^k (A B^k) B\\
				= &A^k B^k A B + \sum_{j=1}^k \binom{k}{j} A^k B^{k-j}\ \ad[B]^j(A)\, B\\
				= &T_1 T_k + \sum_{j=1}^k\, (-\rmi)^j\, \binom{k}{j}\, T_k\, \fieldop[\eps]{u_j}\, B^{j+1}\,.
			\end{align*} 
			Hence, provided that $T_k$ is bounded, we obtain that $T_{k+1}$ is bounded with 
			\begin{align*}
			 \norm{T_{k+1}} \leq &\norm{T_k} \big(\norm{T_1} + \sum_{j=1}^k \binom{k}{j}\ \norm{\fieldop[\eps]{u_j}(\nbop[1] + 1)^{-1}}\big)\\
			 \leq &\norm{T_k}\big(Q(\normhi{u})^{1/2} + 2^k \cdot \sqrt{2} \, \normhi{u}\big) \text{ by \eqref{inégalité_norme_champ_et_nombre}}\,.
			\end{align*}
			From the last inequality, we deduce \eqref{premier_controle_norme_methode_commut}. Indeed, if we define 
			\begin{equation*}
			 a = Q(\normhi{u})^{1/2} \geq 0,\quad b = \sqrt{2}\, \normhi{u} > 0\,,
			\end{equation*}
			we have for $k \in \N^*$
			\begin{align*}
			 \log\big(\norm{T_k}\big) - \log\big(\norm{T_1}\big) \leq &\sum_{j=1}^{k-1}\, \log(a + 2^j\, b)\\
			 = &\sum_{j=1}^{k-1}\, \left[j\,\log(2) + \log(b) + \log\left(1 + 2^{-j}\, \frac{a}{b}\right)\right]\\
			 \leq &\,\frac{k^2}{2}\, \log(2) + k\, \log(b) + \frac{a}{b}\,,
			\end{align*}
			where we have used $\log(1+x) \leq x$. We finally obtain 
			\begin{equation*}
			 \norm{T_k} \leq M'(\abs{u})\, 2^{k^2}\, \normhi{u}^k\,,
			\end{equation*}
			whence \eqref{premier_controle_norme_methode_commut} for all $k \in \N^*$. This estimate extends to the case $k = 0$ up to picking $M'(\abs{u}) \geq 1$.
		
	\subsection*{Proof in the general case}
	
	Let $u \in \hi\, \setminus\, \{0\}$ and $k \in \N$ be fixed. For all $\varphi, \psi \in \focksalg$ with $\norm{\varphi} = 1$, the function \[F_{\varphi,\psi}(w) = \braket{\varphi}{(\nbop[1]+1)^w\, \weylop[\eps]{u}\, (\nbop[1]+1)^{-w}\, \psi}\] can be expressed as a finite linear combination of terms of the form 
	\[(n+1)^w\, (m+1)^{-w} \braket{\varphi^{(n)}}{\weylop[\eps]{u} \psi^{(m)}}, \quad m,n \in \N\,.\] 
	Hence, $F_{\varphi,\psi}$ is entire on $\C$. Applying the three-line Hadamard interpolation theorem \cite[Theorem~12.3, page~187]{Si11} and taking the supremum over $\varphi$, we obtain 
	\begin{equation*}
	 \norm{(\nbop[1]+1)^{k+w}\, \weylop[\eps]{u}\, (\nbop[1]+1)^{-(k+w)}} \leq \norm{T_k}^{1-\Re(w)} \norm{T_{k+1}}^{\Re(w)}
	\end{equation*}
	for all $w \in \C$ such that $0 \leq \Re(w) \leq 1$. This yields, for all $z \in \C$ with $\Re(z) \geq 0$, 
	\begin{equation*}
	 \norm{(\nbop[1]+1)^z\, \weylop[\eps]{u}\, (\nbop[1]+1)^{-z}} \leq \norm{T_k}^{1-x} \norm{T_{\ell+1}}^x \,,
	\end{equation*}
	where $k = \ent{\Re(z)}$ and $k + x = \Re(z)$. Applying \eqref{premier_controle_norme_methode_commut}, we finally obtain 
	\begin{align*}
	 \norm{(\nbop[1]+1)^z\, \weylop[\eps]{u}\, (\nbop[1]+1)^{-z}} \leq &M'(\normhi{u})\, \normhi{u}^{k(1-x)+(k+1)x}\, 2^{k^2(1-x)+(k+1)^2 x}\\
	 = &M'(\normhi{u})\, \normhi{u}^{\Re(z)}\, 2^{\Re(z)^2 + (x-x^2)}\\
	 \leq &M(\normhi{u})\, \normhi{u}^{\Re(z)}\, 2^{\Re(z)^2} \text{ where } M(\normhi{u}) = 2^{\frac{1}{4}}\, M'(\normhi{u})\,.
	\end{align*}
	This concludes the proof of Theorem \ref{thm.op_de_Weyl_preserve_domaines_puissances_nombre}.
	\end{appendices}

\medskip\noindent 
{\small\textbf{Acknowledgments: }{\scriptsize The first author
acknowledges the financial support of the ANR-DFG project ANR-22-CE92-0013. Some of the ideas developed in this work were discussed during the visit of the second and third authors to IRMAR, Université de Rennes 1, in January 2025, and we gratefully acknowledge their kind hospitality.}

\printbibliography

\end{document}